\documentclass[english,10pt,a4paper]{article}
\usepackage[margin=1in]{geometry}

\newcommand{\tbl}[2]{\begin{center} #2 \end{center} \caption{#1}}
\newcommand{\url}[1]{\texttt{#1}}
\newcommand{\ack}[1]{\textbf{Acknowledgment:} #1}

\usepackage[utf8]{inputenc}
\usepackage{xcolor}
\usepackage{array}
\usepackage{amsthm}

\newtheorem{lemma}{Lemma}
\newtheorem{theorem}{Theorem}

\usepackage{verbatim}
\usepackage{rotating}
\usepackage{rotfloat}
\usepackage{wrapfig}
\usepackage{units}
\usepackage{textcomp}
\usepackage{multirow}
\usepackage{amsmath}
\usepackage{amssymb}
\usepackage{graphicx}
\usepackage[USenglish]{babel} %
\usepackage{subfigure}
\usepackage{booktabs}
\usepackage{longtable}

\global\long\def\dist{\mathrm{dist}}
\global\long\def\searchspace{\mathrm{SS}}
\global\long\def\mynew{w^{\mathrm{new}}}
\global\long\def\myold{w^{\mathrm{old}}}
\global\long\def\distI{\dist_{\mathrm{I}}}
\global\long\def\distA{\dist_{\mathrm{A}}}
\global\long\def\distUD{\dist_{\mathrm{UD}}}

\newcommand{\lyxdot}{.}

\newcommand{\musec}{\textmu{}s}
\newcommand{\ie}{i.\,e.}
\newcommand{\eg}{e.\,g.}
\newcommand{\cf}{cf.}
\newcommand{\wrt}{wrt.}
\newcommand{\enabled}{\ensuremath{\bullet}}
\newcommand{\disabled}{\ensuremath{\circ}}

\title{Customizable Contraction Hierarchies\thanks{Partial support by DFG grant WA654/16-2 and EU grant 288094 (eCOMPASS)
and Google Focused Research Award.}}
\author{
Julian Dibbelt, Ben Strasser and Dorothea Wagner \\ Karlsruhe Institute of Technology
}

\begin{document}

\maketitle

\begin{abstract}
We consider the problem of quickly computing shortest paths in weighted graphs. Often, this is achieved in two phases: 1) derive auxiliary data in an expensive preprocessing phase, 2) use this auxiliary data to speedup the query phase. By adding a fast weight-customization phase, we extend Contraction Hierarchies to support a three-phase workflow:
The expensive preprocessing is split into a phase exploiting solely the unweighted topology of the graph, as well as a lightweight phase that adapts the auxiliary data to a specific weight.
We achieve this by basing our Customizable Contraction Hierarchies on nested dissection orders. We provide an in-depth experimental analysis on large road and game maps that shows that Customizable Contraction Hierarchies are a very practicable solution in scenarios where edge weights often change. 
\end{abstract}

\section{Introduction}

Computing optimal routes in road networks has many applications such as navigation, logistics, traffic simulation or web-based route planning. 
Road networks are commonly formalized as weighted graphs and the optimal route is formalized as the shortest path in this graph. 
Unfortunately, road graphs tend to be huge in practice with vertex counts in the tens of millions, rendering Dijkstra's algorithm~\cite{d-ntpcg-59} impracticable for interactive use:
It incurs running times in the order of seconds even for a single path query. 
For practical performance on large road networks, preprocessing techniques that augment the network with auxiliary data in a (possibly expensive) offline phase have proven useful. 
See~\cite{bdgmpsww-rptn-14} for an overview. 
Many techniques work by adding extra edges called \emph{shortcuts} to the graph that allow query algorithms to bypass large regions of the graph efficiently.
While variants of the optimal shortcut selection problem have been proven to be NP-hard~\cite{bddsw-tspca-12}, determining good shortcuts is feasible in practice even on large road graphs.
Among the most successful speedup techniques using this building block are Contraction Hierarchies~(CH) by~\cite{gssd-chfsh-08,gssv-erlrn-12}.
At its core the technique consists of a systematic way of adding shortcuts by iteratively contracting vertices along a given order.
Even though ordering heuristics exist that work well in practice~\cite{gssv-erlrn-12}, the problem of computing an optimal ordering is NP-hard in general~\cite{bckkw-psuth-10}.
Worst-case bounds have been proven in~\cite{adfgw-h-13} in terms of a weight-dependent graph measure called highway dimension and \cite{m-o-12} have shown that many of these bounds are tight on many graph classes.

A central restriction of CHs as proposed by~\cite{gssv-erlrn-12} is that their preprocessing is \emph{metric-dependent}, that is edge weights, also called \emph{metric}, need to be known. 
Substantial changes to the metric, \eg, due to user preferences or traffic congestion, may require expensive recomputation. 
For this reason, a Customizable Route Planning~(CRP) approach was proposed in~\cite{dgpw-crp-11}, extending the multi-level-overlay MLD techniques of~\cite{sww-daola-00,hsw-emlog-08}. 
It works in three phases: In a first, expensive phase, auxiliary data is computed that solely exploits the topological structure of the network, disregarding its metric. 
In a second, much less expensive phase, this auxiliary data is \emph{customized} to a specific metric, enabling fast queries in the third phase. 
In this work we extend CH to support such a three-phase approach. 

\paragraph{Game Scenario}

Most existing CH papers focus solely on road graphs, with~\cite{s-chgg-13}
being a notable exception, but there are many other applications with
differently structured graphs in which fast shortest path computations
are important. One of such applications is games. Think of a real-time
strategy game where units quickly have to navigate across a large
map with many choke points. The basic topology of the map is fixed,
however, when buildings are constructed or destroyed, fields
are rendered impassable or freed up. Furthermore, every player has his own knowledge of the map.
Many games include a feature called \emph{fog of war}: 
Initially only the fields around the player's starting location are revealed.
As his units explore the map, new fields are revealed.
Since a~unit must not route around a building that the player has not yet seen, every player needs his own metric.
Furthermore,
units such as hovercrafts may traverse water and land, while other
units are bound to land. This
results in vastly different, evolving metrics for different unit types per player, making metric-dependent
preprocessing difficult to apply.  Contrary
to road graphs one-way streets tend to be extremely rare, and thus
being able to exploit the symmetry of the underlying graph is a useful
feature.

\paragraph{Metric-Independent Orders for CHs}
One of the central building blocks of this paper is the use of metric-independent \emph{nested dissection orders}~(ND-orders)~\cite{g-ndrfe-73} for CH precomputation instead of the metric-dependent order of~\cite{gssv-erlrn-12}. 
This approach was proposed by~\cite{bcrw-sssch-13}, and a preliminary case study can be found in~\cite{z-wchw-13}. 
A similar idea was followed by~\cite{dw-fcrn-13}, where the authors employ partial CHs to engineer subroutines of their customization phase. 
They also refer to preliminary experiments on full CH but did not publish results. 
Similar ideas have also appeared in~\cite{pwk-caspl-12}: They consider %
graphs of low \emph{treewidth}~(see below) and leverage this property to compute CH-like structures, without explicitly using the term CH. 
Related techniques by~\cite{w-tedi-10,cz-spdst-00} work directly on the tree decomposition.
Interestingly,
our experiments show that even large road networks have relatively low treewidth: Real-world road networks with vertex counts in the~$10^{7}$ have treewidth in the~$10^{2}$.

\paragraph{Tree-Decompositions, Sparse Matrices and Minimum Fill-In}

Customizable speedup techniques for shortest path queries are a very recent development but the idea to order vertices along nested dissection orders is significantly older.
To the best of our knowledge the idea first appeared in 1973 in~\cite{g-ndrfe-73} and was refined in~\cite{lrt-gnd-79}.
They use nested dissection orders to reorder the columns and rows of \emph{sparse matrices} to assure that Gaussian elimination preserves as many zeros as possible.
From the matrix they derive a graph and show that vertex contraction in this graph corresponds to Gaussian variable elimination.
Inserting an extra edge in the graph destroys a zero in the matrix.
The additional edges are called the \emph{fill-in}.
The \emph{minimum fill-in problem} asks for a vertex order that results in a minimum number of extra arcs.
In CH terminology these extra edges are called shortcuts.
The super graph constructed by adding the additional edges is a \emph{chordal graph}.
The \emph{treewidth} of a graph~$G$ can be defined using chordal super graphs:
For every super graph consider the number of vertices in the maximum clique minus one.
The treewidth of a graph~$G$ is the minimum of this number over all chordal super graphs of~$G$.
This establishes a relation between sparse matrices and treewidth and in consequence with CHs.
We refer to~\cite{b-tsa-07} and~\cite{b-atgt-93} for an introduction to the broad field of treewidth and tree decompositions.

\begin{wrapfigure}{o}{4.3cm}%
\vspace{-1.1em}
\begin{centering}
\vspace{-0.3cm}\subfigure[No shortcuts, maximum search space is four arcs]{\includegraphics[scale=1.5,page=1]{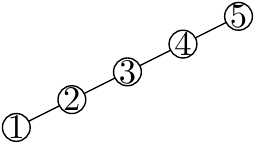}

}
\par\end{centering}

\begin{centering}
\vspace{-0.3cm}\subfigure[Two shortcuts, maximum search space is two arcs]{\includegraphics[scale=1.5,page=2]{path}

}
\par\end{centering}

\caption{\label{fig:path}Contraction Hierarchies for a path graph.}
\end{wrapfigure}%
Minimizing the number of extra edges, \ie, minimizing the fill-in, is NP-hard~\cite{y-cnp-81} but fixed parameter tractable in the number of extra edges~\cite{kst-tpcpc-99}.
Note, however, that from the CH point of view, optimizing the number of extra edges, \ie, the number of shortcuts, is not the sole optimization criterion.
Consider for example a path graph as depicted in Figure~\ref{fig:path}:
Optimizing the CH search space and the number of shortcuts are competing criteria. 
A tree relevant in the theory of treewidth is the \emph{elimination tree}.
\cite{bcrw-sssch-13} have shown that the maximum search space size in terms of vertices corresponds to the height of this elimination tree. 
Unfortunately, minimizing the elimination tree height is also NP-hard~\cite{p-t-88}.
For planar graphs, it has been shown that the number of additional edges is in~$O(n\log n)$~\cite{gt-t-86}.
However, this does not imply a~$O(\log n)$ search space bound in terms of vertices as search spaces can share vertices.

\paragraph{Directed and Undirected Graphs}
Real-world road networks contain one-way streets and highways. 
Such networks are thus usually modeled as directed graphs.
Our algorithms fully support direction of traffic---however, we introduce it at a different stage of the toolchain than most related techniques, which should not be confused with only supporting undirected networks.
Our first preprocessing phase works exclusively on the underlying undirected and unweighted graph, obtained by dropping all edge directions and edge weights.
Direction of traffic as well as traversal weights are only introduced in the second~(customization) phase, where every edge can have two weights: an upward and a downward weight. 
If an edge corresponds to an one-way street, then one of these weights is set to~$\infty$.
Note that this setup is a strength of our algorithm:  
Throughout large parts of the toolchain we are not confronted with additional complexity induced by directed edges.
This contrasts with many other techniques, where considering edge direction induces changes in nearly every step of the algorithm.

\paragraph{Our Contribution}

The main contribution of our work is to show that Customizable Contraction
Hierarchies (CCH) solely based on the ND-principle are feasible and
practical. Compared to CRP~\cite{dgpw-crp-11} we achieve a similar
preprocessing--query tradeoff, albeit with slightly better query performance
at slightly slower customization speed and we need somewhat more space.
Interestingly, for less well-behaved metrics such as travel distance,
we achieve query times below the original metric-dependent CH of~\cite{gssv-erlrn-12}.
Besides this main result, %
we show that given a fixed contraction order, a metric-independent CH can be constructed in time essentially linear in the size of the Contraction Hierarchy with working memory consumption linear in the input graph. 
Our specialized algorithm has better theoretic worst-case running time and performs significantly better empirically than the dynamic adjacency arrays used in~\cite{gssv-erlrn-12}.
Another contribution of our work are perfect witness searches. 
We show that for a fixed metric-independent vertex order it is possible to construct CHs with a provably minimum number of arcs in a few seconds on continental road graphs. 
Our construction algorithm has a running time performance completely independent of the weights used.
We further show that an order based on nested dissection gives a constant factor approximation of the maximum and average search space sizes in terms of the number of arcs and vertices for metric-independent CHs on a class of graphs with very regular recursive vertex separators.
Experimentally, we show that road graphs have such a recursive separator structure. 

\paragraph{Outline}
Section~\ref{sec:basics} sets necessary notation. 
Section~\ref{sec:trad-ch-order} discusses metric-dependent orders as used by~\cite{gssv-erlrn-12}, highlighting specifics of our implementation. 
Next, we discuss metric-independent orders in Section~\ref{sec:ch-order}.
In Section~\ref{sec:ch-construction}, we describe how to efficiently construct the arcs of the CH.
The next Section~\ref{sec:enum-triangles} discusses how to efficiently enumerate triangles in the CH --- an operation needed throughout the customization process detailed in Section~\ref{sec:customization}.
In Section~\ref{sec:customization} we further describe the details of the perfect witness search.
Finally, Section~\ref{sec:query} concludes the algorithm description by introducing the algorithms used in the query phase to compute shortest path distances and compute the corresponding paths in the input graph.
We then present an extensive experimental study that thoroughly evaluates the proposed algorithm.
We finish with the conclusion, and with directions for future work.

\section{Basics}
\label{sec:basics}

We denote by~$G=(V,E)$ an \emph{undirected} $n$-vertex \emph{graph}
where $V$ is the set of \emph{vertices} and $E$ the set of \emph{edges}.
Furthermore, $G=(V,A)$ denotes a \emph{directed graph}, where $A$
is the set of \emph{arcs}. A graph is \emph{simple} if it has no loops
or multi-edges. Graphs in this paper are simple unless noted
otherwise, \eg, in parts of Section~\ref{sec:ch-construction}. 
Furthermore, we assume that input graphs are strongly connected. 
We denote by~$N(v)$ the neighborhood of vertex~$v \in G$, \ie, the set of vertices adjacent to~$v$; for directed graphs the neighborhood ignores arc direction.
A \emph{vertex separator} is a vertex subset~$S\subseteq V$ whose
removal separates $G$ into two disconnected subgraphs induced by
the vertex sets~$A$ and~$B$. The sets $S$, $A$ and $B$ are disjoint and their union forms~$V$. 
Note that the subgraphs induced by~$A$ and $B$ are not necessarily connected and may be empty.
A separator~$S$ is \emph{balanced}
if $\left|A\right|,\left|B\right|\le2n/3$. 

A \emph{vertex order}~$\pi:\{1\ldots n\}\rightarrow V$ is a bijection.
Its inverse~$\pi^{-1}$
assigns each vertex a \emph{rank}. Every
undirected graph can be transformed into a \emph{upward directed graph} with respect to a vertex order~$\pi$, \ie, every edge~$\{\pi(i),\pi(j)\}$ with~$i<j$ is replaced
by an arc~$(\pi(i),\pi(j))$. Note that all upward directed graphs
are acyclic. We denote by~$N_{u}(v)$ the upward neighborhood of $v$, \ie, the neighbors of~$v$ with a
higher rank than~$v$, and by~$N_{d}(v)$ the downward neighborhood of $v$, \ie, the vertices with a lower rank than
$v$. We denote by~$d_{u}(v)=\left|N_{u}(v)\right|$ the \emph{upward degree} and by~$d_{d}(v)=\left|N_{d}(v)\right|$ the \emph{downward degree} of a vertex.

\emph{Undirected edge weights} are denoted using~$w:E\rightarrow\mathbb{R}_{>0}$.
With respect to a vertex order~$\pi$ we define an \emph{upward weight}
$w_{u}:E\rightarrow\mathbb{R}_{>0}$ and a \emph{downward weight}~$w_{d}:E\rightarrow\mathbb{R}_{>0}$.
For directed graphs, one-way streets are modeled by setting~$w_{u}$ or~$w_{d}$ to~$\infty$.

A path~$P$ is a sequence of adjacent vertices and incident edges. Its \emph{hop-length} is
the number of edges in~$P$. Its \emph{weight-length} with respect to~$w$ is the
sum over all edges' weights. Unless noted otherwise, length always refers to weight-length in this paper. 
A shortest $st$-path is a path of minimum
length between vertices $s$ and $t$. The minimum length in~$G$
between two vertices is denoted by~$\dist_{G}(s,t)$. We set $\dist_{G}(s,t)=\infty$
if no path exists. An \emph{up-down path}~$P$ with respect to~$\pi$
is a path that can be split into an upward path~$P_{u}$ and a downward
path~$P_{d}$. The vertices in the upward path~$P_{u}$ must occur
by increasing rank~$\pi^{-1}$ and the vertices in the downward path~$P_{d}$
must occur by decreasing rank~$\pi^{-1}$. 
The upward and downward paths meet at the vertex with the maximum rank on the path.
We refer to this vertex as the \emph{meeting vertex}. 

The vertices of every acyclic directed graph~(DAG) can be partitioned into
\emph{levels}~$\ell:V\rightarrow\mathbb{N}$ such that for every arc
$(x,y)$ it holds that $\ell(x)<\ell(y)$. We only consider levels
such that each vertex has the lowest possible level. Note that such
levels can be computed in linear time given a directed acyclic graph.

The unweighted \emph{vertex contraction} of~$v$ in~$G$ consists
of removing~$v$ and all incident edges and inserting edges between
all neighbors~$N(v)$ if not already present. The inserted edges are
referred to as \emph{shortcuts} and the other edges are \emph{original
edges}. Given an order~$\pi$ the \emph{core graph}~$G_{\pi,i}$ is
obtained by contracting all vertices~$\pi(1)\ldots\pi(i-1)$ in order of their rank. We call the
original graph~$G$ augmented by the set of shortcuts a \emph{contraction hierarchy} $G_{\pi}^{*}=\bigcup_i G_{\pi,i}$.
Furthermore, we denote by~$G_{\pi}^{\wedge}$ the corresponding upward directed graph.

\begin{wrapfigure}{o}{3.5cm}%
\begin{center}
\includegraphics[scale=1.5]{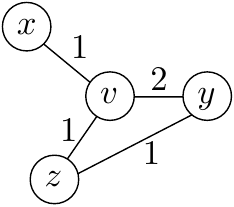}
\end{center}

\caption{Contraction of~$v$.
If the pair~$x,y$ is considered first, a shortcut~$\{x,y\}$ with weight 3 is inserted. 
If the pair~$x,z$ is considered first, an edge~$\{x,z\}$ with weight 2 is inserted. 
This shortcut is part of a witness $x\rightarrow z\rightarrow y$ for the pair~$x,y$. The shortcut~$\{x,y\}$ is thus \emph{not} added if the pair~$x,z$ is considered first.}
\label{fig:shortcut-increasing}
\end{wrapfigure}

Given a fixed weight~$w$, one can exploit that in many applications it is sufficient to only
preserve all shortest path distances~\cite{gssv-erlrn-12}. \emph{Weighted vertex contraction}
of a vertex~$v$ in the graph~$G$ is defined as the operation of removing~$v$ and inserting (a minimum number of shortcuts) among the neighbors of~$v$ to obtain a graph~$G'$ such that $\dist_G(x,y) = \dist_{G'}(x,y)$ 
for all vertices~$x\neq v$ and~$y \neq v$. To compute~$G'$, one iterates over all pairs of neighbors~$x,y$ of~$v$ increasing by~$\dist_G(x,y)$.
For each pair one checks whether a~$xy$-path of length~$\dist_{G}(x,y)$ exists in~$G \backslash \{v\}$, \ie, one checks whether removing~$v$ destroys the~$xy$-shortest path. 
This check is called \emph{witness search}~\cite{gssv-erlrn-12} and the~$xy$-path is called \emph{witness}, if it exists.
If a witness is found, the considered vertex pair is skipped and no shortcut added. 
Otherwise, if an edge~$\{x,y\}$ already exists, its weight is decreased to~$\dist_G(x,y)$, or a new shortcut edge with that weight is added to~$G$.
This new shortcut edge is considered in witness searches for subsequent neighbor pairs as part of~$G$.
If shortest paths are not unique, it is important to iterate over the pairs increasing by~$\dist_G(x,y)$, because otherwise more edges than strictly necessary can be inserted:  
Shorter shortcuts can make longer shortcuts superfluous. 
However, if we insert the shorter shortcut after the longer ones, the witness search will not consider them. 
See Figure~\ref{fig:shortcut-increasing} for an example. 
Note that the witness searches are expensive and therefore the witness search is usually aborted after a certain number of steps~\cite{gssv-erlrn-12}. 
If no witness was found until, we assume that none exists and add a shortcut.
This does not affect the correctness of the technique but might result in slightly more shortcuts than necessary.
To distinguish, \emph{perfect witness search} is without such an one-sided error.

For an order~$\pi$ and a weight~$w$ the \emph{weighted core graph}~$G_{w,\pi,i}$ is obtained
by contracting all vertices~$\pi(1)\ldots\pi(i-1)$.
The original graph~$G$ augmented by the set of weighted shortcuts is called a \emph{weighted contraction hierarchy}~$G_{w,\pi}^{*}$. 
The corresponding upward directed graph is denoted by~$G_{w,\pi}^{\wedge}$.

The search space~$\searchspace(v)$ of a vertex~$v$ is the subgraph
of~$G_{\pi}^{\wedge}$ (respectively~$G_{w,\pi}^{\wedge}$) reachable
from~$v$. For every vertex pair~$s$ and~$t$, it has been shown
that a shortest up-down path must exist. This up-down path can be
found by running a bidirectional search from~$s$ restricted
to~$\searchspace(s)$ and from~$t$ restricted to~$\searchspace(t)$
\cite{gssv-erlrn-12}.
A graph is \emph{chordal} if for every cycle of at least four vertices there exists a pair of vertices that are non-adjacent in the cycle but are connected by an edge.
An alternative characterization is that a vertex order~$\pi$ exists such that for every~$i$ the neighbors of~$\pi(i)$ in~$G_{\pi,i}$, \ie, the core graph before the contraction of $\pi(i)$, form a clique \cite{fg-imig-65}. 
Such an order is called a \emph{perfect elimination order}. 
Another way to formulate this characterization in CH terminology is as follows: A graph is chordal if and only if a contraction order exists such that the CH construction without witness search does not insert any shortcuts.
A chordal super graph can be obtained by adding the CH shortcuts.

The elimination tree~$T_{G,\pi}$ is a tree directed towards its root
$\pi(n)$. The parent of vertex~$\pi(i)$ is its upward neighbor~$v\in N_{u}(\pi(i))$
of minimal rank~$\pi^{-1}(v)$. Note that this definition already
yields a straightforward algorithm for constructing the elimination tree. As shown in \cite{bcrw-sssch-13},
the set of vertices on the path from~$v$ to~$\pi(n)$ is the set
of vertices in~$\searchspace(v)$. Computing a contraction hierarchy
without witness search of graph~$G$ consists of computing a chordal
super graph~$G_{\pi}^{*}$ with perfect elimination order~$\pi$. The
height of the elimination tree corresponds to the maximum number of vertices
in the search space. Note that the elimination tree is only defined
for undirected unweighted graphs.

\section{Metric-Dependent Orders}
\label{sec:trad-ch-order}

Most publications on applications and extensions of Contraction Hierarchies use greedy orders in the spirit of~\cite{gssv-erlrn-12}, but details of vertex order computation and witness search vary. For reproducibility, we describe our precise approach in this section, extending on the general description of metric-dependent CH preprocessing given in Section~\ref{sec:basics}.
Our witness search aborts once it finds some path shorter than the shortcut---or
when both forward and backward search each have settled at most $p$
vertices. For most experiments we choose $p=50$. The only exception is the distance metric on road
graphs, where we set $p=1500$. 
We found that a higher value of $p$ increases the time per witness-search but leads to sparser cores.
For the distance metric we needed a high value because otherwise our cores get too dense.
This effect did not occur for the other weights considered in the experiments.
Our weighting heuristic is similar to the one of \cite{adgw-hhlsp-12}. 
We denote by $L(x)$ a value that approximates the level of vertex $x$.
Initially all $L(x)$ are $0$. 
If $x$ is contracted, then for every incident edge $\{x,y\}$ we perform $\ell(y)\leftarrow\max\{\ell(y),\ell(x)+1\}$.
We further store for every arc $a$ a hop length $h(a)$. This is
the number of arcs that the shortcut represents if fully unpacked.
Denote by $D(x)$ the set of arcs removed if $x$ is contracted and
by $A(x)$ the set of arcs that are inserted. Note that $A(x)$ is
not necessarily a full clique because of the witness search and because
some edges may already exist. We greedily contract a vertex~$x$
that minimizes its \emph{importance} $I(x)$ defined by 
\[
I(x)=L(x)+\frac{\left|A(x)\right|}{\left|D(x)\right|}+\frac{\sum_{a\in A(x)}h(a)}{\sum_{a\in D(x)}h(a)}.
\]
We maintain a priority queue that contains all vertices weighted by
$I$. Initially all vertices are inserted with their exact importance.
As long as the queue is not empty, we remove a vertex~$x$ with minimum
importance $I(x)$ and contract it. This modifies the importance of
other vertices. However, our weighting function is chosen such that
only the importance of adjacent vertices is influenced, if the witness
search was perfect. We therefore only update the importance values
of all vertices~$y$ in the queue that are adjacent to~$x$. In practice, 
with a limited witness search, we sometimes choose a vertex~$x$ with
a sightly suboptimal $I(x)$. However, preliminary experiments have
shown that this effect can be safely ignored. Hence, for the experiments
presented in Section~\ref{sec:experiments}, we do not use lazy updates
or periodic queue rebuilding as proposed in~\cite{gssv-erlrn-12}.

\section{Metric-Independent Order}\label{sec:ch-order}

The metric-dependent orders presented in the previous section lead to very good results on road graphs with travel time metric. 
However, the results for the distance metric are not as good and the orders are completely impracticable to compute Contraction Hierarchies without witness search as our experiments in Section~\ref{sec:experiments} show.
To support metric-independence, we therefore use \emph{nested dissection} orders as suggested in \cite{bcrw-sssch-13} or ND-orders for short. 
An order~$\pi$ for~$G$ is computed recursively by determining a balanced separator~$S$ of minimum cardinality that splits~$G$ into two parts induced by the vertex sets~$A$ and~$B$. 
The vertices of~$S$ are assigned to $\pi(n-\left|S\right|)\ldots\pi(n)$ in an arbitrary order. 
Orders~$\pi_{A}$ and~$\pi_{B}$ are computed recursively and assigned to
$\pi(1)\ldots\pi(\left|A\right|)$ and $\pi(\left|A\right|+1)\ldots\pi(\left|A\right|+\left|B\right|)$,
respectively. 
The base case of the recursion is reached when the subgraphs are empty.
Computing ND-orders requires good graph bisectors,
which in theory is $NP$-hard. However, recent years have seen heuristics
that solve the problem very well even for continental road graphs~\cite{ss-tlagh-13,dgrw-ecbbg-12,dgrw-gpnc-11}.
This justifies assuming in our particular context that an efficient
bisection oracle exists. 
We experimentally examine the performance of
nested dissection orders computed by NDMetis~\cite{kk-mspig-99}
and KaHIP~\cite{ss-tlagh-13} in Section~\ref{sec:experiments}.
After having obtained the nested dissection order we reorder the in-memory
vertex IDs of the input graph accordingly, \ie, the contraction order
of the reordered graph is the identity function. This improves cache locality
and we have seen a resulting acceleration of a factor 2 to 3 in query
times. In the remainder of this section we prepare and provide a theoretical
approximation result. 

For $\alpha\in(0,1)$, let $K_{\alpha}$, be a class of graphs that
is closed under subgraph construction and admits balanced separators
$S$ of cardinality $O(n^{\alpha})$.
\begin{lemma}
\label{lem:nd-performance}For every $G\in K_{\alpha}$ a ND-order
results in $O(n^{\alpha})$ vertices in the maximum search space.
\end{lemma}
The proof of this lemma is a straightforward argument using a geometric
series as described in \cite{bcrw-sssch-13}.
As a direct consequence, the average number of vertices is also in
$O(n^{\alpha})$ and the number of arcs in $O(n^{2\alpha})$. 
\begin{lemma}
\label{lem:sep-creates-clique}For every connected graph $G$ with
minimum balanced separator $S$ and every order $\pi$, the chordal
super graph $G_{\pi}^{*}$ contains a clique of $\left|S\right|$ vertices.
Furthermore, there are at least $n/3$ vertices such that this clique
is a subgraph of their search space in $G_{\pi}^{\wedge}$. 
\end{lemma}
This lemma is a minor adaptation and extension of~\cite{lrt-gnd-79}, who only prove that such a clique exists but not that it lies within enough search spaces.
We provide the full proof for self-containedness.
\begin{proof}
Consider the subgraph~$G_{i}$ of $G_{\pi}^{*}$ induced by the vertices
$\pi(1)\ldots\pi(i)$. Do not confuse with the core graph~$G_{\pi,i}$.
Choose the smallest $i$, such that a connected component $A$ exists
in $G_{i}$ such that $\left|A\right|\ge n/3$. As $G$ is connected,
such an $A$ must exist. We distinguish two cases:
\begin{enumerate}
\item $\left|A\right|\le2n/3$: 
Consider the set of vertices $S'$ adjacent to $A$ in $G_{\pi}^{*}$ but not in $A$. 
Let $B$ be the set of all remaining vertices.
$S'$ is by definition a separator. It is balanced because $\left|A\right|\le2n/3$
and $\left|B\right|=n-\underset{\ge n/3}{\underbrace{\left|A\right|}}-\underset{\ge0}{\underbrace{\left|S'\right|}}\le2n/3$.
As $S$ is minimum, we have that $\left|S'\right|\ge\left|S\right|$.
For every pair of vertices $u\in S'$ and $v\in S'$ there exists a path through
$A$ as $A$ is connected. 
The vertices $u$ and $v$ are not in $G_i$ as otherwise they could be added to $A$.
The ranks of $u$ and $v$ are thus strictly larger than $i$.
On the other hand, the ranks of the vertices in $A$ are at most $i$ as they are part of $G_i$.
The vertices $u$ and $v$ thus have the highest ranks on the path.
They are therefore contracted last and therefore an edge $\{u,v\}$ in $G^{*}$ must exist.
$S'$ is therefore a clique.
Furthermore, from every $u\in A$ to every $v\in S'$ there exists
a path such that $v$ has the highest rank. Hence, $v$ is in the
search space of $u$, \ie, there are at least$\left|A\right|\ge n/3$
vertices whose search space contains the full $S'$-clique.
\item $\left|A\right|>2n/3$: As $i$ is minimum, we know that $\pi(i)\in A$
and that removing it disconnects $A$ into connected subgraphs $C_{1}\ldots C_{k}$.
We know that $\left|C_{j}\right|<n/3$ for all $j$ because $i$ is
minimum. We further know that $\left|A\right|=1+\sum\left|C_{j}\right|>2n/3$.
We can therefore select a subset of components $C_{k}$ such that
the number of their vertices is at most $2n/3$ but at least $n/3$.
Denote by $A'$ their union. Note that $A'$ does not contain $\pi(i)$.
Consider the vertices $S'$ adjacent to $A'$ in $G_{\pi}^{*}$. The
set $S'$ contains $\pi(i)$. Using an argument similar to Case~1,
one can show that $\left|S'\right|\ge\left|S\right|$. But since $A'$
is not connected, we cannot directly use the same argument to show
that $S'$ forms a clique in $G^{*}$. Observe that $A'\cup\{\pi(i)\}$
is connected and thus the argument can be applied to $S'\backslash\{\pi(i)\}$
showing that it forms a clique. This clique can be enlarged by adding
$\pi(i)$ as for every $v\in S'\backslash\{\pi(i)\}$ a path through
one of the components $C_{k}$ exists where $v$ and $\pi(i)$ have
the highest ranks and thus an edge $\{v,\pi(i)\}$ must exist. The
vertex set $S'$ therefore forms a clique of at least the required
size. It remains to show that enough vertices exist whose search space
contains the $S'$ clique. As $\pi(i)$ has the lowest rank in the
$S'$ clique, the whole clique is contained within the search space
of $\pi(i)$. It is thus sufficient to show that $\pi(i)$ is contained
in enough search spaces. As $\pi(i)$ is adjacent to each component
$C_{k}$, a path from each vertex $v\in A'$ to $\pi(i)$ exists such
that $\pi(i)$ has maximum rank showing that $S'$ is contained in
the search space of $v$. This completes the proof as $\left|A'\right|\ge n/3$. 
\end{enumerate}
\end{proof}
\begin{theorem}
\label{thm:Approx}Let $G$ be a graph from $K_{\alpha}$ with a minimum balanced
separator with $\Theta(n^{\alpha})$ vertices. Then a ND-order gives
an $O(1)$-approximation of the average and maximum search spaces
of an optimal metric-independent contraction hierarchy in terms of
vertices and arcs.\end{theorem}
\begin{proof}
The key observation of this proof is that the top level separator solely dominates the performance. 
Denote by $\pi_{nd}$ the ND-order and by $\pi_{opt}$ an optimal order. 
First, we show a lower bound on the performance of $\pi_{opt}$.
We then demonstrate that $\pi_{nd}$ achieves this lower bound showing that $\pi_{nd}$ is an $O(1)$-approximation.

As the minimum balanced separator has cardinality $\Theta(n^{\alpha})$, we know by Lemma~\ref{lem:sep-creates-clique} that a clique with $\Theta(n^{\alpha})$ vertices exists in $G_{\pi_{opt}}^{*}$.
As this clique is in the search space of at least one vertex with respect to $\pi_{opt}$, we know that the maximum number of vertices in the search space is at least $\Omega(n^{\alpha})$. 
Further, as this clique contains $\Theta(n^{2\alpha})$ arcs we also have a lower bound of $\Omega(n^{2\alpha})$ on the maximum number of arcs in a search space.
From these bounds for the worst case search space, we cannot directly derive bounds for the average search space.
Fortunately, Lemma~\ref{lem:sep-creates-clique} does not only tell us that this clique exists but that it must also be inside the search space of at least $n/3$ vertices.
For the remaining $2n/3$ vertices we use a very pessimistic lower bound: We assume that their search space is empty.
The resulting lower bound for the average number of vertices is $2/3\cdot\Omega(0)+1/3\cdot\Omega(n^{\alpha})=\Omega(n^{\alpha})$ and the lower bound for the average number of arcs is $2/3\cdot\Omega(0)+1/3\cdot\Omega(n^{2\alpha})=\Omega(n^{2\alpha})$.

We required that $G\in K_{\alpha}$, \ie, that recursive $O(n^\alpha)$ balanced separators exists.
This allows us to apply Lemma~\ref{lem:nd-performance}.
We therefore know that the number of vertices in the maximum search space of $G_{\pi_{nd}}^{\wedge}$ is in $O(n^{\alpha})$.
In the worst-case this search space contains $O(n^{2\alpha})$ arcs.
As the average case can never be better than the worst case, these upper bounds directly translate to upper bounds for the average search space.

As the given upper and lower bounds match, we can conclude that $\pi_{nd}$ is a $O(1)$-approximation in terms of average and maximum search space in terms of vertices and arcs.
\end{proof}

\section{Constructing the Contraction Hierarchy}\label{sec:ch-construction}

In this section, we describe how to efficiently compute the hierarchy~$G_{\pi}^{\wedge}$ for a given graph~$G$ and order~$\pi$. 
Weighted contraction hierarchies are commonly constructed using a dynamic adjacency array representation of the core graph. 
Our experiments show that this approach also works for the unweighted case, however, requiring more computational and memory resources because of the higher growth in shortcuts. 
It has been proposed~\cite{z-wchw-13} to use hash-tables on top of the dynamic graph structure to improve speed but at the cost of significantly increased memory consumption. 
In this section, we show that the contraction hierarchy construction can be done significantly faster on unweighted and undirected graphs. 
Note that in our toolchain, graph weights and arc directions are accounted for during the customization phase.

\begin{wrapfigure}{o}{6.5cm}%
\begin{center}
\includegraphics[scale=1.5]{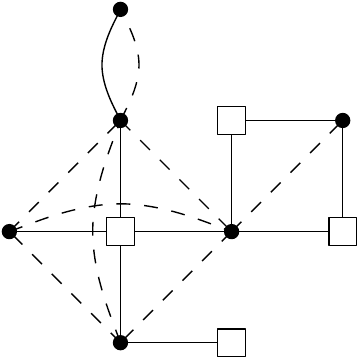}
\end{center}
\caption{Dots represent vertices in $G'$ and $H$. Squares are additional super vertices in $H$. Solid edges are in $H$ and dashed ones in $G'$. Notice how the neighbors of each super vertex in $H$ form a clique in $G'$. Furthermore, there are no two adjacent super vertices in $H$, \ie, they form an independent set.}
\label{fig:quotient-graph}
\end{wrapfigure}

Denote by $n$ the number of vertices in $G$ (and $G_{\pi}^{\wedge}$), by $m$ the number of edges in $G$, by $\hat{m}$ the number of arcs in $G_{\pi}^{\wedge}$, and by $\alpha(n)$ the inverse $A(n,n)$ Ackermann function.
For simplicity we assume that $G$ is connected. 
Our approach enumerates all arcs of~$G_{\pi}^{\wedge}$ in $O(\hat{m}\,\alpha(n))$ running time and has a memory consumption in $O(m)$. To store the arcs of $G_{\pi}^{\wedge}$, additional space in $O(\hat{m})$ is needed. 
The approach is heavily based upon the method of the quotient graph~\cite{gl-aqgms-78}. 
To the best of our knowledge it has not yet been applied in the context of route planning and there exists no complexity analysis for the specific variant employed by us. 
Therefore we discuss both the approach and present a running time analysis in the remainder of the section.

Recall that to compute the contraction hierarchy~$G_{\pi}^{\wedge}$ from a given input graph~$G$ and order~$\pi$, one iteratively contracts each vertex, adding shortcuts between its neighbors.
Let~$G' = G_{\pi,i}$ be the core graph in iteration~$i$.
We do not store~$G'$ explicitly but employ a special data structure called~\emph{contraction graph} for efficient contraction and neighborhood enumeration. 
The contraction graph~$H$ contains both yet uncontracted core vertices as well as an independent set of virtually contracted \emph{super vertices}, see Figure~\ref{fig:quotient-graph} for an illustration.
These super vertices enable us to avoid the overhead of dynamically adding shortcuts to $G'$.
For each vertex in $H$ we store a marker bit indicating whether it is a super vertex.
Note that $G'$ can be obtained by contracting all super vertices in $H$.

\subsection{Contracting Vertices}

A vertex $x$ in $G'$ is contracted by turning it into a super vertex.
However, creating new super vertices can violate the independent set property.
We restore it by merging neighboring super vertices:
Denote by $y$ a super vertex that is a neighbor of $x$.
We rewire all edges incident to $y$ to be incident to $x$ and remove $y$ from $H$.
To support efficiently merging vertices in $H$, we store a linked list of neighbors for each vertex.
When merging two vertices we link these lists together.
Unfortunately, combining these lists is not enough as the former neighbors $z$ of $y$ still have $y$ in their list of neighbors.
We therefore further maintain a union-find data structure:
Initially all vertices are within their own set.
When merging $x$ and $y$, the sets of $x$ and $y$ are united.
We chose $x$ as representative as $y$ was deleted.\footnote{Or alternatively, we can let the union-find data structure choose the new representative. We then denote by $x$ the new representative and by $y$ the other vertex. In this variant, it is possible that the new $x$ is the old $y$, which can be confusing. For this reason, we describe the simpler variant, where $x$ is always chosen as representative and thus $x$ always refers to the same vertex.}
When $z$ enumerates its neighbors, it finds a reference to $y$.
It can then use the union-find data structure to determine that the representative of $y$'s set is $x$.
The reference in $z$'s list is thus interpreted as pointing to $x$.

It is possible that merging vertices can create multi-edges and loops.
For example, consider that the neighborhood list of $y$ contains $x$.
After merging, the united list of $x$ will therefore contain a reference to $x$.
Similarly it will contain a reference to $y$, which after looking up the representative is actually $x$.
Two loops are thus created at $x$ per merge.
Furthermore, consider a vertex $z$ that is a neighbor of both $y$ and $x$. 
In this case the neighborhood list of $x$ will contain two references to $z$. %
These multi-edges and loops need to be removed.
We do this lazily and remove them in the neighborhood enumeration instead of removing them in the merge operation.

\subsection{Enumerating Neighbors}

Suppose that we want to enumerate the neighbors of a vertex $x$ in $H$.
Note that $x$'s neighborhood in $H$ differs from its neighborhood in $G'$.
The neighborhood of $x$ in $H$ can contain super vertices, as super vertices are only contracted in $G'$.
We maintain a boolean marker that indicates which neighbors have already been enumerated.
Initially no marker is set.
We iterate over $x$'s neighborhood list.
For each reference we lookup the representative $v$.
If $v$ was already marked or is $x$, we remove the reference from the list.
If $v$ was not marked and is not $x$, we mark it and report it as a neighbor.
After the enumeration we reset all markers by enumerating the neighbors again.

However, during the execution of our algorithm, we are not interested in the neighborhood of $x$ in $H$, but we want the neighborhood of $x$ in $G'$, \ie, the algorithm should not list super vertices.
Our algorithm conceptually first enumerates the neighborhood of $x$ and then contracts $x$.
We actually do this in reversed order.
We first contract $x$.
After the contraction $x$ is a super vertex.
Because of the independent set property, we know that $x$ has no super vertex neighbors in $H$.
We can thus enumerate $x$'s neighbors in $H$ and exploit that in this particular situation the neighborhoods of $x$ in $G'$ and $H$ coincide.

\subsection{Performance Analysis}

As there are no memory allocations, it is clear that the working space memory consumption is in $O(m)$.
Proving a running time in $O(\hat{m} \alpha(n))$ is less clear.
Denote by $d(x)$ the degree of $x$ just before $x$ is contracted.
$d(x)$ coincides with the upward degree of $x$ in $G_{\pi}^{\wedge}$ and thus $\sum d(x) = \hat{m}$.
We first prove that we can account for the neighborhood cleanup operations outside of the actual algorithm.
This allows us to assume that they are free within the main analysis.
We then show that contracting a vertex $x$ and enumerating its neighbors is in $O(d(x)\alpha (n))$.
Processing all vertices has thus a running time in $O(\hat{m} \alpha(n))$.

The neighborhood list of $x$ can contain duplicated references and thus its length can be larger than the number of neighbors of $x$.
Further, for each entry in the list, we need to perform a union find lookup.
The costs of a neighborhood enumeration can thus be larger than $O(d(x) \alpha(n))$.
Fortunately, the first neighborhood enumeration compactifies the neighborhood list and thus every subsequent enumeration runs in $O(d(x) \alpha(n))$.
Removing a reference has a cost in $O(\alpha(n))$.
Our algorithm never adds references.
Initially there are $\Theta(m)$ references. 
The total costs for removing references over the whole algorithm are thus in $O(m \alpha(n))$.
As our graph is assumed to be connected, we have that $m \in O(m')$ and therefore $O(m \alpha(n))\subseteq O(\hat{m} \alpha(n))$.
We can therefore assume that removing references is free within the algorithm.
As removing a reference is free, we can assume that even the first enumeration of the neighbors of $x$ is within $O(d(x) \alpha(n))$.
Merging two vertices consists of redirecting a constant number of references within a linked list.
The merge operation is thus in $O(1)$.

Our algorithm starts by enumerating all neighbors of $x$ to determine all neighboring super vertices in $O(d(x)\alpha (n))$ time.
There are at most $d(x)$ neighboring super vertices and therefore the costs of merging all super vertices into $x$ is in $O(d(x))$.
We subsequently enumerate all neighbors a second time to output the arcs of $G_{\pi}^{\wedge}$.
The costs of this second enumeration is also within $O(d(x)\alpha (n))$.
The whole algorithm thus runs in $O(\hat{m} \alpha(n))$ time as $\sum d(x) = \hat{m}$, which completes the proof.

\subsection{Adjacency Array}

While the described algorithm is efficient in theory, linked lists cause too many cache misses in practice. %
We therefore implemented a hybrid of a linked list and an adjacency array, which has the same worst case performance, but is more cache-friendly in practice.
An element in the linked list does not only hold a single reference, but a small set of references organized as small arrays called blocks.
The neighbors of every original vertex form a single block. 
The initial linked neighborhood list are therefore composed of a single block.
We merge two vertices by linking their blocks together.
If all references are deleted from a block, we remove it from the list.

\section{Enumerating Triangles}\label{sec:enum-triangles}

\begin{wrapfigure}{o}{3.5cm}%
\vspace{-2em}
\begin{center}
\includegraphics[scale=1.5]{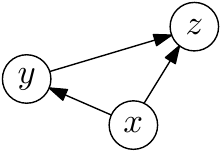}
\end{center}

\caption{A triangle in~$G_{\pi}^{\wedge}$. The triple~$\{x,y,z\}$ is a lower
triangle of the arc~$(y,z)$, an intermediate
triangle of the arc~$(x,z)$, and an upper
triangle of the arc~$(x,y)$.}
\label{fig:triangle}
\end{wrapfigure}

A triangle $\{x,y,z\}$ is a set of 3 adjacent vertices.
A triangle can be an upper, intermediate or lower triangle with respect to an arc $(x,y)$, as illustrated in Figure~\ref{fig:triangle}.
A triangle $\{x,y,z\}$ is a \emph{lower triangle} of $(y,z)$ if $x$ has the lowest rank among the three vertices.
Similarly $\{x,y,z\}$ is a \emph{upper triangle} of $(x,y)$ if $z$ has the highest rank and $\{x,y,z\}$ is a \emph{intermediate triangle} of $(x,z)$ if $y$'s rank is between the ranks of $x$ and $z$.
The triangles of an edge $(a,b)$ can be characterized using the upward $N_u$ and downward $N_d$ neighborhoods of $a$ and $b$.
There is a lower triangle $\{a,b,c\}$ of an arc $(a,b)$ if and only if $c\in N_d(a)\cap N_d(b)$.
Similarly, there is an intermediate triangle $\{a,b,c\}$ of an arc $(a,b)$ with $\pi^{-1}(a)<\pi^{-1}(b)$ if and only if $c \in N_u(a)\cap N_d(b)$ and an upper triangle $\{a,b,c\}$ of an arc $(a,b)$ if and only if $c \in N_u(a)\cap N_u(b)$.
The triangles of an arc can thus be enumerated by intersecting the neighborhoods of the arc's endpoints.

Efficiently enumerating all lower triangles of an arc is an important base operation of the customization~(Section~\ref{sec:customization}) and path unpacking algorithms~(Section~\ref{sec:query}).  
It can be implemented using adjacency arrays or accelerated using extra preprocessing. 
Note that in addition to the vertices of a triangle we are interested in the IDs of the participating arcs as we need these to access the metric of an arc.

\paragraph{Basic Triangle Enumeration}

Triangles can be efficiently enumerated by exploiting their characterization using neighborhood intersections.
We construct an upward and a downward adjacency array for $G_{\pi}^{\wedge}$, where incident arcs are ordered by their head respectively tail vertex ID.
The lower triangles of an arc $(x,y)$ can be enumerated by simultaneously scanning the downward neighborhoods of $x$ and $y$ to determine their intersection.
Intermediate and upper triangles are enumerated analogously using the upward adjacency arrays.
For later access to the metric of an arc, we also store each arc's ID in the adjacency arrays.
This approach requires space proportional to the number of arcs in $G_{\pi}^{\wedge}$.

\paragraph{Triangle Preprocessing}

Instead of merging the neighborhoods on demand to find all lower triangles, we propose to create a \emph{triangle adjacency array} structure that maps the arc ID of $(x,y)$ to the pair of arc ids of $(z,x)$ and $(z,y)$ for every lower triangle $\{x,y,z\}$ of $(x,y)$. 
This requires space proportional to the number of triangles~$t$ in $G_{\pi}^{\wedge}$, but allows for a very fast access. 
Analogous structures allow us to efficiently enumerate all upper triangles and all intermediate triangles.

\paragraph{Hybrid Approach}

For less well-behaved graphs the number of triangles $t$ can significantly outgrow the number of arcs in $G_{\pi}^{\wedge}$.
In the worst case $G$ is the complete graph and the number of triangles $t$ is in $\Theta(n^{3})$ whereas the number of arcs is in $\Theta(n^{2})$.
It can therefore be prohibitive to store a list of all triangles. 
We therefore propose a hybrid approach. 
We only precompute the triangles for the arcs $(u,v)$ where the level of $u$ is below a certain threshold.
The threshold is a tuning parameter that trades space for time.

\paragraph{Comparison with CRP}

Triangle preprocessing has similarities with micro and macro code in CRP~\cite{dw-fcrn-13}. 
In the following, we compare the space consumption of these two approaches against our lower triangles preprocessing scheme. However, note that at this stage we do not yet consider travel direction on arcs.  
Hence, let $t$ be the number of undirected triangles and $m$ be the number of arcs in $G_{\pi}^{\wedge}$; further let $t'$ be the number of directed triangles and $m'$ be the number of arcs used in~\cite{dw-fcrn-13}.
If every street is a one-way street, then $m' = m$ and $t' = t$; otherwise, without one-way streets, $m'=2m$ and $t'=2t$.

Micro code stores an array of triples of pointers to the arc weights of the three arcs in a directed triangle, \ie, it stores the equivalent of $3t'$ arc IDs. 
Computing the exact space consumption of macro code is more difficult.
However, it is easy to obtain a lower bound:
Macro code must store for every triangle at least the pointer to the arc weight of the upper arc.
This yields a space consumption equivalent to \emph{at least} $t'$ arc IDs.
In comparison, our approach stores for each triangle the arc IDs of the two lower arcs.
Additionally, the index array of the triangle adjacency array, which maps each arc to the set of its lower triangles, maintains $m+1$ entries. Each entry has a size equivalent to an arc ID.
Our total memory consumption is thus $2t+m+1$ arc IDs.

Hence, our approach always requires less space than micro code. It has similar space consumption as macro code if one-way streets are rare, otherwise it needs at most twice as much data. 
However, the main advantage of our approach over macro code is that it allows for random access, which is crucial in the algorithms presented in the following sections.

\section{Customization}\label{sec:customization}

Up to now we only considered the metric-independent first preprocessing phase. 
In this section, we describe the second metric-dependent preprocessing phase, known as customization. That is, we show how to efficiently extend the weights of the input graph to a corresponding metric with weights for all arcs in $G_\pi^\wedge$.
We consider three different distances between the vertices:
We refer to $\distI(s,t)$ as the shortest $st$-path distance in the input graph $G$.
With $\distUD(s,t)$ we denote the shortest $st$-path distance in $G_\pi^\wedge$ when only considering up-down paths.
Finally, let $\distA(s,t)$ be the shortest $st$-path distance in $G_\pi^*$, \ie, when allowing arbitrary not-necessarily up-down paths in $G_\pi^\wedge$.

For correctness of the CH query algorithms (\cf~Section~\ref{sec:query}) it is necessary that between any pair of vertices $s$ and $t$ a shortest up-down $st$-path in $G_\pi^\wedge$ exists with the same distance as the shortest $st$-path in the input graph $G$. 
In other words, $\distI(s,t)=\distA(s,t)=\distUD(s,t)$ must hold for all vertices $s$ and $t$.
We say that a metric that fulfills $\distI(s,t)=\distA(s,t)$ \emph{respects} the input weights.
If additionally $\distA(s,t)=\distUD(s,t)$ holds, we call the metric \emph{customized}.
Note that customized metrics are not necessarily unique. However, there is a special customized metric, called \emph{perfect} metric $m_P$, where for every arc $(x,y)$ in $G_\pi^\wedge$ the weight of this arc $m_P(x,y)$ is equal to the shortest path distance $\distI(x,y)$.
We optionally use the perfect metric to perform perfect witness search.

Constructing a respecting metric is trivial: 
Assign to all arcs of~$G_\pi^\wedge$ that already exist in $G$ their input weight and to all other arcs $+\infty$.
Computing a customized metric is less trivial.
We therefore describe in Section~\ref{sec:basic-custom} the basic customization algorithm that computes a customized metric $m_C$ given a respecting one.
Afterwards, we describe the perfect customization algorithm that computes the perfect metric $m_P$ given a customized one (such as for example $m_C$).
Finally, we show how to employ the perfect metric to perform a perfect witness search.

\subsection{Basic Customization}\label{sec:basic-custom}

A central notion of the basic customization algorithm is the \emph{lower triangle inequality}, which is defined as following:
A metric $m_C$ fulfills it if for all lower triangles $\{x,y,z\}$ of each arc $(x,y)$ of $G_\pi^\wedge$, it holds that $m_C(x,y) \le m_C(x,z)+m_C(z,y)$.
We show that every respecting metric that also fulfills this inequality is customized.
Our algorithm exploits this by transforming the given respecting metric in a coordinated way that maintains the respecting property and assures that the lower triangle inequality holds.
The resulting metric is thus customized.
We first describe the algorithm and prove that the resulting metric is respecting and fulfills the inequality. 
We then prove that this is sufficient for the resulting metric to be customized.

Our algorithm iterates over all arcs $(x,y) \in G_\pi^\wedge$ ordered \emph{increasingly} by the rank of $x$ in a bottom-up fashion.
For each arc $(x,y)$, it enumerates all lower triangles $\{x,y,z\}$ and checks whether the path $x\rightarrow z \rightarrow y$ is shorter than the path $x\rightarrow y$. 
If this is the case, then it decreases $m_C(x,y)$ so that both paths are equally long.
Formally, it performs for every arc $(x,y)$ the operation $m_C(x,y)\leftarrow\min\{m_C(x,y),m_C(x,z)+m_C(z,y)\}$.
Note, that this operation never assigns values that do not correspond to a path length and therefore $m_C$ remains respecting.
By induction over the vertex levels, we can show that after the algorithm is finished, the lower triangle inequality holds for every arc, \ie, for every arc $(x,y)$ and lower triangle $\{x,y,z\}$ the inequality $m_C(x,y)\le m_C(x,z)+m_C(z,y)$ holds.
The key observation is that by construction the rank of $z$ must be strictly smaller than the ranks of $x$ and $y$. 
The final weights of~$m_C(x,z)$ and $m_C(z,y)$ have therefore already been computed when considering~$(x,y)$. 
In other words, when the algorithm considers the arc $(x,y)$, the weights $m_C(x,z)$ and $m_C(z,y)$ are guaranteed to remain unchanged until termination.

\begin{theorem}
\label{thm:ch-correct}
Every respecting metric that additionally fulfills the lower triangle inequality is customized. 
\end{theorem}
\begin{proof}
We need to show that between any pair of vertices $s$ and $t$ a shortest up-down $st$-path exists.
As we assumed for simplicity that $G$ is connected, there always exists a shortest not-necessarily up-down path from $s$ to $t$.
Either this is an up-down path, or a subpath $x\rightarrow z\rightarrow y$ with~$\pi^{-1}(x)>\pi^{-1}(z)$ and $\pi^{-1}(y)>\pi^{-1}(z)$ must exist.
As~$z$ is contracted before~$x$ and $y$, an edge~$\{x,y\}$ must exist. 
Because of the lower triangle inequality, we further know that $m(x,y)\le m(x,z)+m(z,y)$ and thus replacing $x\rightarrow z\rightarrow y$ by $x\rightarrow y$ does not make the path longer.
Either the path is now an up-down path or we can apply the argument iteratively.
As the path has only a finite number of vertices, this is guaranteed to eventually yield the up-down path required by the theorem and thus this completes the proof.
\end{proof}

\subsection{Perfect Customization}

\begin{wrapfigure}{o}{4.5cm}%
\begin{center}
\includegraphics[scale=1.5]{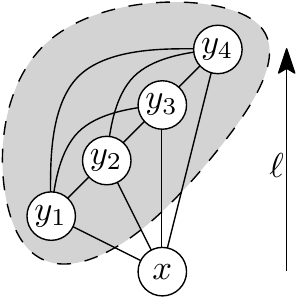}
\end{center}

\caption{The vertices $y_1\ldots y_4$ denote the upper neighborhood $N_{u}(x)$ of $x$. 
They form a clique (grey area) because $x$ was contracted first. As $\ell(x)<\ell(y_j)$ for every $j$, we know by the induction hypothesis that the arcs in this clique are weighted by shortest path distances. We therefore have an all-pair shortest path distance table among all~$y_j$. We have to show that using this information we can compute shortest path distances for all arcs outgoing of $x$.}
\label{fig:perfect-custom}
\end{wrapfigure}

Given a customized metric $m_C$, we want to compute the perfect metric $m_P$.
We first copy all values of $m_C$ into $m_P$.
Our algorithm then iterates over all arcs $(x,y)$ \emph{decreasing} by the rank of $x$ in a top-down fashion.
For every arc it enumerates all intermediate and upper triangles $\{x,y,z\}$ and checks whether the path over $z$ is shorter and adjusts the value of $m_P(x,y)$ accordingly, \ie, it performs $m_P(x,y)\leftarrow\min\{m_P(x,y),m_P(x,z)+m_P(z,y)\}$.
After all arcs have been processed $m_P$ is the perfect metric, as is shown in the following theorem.
\begin{theorem}
\label{thm:perfect-custom}
After the perfect customization, $m_P(x,y)$ corresponds to the shortest $xy$-path distance for every arc $(x,y)$, \ie, $m_P$ is the perfect metric.  
\end{theorem}
\begin{proof}
We have to show that after the algorithm has finished processing a vertex~$x$, all of its outgoing arcs in $G_{\pi}^{\wedge}$ are weighted by the shortest path distance. 
We prove this by induction over the level of the processed vertices.
The top-most vertex is the only vertex in the top level. 
It does not have any upward arcs and thus the algorithm does not have anything to do. 
This forms the base case of the induction. 
In the inductive step, we assume that all vertices with a strictly higher level have already been processed.
As consequence, we know that the upward neighbors of $x$ form a clique weighted by shortest path distances.
Denote these neighbors by $y_i$.
The situation is depicted in Figure~\ref{fig:perfect-custom}.
The weights of the $y_i$ encode a complete shortest path distance table between the upward neighbors of $x$. 

Pick some arbitrary arc~$(x,y_j)$. 
We show the correctness of our algorithm by proving that either $m_C(x,y_j)$ is already the shortest path distance or a neighbor $y_k\in N_{u}(x)$ must exist such that $x\rightarrow y_k \rightarrow y_j$ is a shortest up-down path.
For the rest of this paragraph assume the existence of $y_k$, we prove its existence in the next paragraph.
If $m_C(x,y_j)$ is already the shortest $xy_j$-path distance, then enumerating triangles will not change $m_C(x,y_j)$ and is thus correct.
If $m_C(x,y_j)$ is not the shortest $xy_j$-path distance, then enumerating all intermediate and upper triangles of $(x,y_j)$ is guaranteed to find the $x\rightarrow y_k \rightarrow y_j$ path and thus the algorithm is correct.
The upper triangles correspond to paths with $\ell(y_k) > \ell(y_j)$ while the intermediate triangles to paths with $\ell(y_k) < \ell(y_j)$.

It remains to show that the $x\rightarrow y_k \rightarrow y_j$ shortest up-down path actually exists.
As the metric is customized at every moment during the perfect customization, we know that a shortest up-down $xy_j$-path $K$ exists.
As $K$ is an up-down path, we can conclude that the second vertex of $K$ must be an upward neighbor of $x$.
We denote this neighbor by $y_k$ .
$K$ thus has the following structure: $x\rightarrow y_k \rightarrow \ldots \rightarrow y_j$.
As $y_k$ has a higher rank than $x$, $m_P(y_k, y_j)$ is guaranteed to be the shortest $y_{k}y_{j}$-path distance, and therefore we can replace the $y_k \rightarrow \ldots \rightarrow y_j$ subpath of $K$ by $y_k \rightarrow y_j$ and we have proven that the required $x\rightarrow y_k \rightarrow y_j$ shortest up-down path exists, which completes the proof.
\end{proof}

\subsection{Perfect Witness Search}

Using the perfect customization algorithm, we can efficiently compute the weighted CH with a minimum number of arcs with respect to the same contraction order.
We present two variants of our algorithm. 
The first variant consists of removing all arcs $(x,y)$ whose weight after the basic customization $m_C(x,y)$ does not correspond to the shortest $xy$-path distance $m_P(x,y)$.
This variant is simple and always correct, but it does not remove as many arcs as possible, if a pair of vertices $a$ and $b$ exists in the input graph such that there are multiple shortest $ab$-paths.
The second variant\footnote{Note that the second algorithm variant exploits that we defined weights as being non-zero. If zero weights are allowed, it may remove too many arcs. A workaround consists of replacing all zero weights with a very small but non-zero weight.} also removes these additional arcs.
An arc $(x,y)$ is removed if and only if an upper or intermediate triangle $\{x,y,z\}$ exists such that the shortest path from $x$ over $z$ to $y$ is no longer than the shortest $xy$-path.
However, before we can proof the correctness of the second variant, we need to introduce some technical machinery, which will also be needed in the correctness proof of the stalling query algorithm.
We define the ``height'' of a not-necessarily up-down path in $G_{\pi}^{*}$.
We show that with respect to every customized metric, for every path that is not up-down, an up-down path must exist that is strictly higher and is no longer.

\subsubsection{Variant for Graphs with Unique Shortest Paths}

The first algorithm variant consists of removing all arcs $(x,y)$ from the CH for which $m_P(x,y)\neq m_C(x,y)$. 
It is optimal if shortest path are unique in the input graph, \ie, between every pair of vertices $a$ and $b$ there is only one shortest $ab$-path.
This simple algorithm is correct as the following theorem shows.
\begin{theorem}
If the input graph has unique shortest paths between all pairs of vertices, then we can remove an arc $(x,y)$ from the CH if and only if $m_P(x,y)\neq m_C(x,y)$.
\end{theorem}
\begin{proof}
We need to show that after removing all arcs, there still exists a shortest up-down path between every pair of vertices $s$ and $t$.
We know that before removing any arc a shortest up-down $st$-path $K$ exists. 
We show that no arc of $K$ is removed and thus $K$ also exists after removing all arcs.
Every subpath of $K$ must be a shortest path as $K$ is a shortest path.
Every arc of $K$ is a subpath.
However, we only remove arcs such that $m_P(x,y)\neq m_C(x,y)$, \ie, which are not shortest paths.

To show that no further arcs can be removed we need to show that if $m_P(x,y)=m_C(x,y)$, then the path $x\rightarrow y$ is the only shortest up-down path.
Denote the $x\rightarrow y$ path by $Q$.
Suppose that another shortest up-down path $R$ existed. 
$R$ must be different than $Q$, \ie, a vertex $z$ must exist that lies on $R$ but not on $Q$.
As $z$ must be reachable from $x$, we know that $z$ is higher than $x$.
Unpacking the path $Q$ in the input graph yields a path where $x$ and $y$ are the highest ranked vertices and thus this unpacked path cannot contain $z$.
Unpacking $R$ yields a path that contains $z$ and is therefore different.
Both paths are shortest paths from $x$ to $y$ in the input graph.
This contradicts the assumption that shortest paths are unique.
We have thus proven that, if the input graph has unique shortest paths, we can remove an arc $(x,y)$ if and only if $m_P(x,y)\neq m_C(x,y)$.
\end{proof}

\subsubsection{Variant for General Graphs}

\begin{wrapfigure}{o}{7.5cm}%
\vspace{-5em}
\begin{center}

\includegraphics[scale=1.2]{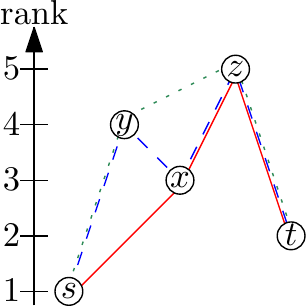}
\vspace{-2em}
\end{center}

\caption{The rank sequence of the solid red path is $[3,2,1]$.
The 3 is the minimum of the ranks of the endpoints of the $\{x,z\}$ edge.
Similarly the 2 is induced by the $\{z,t\}$ edge and the 1 by the $\{s,x\}$ edge.
The rank sequence of the blue dashed path is $[3,3,2,1]$ and the rank sequence of the green dotted path is $[4,2,1]$.
The solid red path is the lowest followed by the blue dashed path and the green dotted path is the highest.
}
\label{fig:stalling}
\end{wrapfigure}

Using the first variant of our algorithm, even when shortest paths are not unique in the original graph is not wrong.
However, it is possible that some arcs are not removed that could be removed.
Our second algorithm variant does not have this weakness.
It removes all arcs $(x,y)$ for which an intermediate or upper triangle $\{x,y,z\}$ exists such that $m_P(x,y)=m_P(x,z)+m_P(z,y)$.
These arcs can efficiently be identified while running the perfect customization algorithm.
An arc $(x,y)$ is marked for removal if an upper or intermediate triangle $\{x,y,z\}$ with $m_C(x,y)\ge m_C(x,z)+m_C(z,y)$ is encountered.
However, before we can proof the correctness of the second variant, we need to introduce some technical machinery.

We want to order paths by ``height''.
To achieve this, we first define for each path $K$ in $G_{\pi}^{*}$ its \emph{rank sequence}.
We order paths by comparing the rank sequences lexicographically.
Denote by $v_i$ the vertices in $K$.
For each edge $\{v_i,v_{i+1}\}$ in $K$ the rank sequence contains $\min \{r(v_i), r(v_{i+1})\}$.
The numbers in the rank sequences are order decreasingly.
Two paths have the same height if one rank sequence is a prefix of the other.
Otherwise we compare the rank sequences lexicographically.
This ordering is illustrated in Figure~\ref{fig:stalling}.
We proof the following technical lemma:
\begin{lemma}
\label{lem:higher-path}
Let $m_C$ be some customized metric.
For every $st$-path $K$ that is no up-down path, an up-down $st$-path $Q$ exists, such that $Q$ is strictly higher $K$ and $Q$ is no longer than $K$ with respect to $m_C$.
\end{lemma}
\begin{proof}
Denote by $v_i$ the vertices on the path $K$.
As $K$ is no up-down path, there must exist a vertex $v_i$ on $K$ that has lower ranks than its neighbors $v_{i-1}$ and $v_{i+1}$.
$v_{i-1}$ and $v_{i+1}$ are different vertices because they are part of a shortest path and zero weights are not allowed.
Further, as $v_i$ is contracted before its neighbors, there must be a edge between $v_{i-1}$ and $v_{i+1}$.
As the metric is customized, $m_C(v_{i-1},v_{i+1})\le m_C(v_{i-1},v_{i})+m_C(v_{i},v_{i+1})$ must hold.
We can thus remove $v_i$ from $K$ and replace it with the $(v_{i-1},v_{i+1})$ are without making the path longer.
Denote this new path by $R$.
$R$ is higher than $K$ as we replaced $r(v_i)$ in the rank sequence by $\min\{r(v_{i-1}), r(v_{i+1})\}$, which must be larger.
Either, $R$ is an up-down path or we apply the argument iteratively.
In each iteration the path looses a vertex and therefore we can guarantee that eventually we obtain an up-down path that is higher than $K$ and no longer.
This is the desired up-down path $Q$ that is no longer than $K$ and strictly higher.
\end{proof}
Note that, this lemma does not exploit any property that is inherent to CHs with a metric-independent contraction ordering and is thus applicable to every CH.

Given this technical lemma, we can prove the correctness of the second variant of our algorithm.
\begin{theorem}
We can remove an arc $(x,y)$ if and only if an upper or intermediate triangle $\{x,y,z\}$ exists with $m_P(x,y)=m_P(x,z)+m_P(z,y)$.
\end{theorem}
\begin{proof}
We need to show that for every pair of vertices $s$ and $t$ a shortest up-down $st$-path exists, that uses no removed arc.
We show that a highest shortest up-down $st$-path has this property.
As the metric is customized, we know that a shortest up-down $st$-path $K$ exists before removing any arcs.
If $K$ does not contain an arc $(x,y)$ for which an upper or intermediate triangle $\{x,y,z\}$ exists with $m_P(x,y)=m_P(x,z)+m_P(z,y)$, then there is nothing to show.
Otherwise, we modify $K$ by inserting $z$ between $x$ and $y$.
This does not modify the length of $K$, but we can no longer guarantee that $K$ is an up-down path.
If $\{x,y,z\}$ was an intermediate triangle, then $K$ is still an up-down path.
However, it is strictly higher, as we added $r(z)$ into the rank sequence, which is guaranteed to be larger than $r(x)$.
If $\{x,y,z\}$ was an upper triangle, then $K$ is still no up-down path anymore.
Fortunately, using Lemma~\ref{lem:higher-path} we can transform $K$ into an up-down path, that is no longer and strictly higher.
In both case, the new $K$ is an up-down path or we apply the argument iteratively.
As $K$ gets strictly higher in each iteration and the number of up-down paths is finite, we know that we will eventually obtain a shortest up-down $st$-path where no arc can be removed. 

Further, we need to show that if no such triangle exists, then an arc cannot be removed, \ie, we need to show that the only shortest up-down path from $x$ to $y$ is the path consisting only of the $(x,y)$ arc.
Assume that no such triangle and a further up-down path $Q$ existed.
$Q$ must contain a vertex beside $x$ and $y$ and all vertices in $Q$ must have the rank of $x$ or higher.
Consider the vertex $z$ that comes directly after $x$ in $Q$.
As $x$ is contracted before $z$ and $y$, an arc between $z$ and $y$ must exist.
Therefore, a triangle $\{x,y,z\}$ must exist that is an intermediate triangle, if $z$ has a lower rank than $y$ and is an upper triangle, if $z$ has a higher rank than $y$.
However, we assumed that no such triangle can exist.
We have thus proven that we can remove an arc $(x,y)$ if and only if an upper or intermediate triangle $\{x,y,z\}$ exists with $m_P(x,y)=m_P(x,z)+m_P(z,y)$.
\end{proof}

\subsection{Parallelization}

The basic customization can be parallelized by processing the arcs $(x,y)$ that depart within a level in parallel.
Between levels, we need to synchronize all threads using a barrier.
As all threads only write to the arc they are currently assigned to and only read from arcs processed in a strictly lower level, we can thus guarantee that no read/write conflict occurs.
Hence, no locks or atomic operations are needed.

On most modern processors, the perfect customization can be parallelized analogously to the basic customization algorithm.
However, seeing why this is correct is non-obvious because the exact order in which the threads are executed influences intermediate results.
Fortunately, the end result is always the same and independent of the execution order.
Our algorithm works as following: 
We iterate over all arcs departing within a level in parallel and synchronize all threads between levels.
For every arc $(x,y)$ we enumerate all upper and intermediate triangles and update $m_P(x,y)$ accordingly.
Consider the situation from Figure~\ref{fig:perfect-custom}.
Suppose that thread $A$ processes the arc $(x,y_A)$ at the same time as thread $B$ processes the arc $(x,y_B)$. 
Further, suppose that thread $A$ updates $m_P(x,y_A)$ at the same moment as thread $B$ enumerates the $\{x,y_B, y_A\}$ triangle.
In this situation it is unclear what value thread $B$ will see.
However, our algorithm is correct as long it is guaranteed that thread $B$ will either see the old value or the new value.

In the proof of Theorem~\ref{thm:perfect-custom}, we have shown, that for every vertex $x$ and arc $(x, y_i)$ either the arc $(x, y_i)$ already has the shortest path distance or an upper or intermediate triangle $\{x,y_i,y_j\}$ exists, such that $x \rightarrow y_j \rightarrow y_i$ is a shortest path.
No matter in which order the threads process the arcs, they do not modify shortest path weights.
This implies that the shortest path $x \rightarrow y_j \rightarrow y_i$ is thus retained, regardless of the execution order.
This shortest path is not modified and is guaranteed to exist before any arcs outgoing from the current level are processed.
Every thread is thus guaranteed to see it.
However, other weights can be modified. 
Fortunately, this is not a problem as long as we can guarantee that no thread sees a value that is below the corresponding shortest path distance.
Therefore, if we can guarantee that thread $B$ either sees the old value or the new value, as is the case on x86 processors, then the algorithm is correct.
If thread $B$ can see some mangled combination of the old value's bits and new value's bits, then we need to use locks or make sure that all outgoing arcs of $x$ are processed by the same thread.

\subsection{Directed Graphs}

Up to now we have focused on customizing undirected graphs.
If the input graph~$G$ is \emph{directed}, our toolchain works as follows: Based on the \emph{undirected unweighted} graph induced by~$G$ we compute a vertex ordering~$\pi$~(Section~\ref{sec:ch-order}), build the upward directed Contraction Hierarchy~$G_\pi^{\wedge}$~(Section~\ref{sec:ch-construction}), and optionally perform triangle preprocessing~(Section~\ref{sec:enum-triangles}).
For customization, however, we consider two weights per arc in~$G_\pi^\wedge$, one for each direction of travel. One-way streets are modeled by setting the weight corresponding to the forbidden traversal direction to~$\infty$. 
With respect to $\pi$ we define an upward metric~$m_u$ and a downward metric~$m_d$ on $G_\pi^\wedge$. For each arc~$(x,y) \in G$ in the directed input graph with input weight~$w(x,y)$, we set $m_u(x,y)=w(x,y)$ if $\pi^{-1}(x) < \pi^{-1}(y)$, \ie, if $x$ is ordered before $y$; otherwise, we set $m_d(x,y)=w(x,y)$. All other values of $m_u$ and $m_d$ are set to $\infty$.
In other words, each arc~$(x,y) \in G_\pi^\wedge$ of the Contraction Hierarchy has upward weight~$m_u(x,y) = w(x,y)$ if $(x,y) \in G$, downward weight~$m_d(x,y) = w(y,x)$ if $(y,x) \in G$, and~$\infty$ otherwise.

The basic customization considers both metrics~$m_u$ and $m_d$ simultaneously.
For every lower triangle $\{x,y,z\}$ of $(x,y)$ it sets $m_{u}(x,y)\leftarrow\min\{m_{u}(x,y),\allowbreak m_{d}(x,z)+m_{u}(z,y)\}$ and $m_{d}(x,y)\leftarrow\min\{m_{d}(x,y),\allowbreak m_{u}(x,z)+m_{d}(z,y)\}$.
The perfect customization can be adapted analogously.
For every intermediate triangle $\{x,y,z\}$ of $(x,y)$ the perfect customization sets $m_{u}(x,y)\leftarrow\min\{m_{u}(x,y),\allowbreak m_{u}(x,z)+m_{u}(z,y)\}$ and $m_{d}(x,y)\leftarrow\min\{m_{d}(x,y),\allowbreak m_{d}(x,z)+m_{d}(z,y)\}$.
Similarly for every upper triangle $\{x,y,z\}$ of $(x,y)$ the perfect customization sets $m_{u}(x,y)\leftarrow\min\{m_{u}(x,y),\allowbreak m_{u}(x,z)+m_{d}(z,y)\}$ and $m_{d}(x,y)\leftarrow\min\{m_{d}(x,y),\allowbreak m_{d}(x,z)+m_{u}(z,y)\}$.
The perfect witness search might need to remove an arc only in one direction.
It therefore produces, just as in the original CHs, two search graphs: an upward search graph and a downward search graph.
The forward search in the query phase is limited to the upward search graph and the backward search to the downward search graph, just as in the original CHs.
The arc $(x,y)$ is removed from the upward search graph if and only if an intermediate triangle $\{x,y,z\}$ with $m_{u}(x,y)=m_{u}(x,z)+m_{u}(z,y)$ exists or an upper triangle $\{x,y,z\}$ with $m_{u}(x,y)=m_{u}(x,z)+m_{d}(z,y)$ exists.
Analogously, the arc $(x,y)$ is removed from the downward search graph if and only if an intermediate triangle $\{x,y,z\}$ with $m_{d}(x,y)=m_{d}(x,z)+m_{d}(z,y)$ exists or an upper triangle $\{x,y,z\}$ with $m_{d}(x,y)=m_{d}(x,z)+m_{u}(z,y)$ exists.

\subsection{Single Instruction Multiple Data}

The weights attached to each arc in the CH can be replaced by an interleaved set of $k$ weights by storing for every arc a vector of $k$ elements.
Vectors allows us to customize all $k$ metrics in one go, amortizing triangle enumeration time and they allow us to use single instruction multiple data (SIMD) operations.
Further, as we use essentially two metrics to handle directed graphs, we can store both of them in a 2-dimensional vector.
This allows us to handle both directions in a single processor instruction.
Similarly, if we have $k$ directed input weights we can store them in a $2k$-dimensional vectors.

The processor needs to support component-wise minimum and saturated addition, \ie, $a+b=\mbox{int}_{\max}$ must hold in the case of an overflow.
In the case of directed graphs it additionally needs to support efficiently swapping neighboring vector components.
A current SSE-enabled processor supports all the necessary operations for 16-bit integer components. 
For 32-bit integer saturated addition is missing.
There are two possibilities to work around this limitation: 
The first is to emulate saturated-add using a combination of regular addition, comparison and blend/if-then-else instruction. 
The second consists of using 31-bit weights and use $2^{31}-1$ as value for $\infty$ instead of $2^{32}-1$. 
The algorithm only computes the saturated addition of two weights followed by taking the minimum of the result and some other weight, \ie, if computing $\min(a+b,c)$ for all weights $a$, $b$ and $c$ is unproblematic, then the algorithms works correctly. 
We know that $a$ and $b$ are at most $2^{31}-1$ and thus their sum is at most $2^{32}-2$ which fits into a 32-bit integer.
In the next step we know that $c$ is at most $2^{31}-1$ and thus the resulting minimum is also at most $2^{31}-1$.

\subsection{Partial Updates}
\label{sec:partial-updates}

Until now we have only considered computing metrics from scratch. 
However, in many scenarios this is overkill, as we know that only a few edge weights of the input graph were changed.
It is unnecessary to redo all computations in this case.
The ideas employed by our algorithm are somewhat similar to those presented in \cite{gssv-erlrn-12}, but our situation differs as we know that we do not have to insert or remove arcs. 
Denote by $U=\left\{ ((x_{i},y_{i}),\mynew_{i})\right\}$ the set of arcs whose weights should be updated, where $(x_{i},y_{i})$ is the arc ID and $\mynew_{i}$ the new weight. 
Note that modifying the weight of one arc can trigger further changes. 
However, these new changes have to be at higher levels. 
We therefore organize $U$ as a priority queue ordered by the level of $x_{i}$. 
We iteratively remove arcs from the queue and apply the change. 
If new changes are triggered we insert these into the queue.
The algorithm terminates once the queue is empty.

Denote by $(x,y)$ the arc that was removed from the queue and by $\mynew$ its new weight and by $\myold$ its old weight. 
We first have to check whether $\mynew$ can be bypassed using a lower triangle. 
For this reason, we iterate over all lower triangles $\{x,y,z\}$ of $(x,y)$ and perform $\mynew\leftarrow\min\{\mynew,m(z,x)+m(z,y)\}$.
Furthermore, if $\{x,y\}$ is an edge in the input graph $G$, we might have overwritten its weight with a shortcut weight, which after the update might not be shorter anymore. Hence, we additionally test that $\mynew$ is not larger than the input weight.
If after both checks $\mynew=m(x,y)$ holds, then no change is necessary and no further changes are triggered. 
If $\myold$ and $\mynew$ differ we iterate over all upper triangles $\{x,y,z\}$ of $(x,y)$ and test whether $m(x,z)+\myold=m(y,z)$ holds and if so the weight of the arc $(y,z)$ must be set to $m(x,z)+\mynew$.
We add this change to the queue. 
Analogously we iterate over all intermediate triangles $\{x,y,z\}$ of $(x,y)$ and queue up a change to $(z,y)$ if $m(x,z)+\myold=m(z,y)$ holds.

How many subsequent changes a single change triggers heavily depends on the metric and can significantly vary. 
Slightly changing the weight of a dirt road has near to no impact whereas changing a heavily used highway segment will trigger many changes. 
In the game setting such largely varying running times are undesirable as they lead to lag-peaks.
We propose to maintain a queue into which all changes are inserted.
Every round a fixed amount of time is spent processing elements from this queue. 
If time runs out before the queue is emptied the remaining arcs are processed in the next round. 
This way costs are amortized resulting in a constant workload per turn. 
The downside is that as long the queue is not empty some distance queries will use outdated data. 
How much time is spent each turn updating the metric determines how long an update needs to be propagated along the whole graph. 

\section{Distance Query}\label{sec:query}

In this section, we describe how to compute distance queries\footnote{We refer to the minimum sum over all weights over all $st$-paths as the distance between $s$ and $t$. The term ``distance query'' does not imply that we only consider shortest paths according to geographical distance.}, \ie, a shortest up-down path in $G_{\pi}^{\wedge}$ between two vertices $s$ and $t$ given a customized metric and how to unpack into a shortest path edge sequence in $G$.

\subsection{Basic}

The basic query runs two instances of Dijkstra's algorithm on $G_{\pi}^{\wedge}$
from $s$ and from $t$. If $G$ is undirected, then both searches
use the same metric. Otherwise if $G$ is directed the search from
$s$ uses the upward metric $m_{u}$ and the search from $t$ the
downward metric $m_{d}$. In either case in contrast to \cite{gssv-erlrn-12}
they operate on the same upward search graph $G_{\pi}^{\wedge}$.
Once the radius of one of the two searches is larger
than the shortest path found so far, we stop the search because we
know that no shorter path can exist. 
We alternate between processing vertices in the forward search and processing vertices in the backward search.

\subsection{Stalling}
\label{sec:stalling}

We implemented a basic version of an optimization presented in \cite{gssv-erlrn-12,ss-ehh-12} called stall-on-demand. 
The optimization exploits that the shortest strictly upward $sv$-path in $G_{\pi}^{\wedge}$ can be longer than the shortest $sv$-path in $G_{\pi}^{*}$, which can go up and down arbitrarily. 
The search from $s$ only finds upward paths and if we observe that an up-down path exists that is not longer, then we can prune the upward search. 
Denote by $x$ the vertex removed from the queue. 
We iterate over all outgoing arcs $(x,y)$ and test whether $d(x)\ge m(x,y) + d(y)$ holds.
If it holds for some arc we prune $x$ by not relaxing its outgoing arcs.

If $d(x) > m(x,y) + d(y)$ holds, then pruning is correct because all subpaths of shortest up-down paths must be shortest paths and the upward path ending at $x$ is not shortest path as a shorter up-down path through $y$ exists.
We can also prune when $d(x) \ge m(x,y) + d(y)$, but a different argument is needed.
To the best of our knowledge, correctness has so far not been proven for the $d(x) = m(x,y) + d(y)$ case.
Notice that we do not exploit any special properties of metric independent %
\begin{wrapfigure}{o}{4.5cm}%
\vspace{2em}
\begin{center}
\includegraphics[scale=1.5]{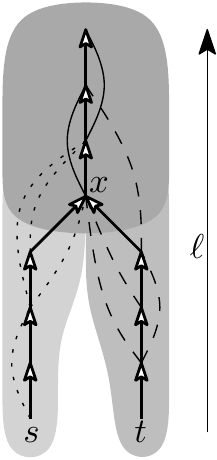}
\end{center}
\vspace{-1em}

\caption{The union of the darkgray and lightgray areas is the search space of~$s$. Analogously the union of the darkgray and middlegray areas is the search space of $t$. The darkgray area is the intersection of both search spaces.  The dotted arcs start in the search space of~$s$, but not in the search space of~$t$. Analogously the dashed arcs start in the search space of~$t$, but not in the search space of~$s$. The solid arcs start in the intersection of the two search spaces. The vertex~$x$ is the LCA of~$s$ and~$t$.
\vspace{-4em}
}
\label{fig:elim-query}
\end{wrapfigure}%
orders and thus our prove works for every CH.

\begin{theorem}
The upward search can be pruned when $d(x) \ge m(x,y) + d(y)$ holds.
\end{theorem}

\begin{proof}
We show that for every pair of vertices $s$ and $t$ an unprunable, shortest, up-down $st$-path exists.
Our proof relies on Lemma~\ref{lem:higher-path} which orders paths by height and states that $st$-path that are no up-down paths can be transformed into up-down paths that are no longer and strictly higher.
We know that some shortest $st$-path $K$ exists.
If $K$ is not pruned, then there is nothing to show.
If $K$ is pruned, then there exists a vertex $x$ on $K$ at which the search is pruned.
Without loose of generality we assume that $x$ lies on the upward part of $K$.
Further there must exist a vertex $y$ and a path $Q$ from $s$ to $x$ going through $y$ such that $Q$ is no longer than the $sx$-prefix of $K$.
Consider the path $R$ obtained by concatenating $Q$ with the $xt$-suffix of $K$.
$R$ is by construction not longer than $K$. 
If $x$ is the highest vertex on $K$ then $R$ is an up-down path and $R$ is strictly higher.
Otherwise, $R$ is no up-down path, but using Lemma~\ref{lem:higher-path} $R$ can be transformed into an up-down path that is strictly higher and no longer.
In both cases, $R$ is no longer and strictly higher.
Either, $R$ is unprunable or we apply the argument iteratively.
As there are only finitely many up-down paths and each iteration increases the height of $R$, we eventually end up at an unprunable, shortest, up-down $st$-path, which concludes the proof.
\end{proof}

\subsection{Elimination Tree}

We precompute for every vertex its parent's vertex ID in the elimination tree in a preprocessing step.
This allows us to efficiently enumerate all vertices in $\searchspace(s)$ and $\searchspace(t)$ at query time.
The vertices are enumerated increasing by rank.

We store two tentative distance arrays $d_{f}(v)$ and $d_{b}(v)$. 
Initially these are all set to~$\infty$. 
In a first step we compute the lowest common ancestor (LCA)~$x$ of~$s$ and~$t$ in the elimination tree. 
We do this by simultaneously enumerating all ancestors of~$s$ and~$t$ by increasing rank until a common ancestor is found. 
In a second step we iterate over all vertices~$y$ on the tree-path from~$s$ to~$x$ and relax all forward arcs of such $y$.
In a third step we do the same for all vertices~$y$ from~$t$ to~$x$ in the backward search. 
In a final fourth step we iterate over all vertices~$y$ from~$x$ to the root~$r$ and relax all forward and backward arcs. 
Further in the fourth step we also determine the vertex~$z$ that minimizes $d_{f}(z)+d_{b}(z)$.
A shortest up-down path must exist that goes through~$z$.
Knowing~$z$ is necessary to determine the shortest path distance and to compute the sequence of arcs that compose the shortest path.
In a fifth cleanup step we iterate over all vertices from $s$ and $t$ to the root~$r$ to reset all $d_{f}$ and $d_{b}$ to $\infty$. 
This fifth step avoids having to spend $O(n)$ running time to initialize all tentative distances to $\infty$ for each query.
Consider the situation depicted in Figure~\ref{fig:elim-query}. 
In the first step the algorithm determines $x$. 
In the second step it relaxes all dotted arcs and the tree arcs departing in the lightgray area. 
In the third step all dashed arcs and the tree arcs departing in the middlegray area and in the fourth step the solid arcs and the remaining tree arcs follow.

The elimination tree query can be combined with the perfect witness search.
Before pruning any arc, we compute the elimination tree. 
We then prune the arcs.
It is now possible that a vertex has an ancestor in the tree that is not in its pruned search space.
However, we can still guarantee that every vertex in the pruned search space is an ancestor and this is enough to prove the query correctness. 
To avoid relaxing the outgoing arcs of an ancestor outside of the search space, we prune vertices whose tentative distance $d_{f}(x)$ respectively $d_b(x)$ is $\infty$.

Contrary to the approaches based upon Dijkstra's algorithm the elimination tree query approach does not need a priority queue.
This leads to significantly less work per processed vertex. 
Unfortunately the query must always process all vertices in the search space. 
Luckily, our experiments show that for random queries with $s$ and $t$ sampled uniformly at random the query time ends up being lower for the elimination tree query.
If $s$ and $t$ are close in the original graph, \ie, not sampled uniformly at random, then the Dijkstra-based approaches win.

\subsection{Path Unpacking}

All shortest path queries presented only compute shortest up-down paths.
This is enough to determine the distance of a shortest path in the original graph. 
However, if the sequence of edges that form a shortest path should be computed, then the up-down path must be unpacked. 
The original CH of \cite{gssv-erlrn-12} unpacks an up-down path by storing for every arc $(x,y)$ the vertex $z$ of the lower triangle $\{x,y,z\}$ that caused the weight at $m(x,y)$.
This information depends on the metric and we want to avoid storing additional metric-dependent information. 
We therefore resort to a different strategy:
Denote by $p_{1}\ldots p_{k}$ the up-down path found by the query. 
As long as a lower triangle $\{p_{i},p_{i+1},x\}$ of an arc $(p_{i},p_{i+1})$ exists with $m(p_{i},p_{i+1})=m(x,p_{i})+m(x,p_{i+1})$, our algorithm inserts the vertex $x$ between $p_{i}$ and $p_{i+1}$ into the path.

\section{Experiments}\label{sec:experiments}

In this section we present our careful and extensive experimental evaluation of the algorithms introduced and described before.

\paragraph{Compiler and Machine}
We implemented our algorithms in C++, using g++ 4.7.1 with -O3 for compilation.
The customization and query experiments were run on a dual-CPU 8-core
Intel Xeon E5-2670 processor, which is based on the Sandy Bridge architecture, clocked at 2.6 GHz, with 64 GiB of DDR3-1600
RAM, 20 MiB of L3 and 256 KiB of L2 cache. The order computation experiments reported in Table~\ref{tab:order-time}
were run on a single core of an Intel Core i7-2600K CPU processor.

\paragraph{Instances}

\begin{table}
\tbl{Instances. We report the number of vertices and of directed arcs of the benchmark graphs. We further present the number of edges in the induced undirected graph. We also report the running time of Dijkstra's algorithm with stop criterion averaged over 10\,000 $st$-queries, where $s$ and $t$ are chosen uniformly at random.\label{tab:instances}}{%
\begin{tabular}{lrrrrr}
\toprule 
 & & & & & Dijstra Baseline\\
Instance & \#\,Vertices & \#\,Arcs & \#\,Edges & Symmetric? & Running Time [ms]\\
\midrule
Karlsruhe & 120\,412 & 302\,605 & 154\,869 & no & 6\\
TheFrozenSea & 754\,195 & 5\,815\,688 & 2\,907\,844 & yes & 58\\
Europe & 18\,010\,173 & 42\,188\,664 & 22\,211\,721 & no & 1\,560\\
\bottomrule
\end{tabular}
}%
\end{table}

\begin{wrapfigure}{o}{4.5cm}%
\begin{center}
\includegraphics[scale=0.2]{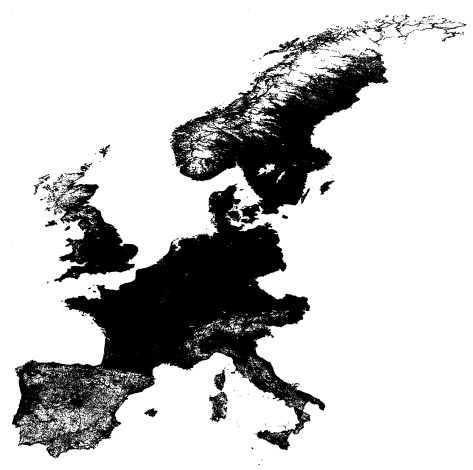}
\end{center}

\caption{All vertices in the DIMACS-Europe graph.}
\label{fig:europe}
\end{wrapfigure}

We evaluate three large instances of practical relevance in detail. 
In Section~\ref{sec:further-instances} we provide summarized experiments on further instances. 
The sizes of our main test instances are reported in Table~\ref{tab:instances}:
The DIMACS-Europe graph\footnote{Visit \url{http://i11www.iti.kit.edu/resources/roadgraphs.php} for details on how to acquire this graph.} was provided by PTV\footnote{\url{http://www.ptvgroup.com}} for the DIMACS challenge~\cite{dgj-spndi-09}.
The vertex positions are depicted in Figure~\ref{fig:europe}.
It is the standard benchmarking instance used by road routing papers over the past few years. 
Note that besides roads it also contains a few ferries to connect Great Britain and some other islands with the continent. 
The Europe graph analyzed here is its largest strongly connected component, which is a common method to remove bogus vertices.
The numbers in Table~\ref{tab:instances} are the numbers after computing the strongly connected component. 
The graph is directed, and we consider two different weights. 
The first weight is the travel time and the second weight is the straight line distance between two vertices on a perfect Earth sphere. 
Note that in the input data highways are often modeled using only a small number of vertices compared to the streets going through the cities. 
This differs from other data sources, such as OpenStreetMap\footnote{\url{http://www.openstreetmap.org}} that have a high number of vertices on highways to model road bends.
As demonstrated in Section~\ref{sec:osm-europe}, degree-2 vertices do not hamper the performance of CHs. 
The Karlsruhe graph is a subgraph of the PTV graph for a larger region around Karlsruhe.
We consider the largest connected component of the graph induced by
all vertices with a latitude between~48.3\textdegree{} and~49.2\textdegree{},
and a longitude between~8\textdegree{} and~9\textdegree{}. The TheFrozenSea
graph is based on the largest Star Craft map presented in \cite{s-bgbp-12}.
The map is composed of square tiles having at most eight neighbors
and distinguishes between walkable and non-walkable tiles. These are
not distributed uniformly, but rather form differently-sized pockets
of freely walkable space alternating with \emph{choke points} of very
limited walkable space. The corresponding graph contains for every
walkable tile a vertex and for every pair of adjacent walkable tiles
an edge. Diagonal edges are weighted by~$\sqrt{2}$, while horizontal
and vertical edges have weight~$1$. The graph is symmetric, \ie,
for each forward arc there is a backward arc, and contains large grid
subgraphs. 
For comparability with other works we report in Table~\ref{tab:instances} the time needed by Dijkstra's algorithm.
The variant of Dijkstra's algorithm used to compute the baseline running times did not reorder the vertices in-memory, used a 4-ary heap, was unidirectional, used the stopping criterion.

\subsection{Orders}
\label{sec:experiments:orders}

\begin{wraptable}{o}{7cm}%
\begin{center}%
\vspace{-2em}
\begin{tabular}{lrrr}
\toprule
Instance & MetDep & Metis & KaHIP\\
\midrule
Karlsruhe & 4.1 & 0.5 & $<$ 1\,532\\
TheFrozenSea & 1\,280.4 & 4.7 & $<$ 22\,828\\
Europe & 813.5 & 131.3 & $<$ 249\,082\\
\bottomrule
\end{tabular}
\end{center}
\caption{Orders. Duration of order computation in seconds.
No parallelization was used. KaHIP was parametrized for quality only, disregarding running time as it is metric-independent. We are certain that large speedups are possible. However, this is not the focus of this work. See the Future Works section~\ref{sec:future-work} for a discussion about how good orders can be quickly computed.\label{tab:order-time}}
\end{wraptable}

\begin{figure}
\begin{centering}
\subfigure[Karlsruhe]{\begin{centering}
\includegraphics[clip,scale=0.6]{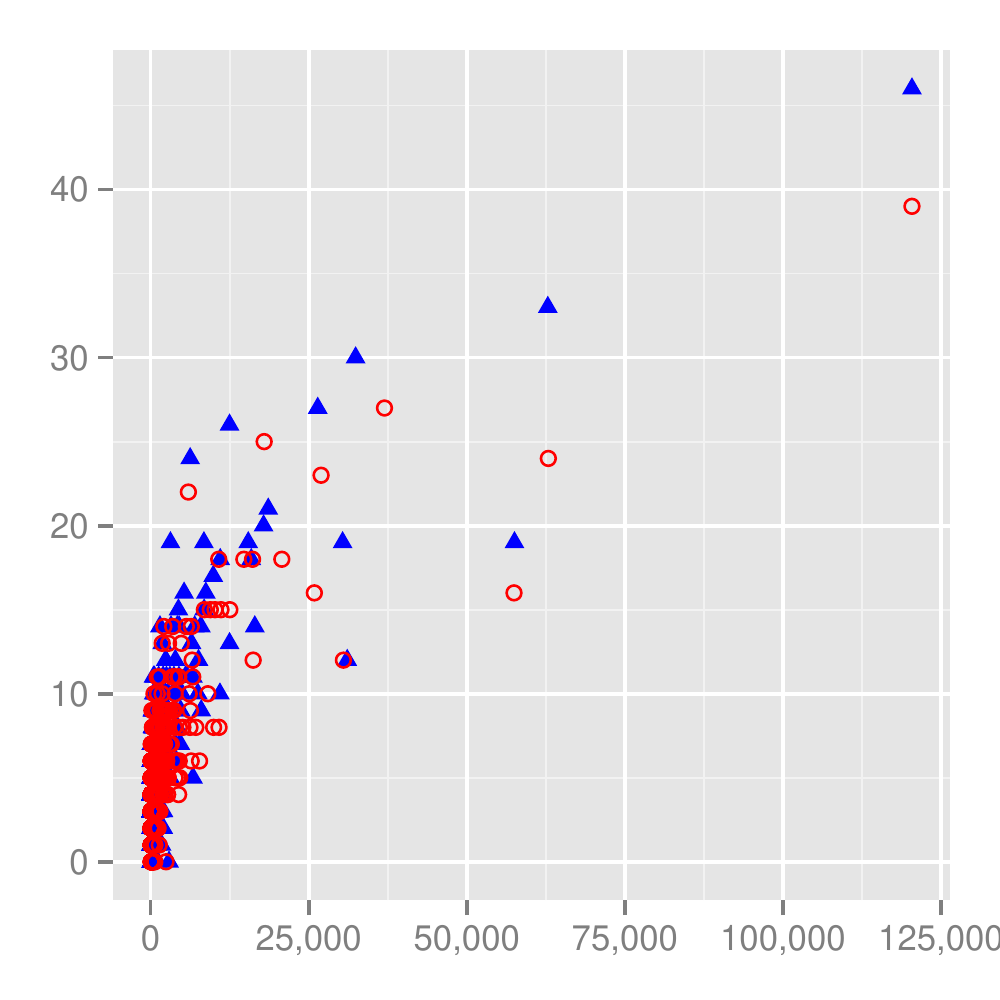}
\par\end{centering}

}\subfigure[Europe]{\centering{}\includegraphics[clip,scale=0.6]{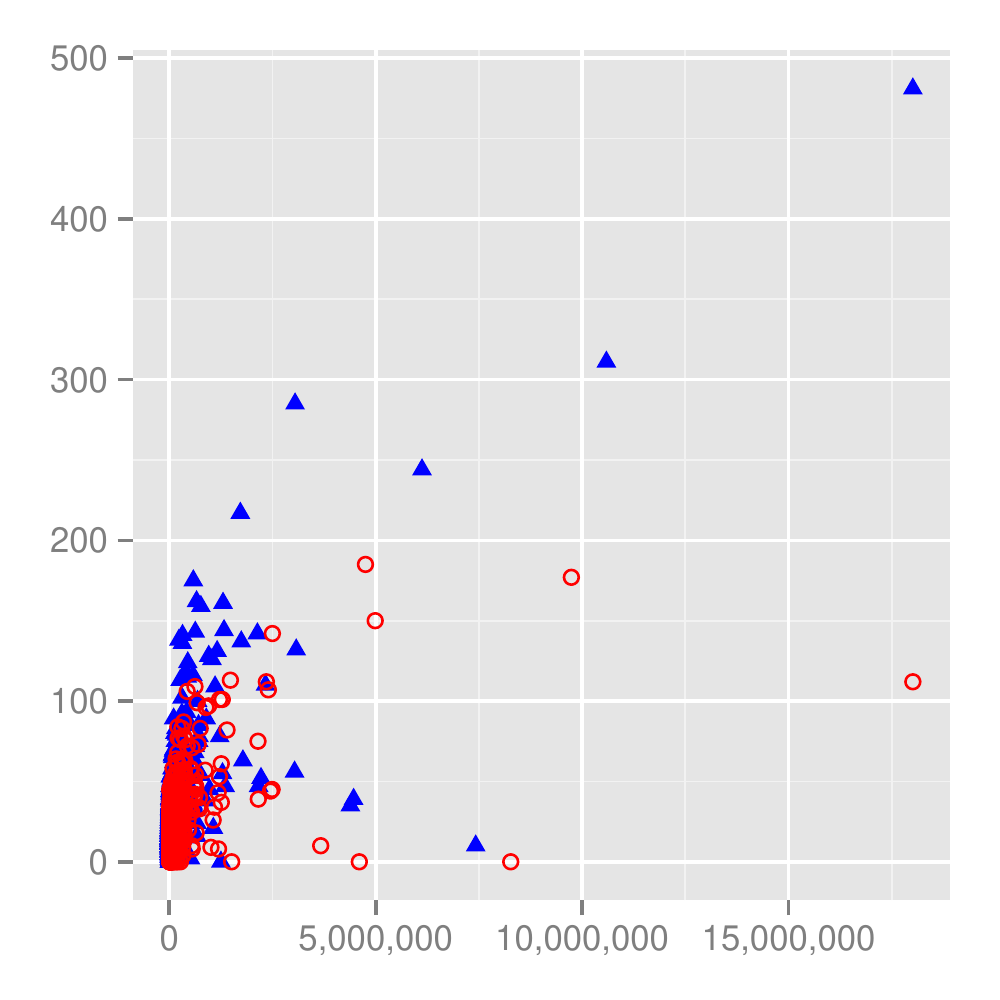}}
\par\end{centering}

\begin{centering}
\subfigure[TheFrozenSea]{\centering{}\includegraphics[clip,scale=0.6]{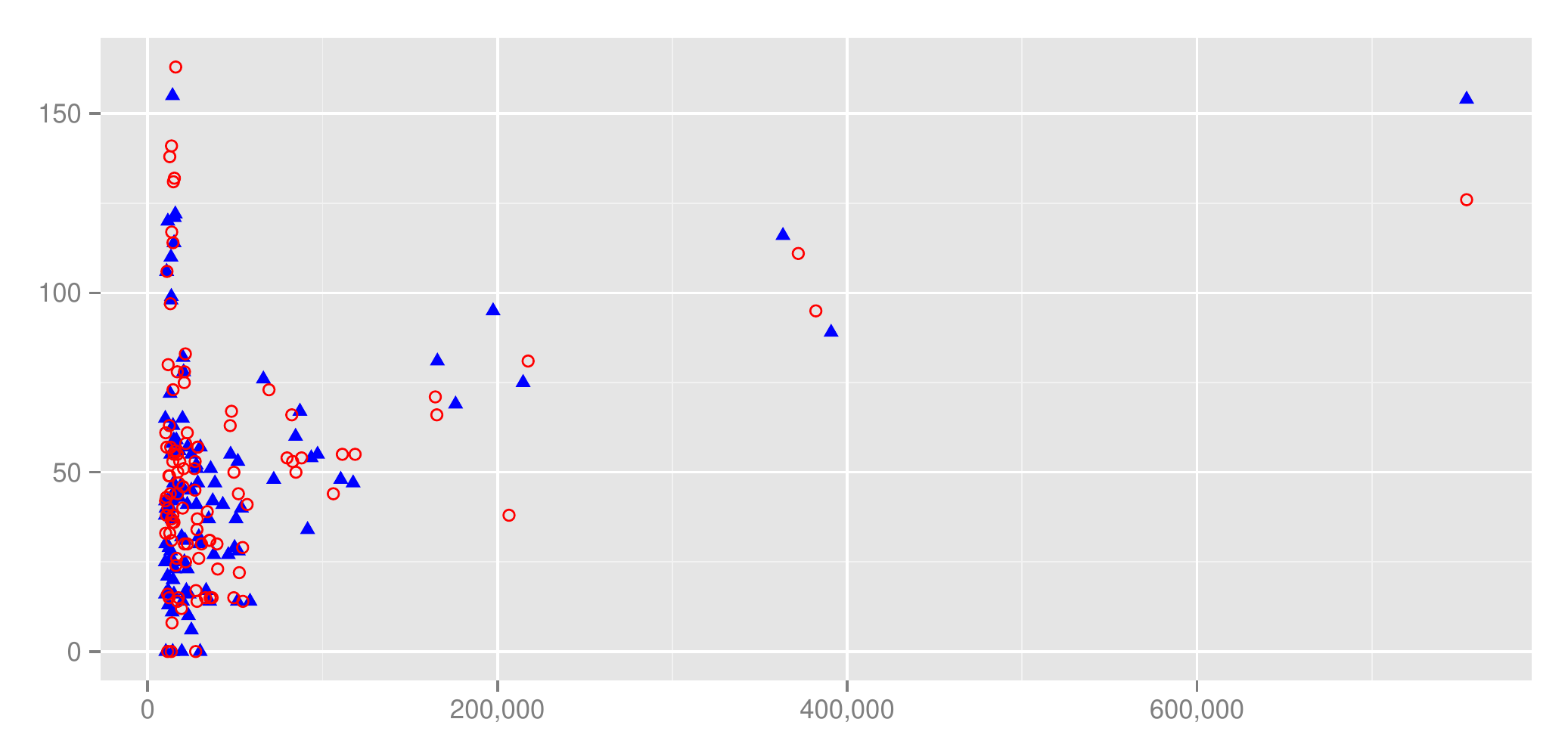}}
\par\end{centering}

\caption{The amount of vertices in the separator (vertical) vs the number of
vertices in the subgraph being bisected (horizontal). We only plot
the separators for (sub)graphs of at least 1000 vertices. The red hollow
circles is KaHIP and the blue filled triangles is Metis.}
\label{fig:separator-sizes}

\end{figure}

We analyze three different vertex orders: 
1)~The greedy metric-dependent order is an order in the spirit of \cite{gssv-erlrn-12}. We refer to it as ``MetDep'' in the tables.
2)~The Metis~5.0.1 graph partitioning package contains a tool called \texttt{ndmetis} to create ND-orders. 
3)~KaHIP~0.61 provides just graph partitioning tools. 
We therefore implemented a very basic nested dissection based program on top of it. 
We aimed at computing the best order possible.
Doing this fast is not the focus of this work.
For every graph we iteratively compute bisections with different random seeds using the ``strong'' configuration until for 10 consecutive runs no better cut is found. 
We recursively bisect the graph until the parts are too small for KaHIP to handle and assign the order arbitrarily in these small parts. 
We set the imbalance for KaHIP to 20\%. 
Note that our program is solely tuned for quality completely disregarding running time. 
It is certainly possible to trade much speed for a negligible or even no decrease in quality. 
We therefore report the running times as upper bounds, as no attempt was made to optimize them.
Please do not cite our reported running times to claim that KaHIP is slow.
This is no valid conclusion.

Table~\ref{tab:order-time} reports the times needed to compute the orders. Interestingly, Metis outperforms
even the metric-dependent greedy vertex ordering strategy. Figure~\ref{fig:separator-sizes} shows
the sizes of the computed separators. As expected, KaHIP results in
better quality. The road graphs seem to have separators following
a $\Theta(\sqrt[3]{n})$-law. On Karlsruhe the separator sizes steadily
decrease from the top level to the bottom level, making Theorem~\ref{thm:Approx}
directly applicable under the assumption that no significantly better separators exist. 
The KaHIP separators on the Europe graph have
a different structure on the top level. The separators first increase
before they get smaller. This is because of the special structure
of the European continent. For example the cut separating Great Britain and Spain
from France is far smaller than one would expect for a graph of that
size. In the next step KaHIP cuts Great Britain from Spain which results
in one of the extremely thin cuts observed in the plot. Interestingly
Metis is not able to find these cuts that exploit the continental
topology. The game map has a structure that differs from road graphs
as the plots have two peaks. This effect results from the large grid
subgraphs. The grids have $\Theta(\sqrt{n})$ separators, whereas
at the higher levels the choke points results in separators that approximately
follow a $\Theta(\sqrt[3]{n})$-law. At some point the bisector
has cut all choke points and has to start cutting through the grids.
The second peak is at the point where this switch happens.

\subsection{CH Construction}

\begin{wraptable}{o}{8.5cm}
\vspace{-2em}
\begin{center}
\begin{tabular}{lrr}
\toprule
Instance & Dyn. Adj. Array & Contraction Graph\\
\midrule
Karlsruhe & 0.6 & $<$0.1\\
TheFrozenSea & 490.6 & 3.8\\
Europe & 305.8 & 15.5\\
\bottomrule
\end{tabular}
\end{center}
\caption{Construction of the Contraction Hierarchy. We report the time in seconds required to compute the arcs in $G_{\pi}^{\wedge}$ given
a KaHIP ND-order $\pi$. No witness search is performed. No weights are assigned.\label{tab:ch-construction}}
\end{wraptable}

Table \ref{tab:ch-construction} compares the performance of our specialized
Contraction Graph data structure, described in Section~\ref{sec:ch-construction},
to the dynamic adjacency structure, as used in~\cite{gssv-erlrn-12} to
compute undirected and unweighted CHs. We do not report numbers for
the hash-based approach of~\cite{z-wchw-13} as it is fully dominated.
Our data structure dramatically improves performance. However to be fair, our approach
cannot immediately be extended to directed or weighted graphs.
Fortunately, this is no problem as we can introduced weights and directions during the customization phase.

\subsection{CH Size}

\begin{table}
\tbl{
Size of the Contraction Hierarchies for different instances and orders. 
We report the average number of vertices and arcs reachable in the upward search space of a vertex.
This number varies depending on whether a witness search is performed or not.
It also varies depending on whether we follow one-way streets in both direction or not.
We also report the number of triangles. 
As an indication for query performance, we report the average search space size in vertices and arcs, by sampling the search space of 1000 random vertices. 
Metis and KaHIP orders are metric-independent.
We report resulting figures after applying different variants of witness search. A heuristic witness search is one 
that exploits the metric in the preprocessing phase.
A perfect witness search is described in Section~\ref{sec:customization}.
\label{tab:ch-sizes}}{%
 \begin{tabular}{@{}c@{\hspace{2ex}}l@{\hspace{2ex}}l@{\hspace{2ex}}r@{\hspace{2ex}}rrr@{\hspace{1ex}}rr@{\hspace{1ex}}r@{}}
\toprule
 &  &  &  &  & & \multicolumn{4}{c}{Average upward search space size}\\
\cmidrule(lr){7-10}
 &  & \multirow{2}{*}[-5pt]{\shortstack{Witness\\search}} &  \multicolumn{2}{c}{\#\,Arcs $[\cdot10^3]$} & \#\,Triangles & \multicolumn{2}{c}{unweighted} & \multicolumn{2}{c}{weighted}\\
  \cmidrule(r){4-5}\cmidrule(lr){7-8}\cmidrule(lr){9-10}
 & Order & & undir. & upward & $[\cdot10^3]$ & \#\,Vertices  & \#\,Arcs  & \#\,Vertices & \#\,Arcs \\
\midrule
\multirow{7}{*}{\begin{sideways}
Karlsruhe~
\end{sideways}} & \multirow{3}{*}{MetDep} 
    & none & 21\,926 & 17\,661 & 37\,439\,858 & 5\,870 & 15\,786\,622 & 5\,246 & 11\,281\,564\\
 &  & heuristic & --- & 244 & --- & --- & --- & 108 & 503\\
 &  & perfect & --- & 239 & --- & --- & --- & 107 & 498\\
\addlinespace 
 & \multirow{2}{*}{Metis} 
    & none & 478 & 463 & 2\,590 & 164 & 6\,579 & 163 & 6\,411\\
 &  & perfect & --- & 340 & --- & --- & --- & 152 & 2\,903\\
\addlinespace 
 & \multirow{2}{*}{KaHIP} 
    & none & 528 & 511 & 2\,207 & 143 & 4\,723 & 142 & 4\,544\\
 &  & perfect & --- & 400 & --- & --- & --- & 136 & 2\,218\\
\midrule
\multirow{5}{*}{\begin{sideways}
TheFrozenSea\,
\end{sideways}} & \multirow{1}{*}{MetDep} & heuristic & --- & 6\,400 & --- & --- & --- & 1\,281 & 13\,330\\
\addlinespace
& \multirow{2}{*}{Metis} & none & 21\,067 & 21\,067 & 601\,846 & 676 & 92\,144 & 676 & 92\,144\\
 &  & perfect & --- & 10\,296 & --- & --- & --- & 644 & 32\,106\\
\addlinespace 
 & \multirow{2}{*}{KaHIP} & none & 25\,100 & 25\,100 & 864\,041 & 674 & 89\,567 & 674 & 89\,567\\
 &  & perfect & --- & 10\,162 & --- & --- & --- & 645 & 24\,782\\
\midrule
\multirow{5}{*}{\begin{sideways}
Europe~
\end{sideways}} & \multirow{1}{*}{MetDep} & heuristic & --- & 33\,912 & --- & --- & --- & 709 & 4\,808\\
 \addlinespace 
 & \multirow{2}{*}{Metis} & none & 70\,070 & 65\,546 & 1\,409\,250 & 1\,291 & 464\,956 & 1\,289 & 453\,366\\
 &  & perfect & --- & 47\,783 & --- & --- & --- & 1\,182 & 127\,588\\
\addlinespace 
 & \multirow{2}{*}{KaHIP} & none & 73\,920 & 69\,040 & 578\,248 & 652 & 117\,406 & 651 & 108\,121\\
 &  & perfect & --- & 55\,657 & --- & --- & --- & 616 & 44\,677\\
\bottomrule
\end{tabular}
}%
\end{table}

\begin{figure}
\begin{centering}
\subfigure[Karlsruhe]{\begin{centering}
\includegraphics[scale=0.6]{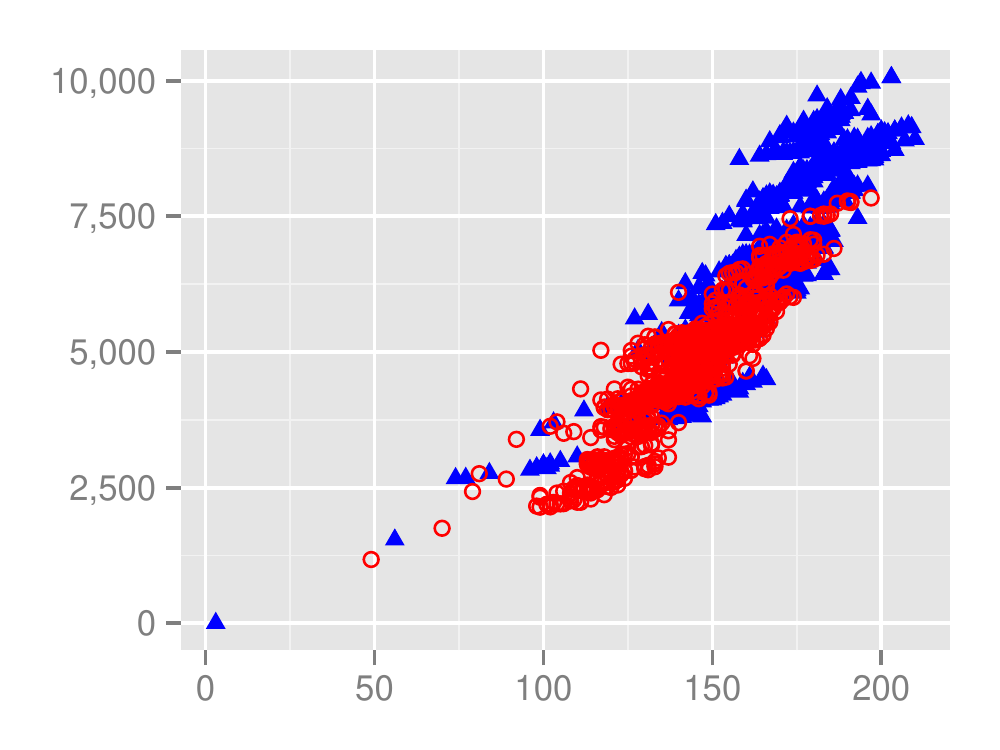}
\par\end{centering}

}~\subfigure[TheFrozenSea]{\begin{centering}
\includegraphics[scale=0.6]{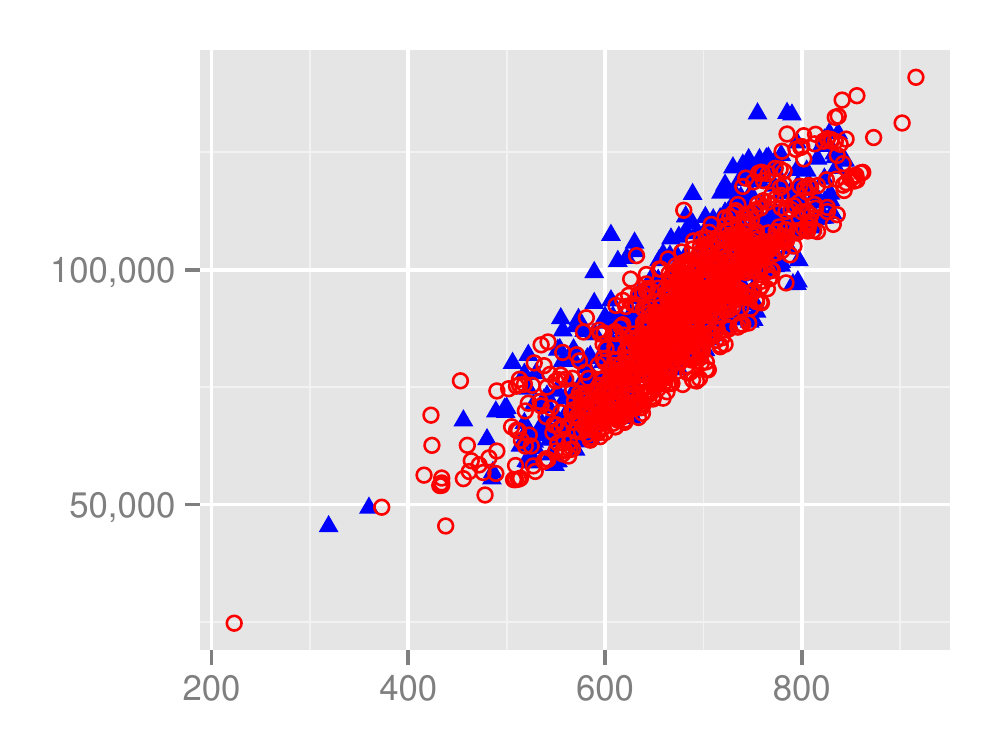}
\par\end{centering}

}
\par\end{centering}

\begin{centering}
\subfigure[Europe]{\begin{centering}
\includegraphics[scale=0.6]{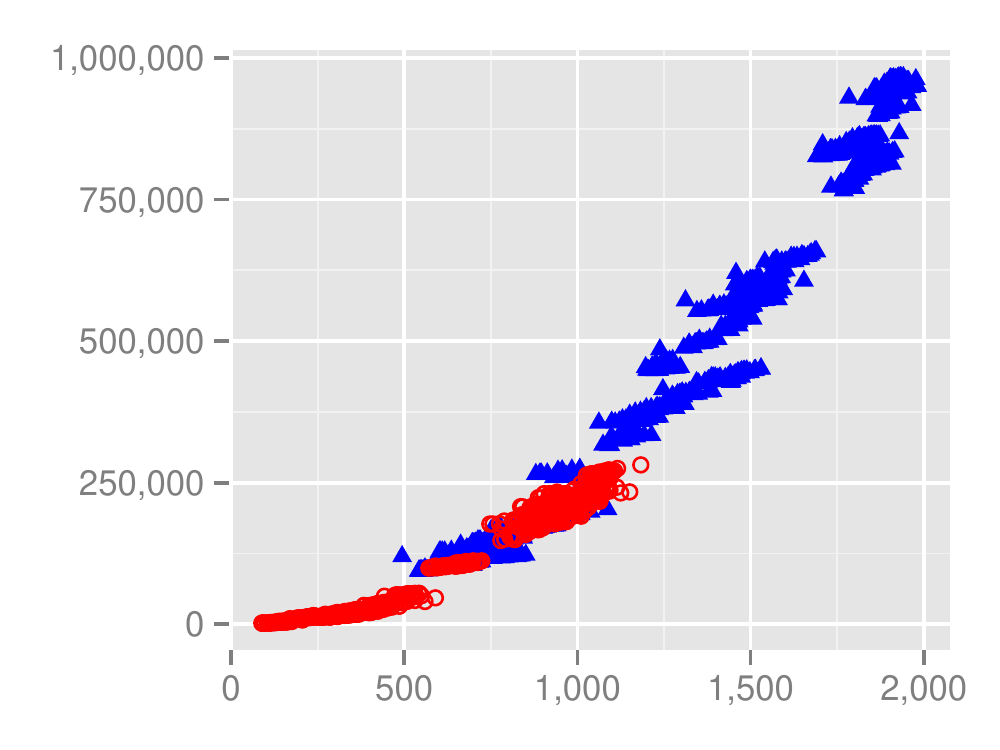}
\par\end{centering}

}
\par\end{centering}

\caption{The number of vertices (horizontal) vs the number of arcs (vertical)
in the search space of 1000 random vertices. The red hollow circles
is KaHIP and the blue filled triangles is Metis.}
\label{fig:ss-size}
\end{figure}

\begin{table}
\tbl{Elimination tree characteristics. Note that unlike in Table~\ref{tab:ch-sizes},
these values are exact and not sampled over a random subset of vertices. We also report upper bounds on the treewidth of the input graphs, after dropping the directions of arcs.
\label{tab:elimination-tree-and-treewidth}}{%
\begin{tabular}{llrrrrr}
\toprule
 &  & \multicolumn{2}{c}{\#\,Children} & \multicolumn{2}{c}{Height} & \multirow{2}{*}[-5pt]{\shortstack{Upper bound\\of Treewidth}}\\
 \cmidrule(lr){3-4}\cmidrule(lr){5-6}
Instance & Order  & avg. & max. & avg. & max.\\
\midrule
\multirow{2}{*}{Karlsruhe} & Metis & 1 & 5 & 163.48 & 211 & 92\\
 & KaHIP & 1 & 5 & 142.19 & 201 & 72\\
\addlinespace
\multirow{2}{*}{TheFrozenSea} & Metis & 1 & 3 & 675.61 & 858 & 282\\
 & KaHIP & 1 & 3 & 676.71 & 949&287\\
\addlinespace
\multirow{2}{*}{Europe} & Metis & 1 & 8 & 1283.45 & 2017 &876\\
 & KaHIP & 1 & 7 & 654.07 & 1232 & 479\\
\bottomrule
\end{tabular}
}%
\end{table}

\begin{table}
\tbl{Detailed analysis of the size of CHs. We evaluate uniform, random and distance weights on the Karlsruhe input graph.
Random weights are sampled from $[0,10000]$. The distance weight is
the straight distance along a perfect Earth sphere's surface. All
weights respect one-way streets of the input graph.\label{tab:ch-sizes-alternative-metric}}{%
\begin{tabular}{ccllrrr}
\toprule
 &  &  & \multirow{2}{*}[-5pt]{\shortstack{Witness\\search}} &  & \multicolumn{2}{c}{Avg. weighted upw. search space}\\
 \cmidrule(lr){6-7}
Instance  & Metric  & Order &  & \#\,Upward arcs & \#\,Vertices  & \#\,Arcs \\
\midrule
\multirow{15}{*}{Karlsruhe} & \multirow{5}{*}{Distance} 
 & \multirow{3}{*}{MetDep} %
 & none & 8\,000\,880 & 3\,276 & 4\,797\,224\\
 &  &  & heuristic & 295\,759 & 283 & 2\,881\\
 &  &  & perfect & 295\,684 & 281 & 2\,873\\
 \addlinespace
 &  & Metis & perfect & 382\,905 & 159 & 3\,641\\
 &  & KaHIP & perfect & 441\,998 & 141 & 2\,983\\
\cmidrule(l){2-7}
 & \multirow{5}{*}{Uniform} 
& \multirow{3}{*}{MetDep} %
& none & 5\,705\,168 & 2\,887 & 3\,602\,407\\
 &  &  & heuristic & 272\,711 & 151 & 808\\
 &  &  & perfect & 272\,711 & 151 & 808\\
\addlinespace
 &  & Metis & perfect & 363\,310 & 153 & 2\,638\\
 &  & KaHIP & perfect & 426\,145 & 136 & 2\,041\\
\cmidrule(l){2-7}
 & \multirow{5}{*}{Random} 
& \multirow{3}{*}{MetDep} %
& none & 6\,417\,960 & 3\,169 & 4\,257\,212\\
 &  &  & heuristic & 280\,024 & 160 & 949\\
&  &  & perfect & 276\,742 & 160 & 948\\
\addlinespace
&  & Metis & perfect & 361\,964 & 154 & 2\,800\\
&  & KaHIP & perfect & 424\,999 & 138 & 2\,093\\
\midrule
\multirow{3}{*}{Europe} & \multirow{3}{*}{Distance} & MetDep & heuristic & 39\,886\,688 & 4\,661 & 133\,151\\
 &  & Metis & perfect & 53\,505\,231 & 1\,257 & 178\,848\\
 &  & KaHIP & perfect & 60\,692\,639 & 644 & 62\,014\\
\bottomrule
\end{tabular}
}%
\end{table}

In Table~\ref{tab:ch-sizes} we report the resulting CH sizes for
various approaches. Computing a CH on Europe\emph{ without witness
search} with the metric-dependent order is infeasible even using the Contraction
Graph data structure. This is even true if we only want to count the number of arcs:
We aborted calculations after several days. We can however say with
certainty that there are at least $1.3\times10^{12}$ arcs in the
CH and the maximum upward vertex degree is at least $1.4\times10^{6}$.
As the original graph has only $4.2\times10^{7}$ arcs, it is safe
to assume that using this order it is impossible to achieve a speedup
compared to Dijkstra's algorithm on the input graph. However, on the
Karlsruhe graph we can actually compute the CH without witness search
and perform a perfect witness search. The numbers show that the heuristic
witness search employed by \cite{gssv-erlrn-12} is nearly optimal.
Furthermore, the numbers clearly show that using metric-dependent orders in
a metric-independent setting, \ie, without witness search, results
in unpractical CH sizes. However, they also show that a metric-dependent order
exploiting the weight structure dominates ND-orders. In Figure~\ref{fig:ss-size} we plot the number
of arcs in the search space vs the number of vertices. The plots show
that the KaHIP order significantly outperforms the Metis order on
the road graphs whereas the situation is a lot less clear on the game
map where the plots suggest nearly a tie.
KaHIP only slightly outperforms Metis, when using a perfect customization.
 Table~\ref{tab:elimination-tree-and-treewidth}
examines the elimination tree. Note that the height of the elimination
tree corresponds\footnote{The numbers in Table~\ref{tab:ch-sizes} and Table~\ref{tab:elimination-tree-and-treewidth} deviate a little because the search spaces in the former table are sampled while in the latter we compute precise values.} to the number of vertices in the (undirected) search
space. As the ratio between the maximum and the average height is
only about 2, we know that no special vertex exists that has a search
space significantly differing from the numbers shown in Table~\ref{tab:elimination-tree-and-treewidth}.
The elimination tree has a relatively small height compared to the
number of vertices in~$G$.

The treewidth of a graph is a measure widely used in theoretical computer
science and thus interesting on its own. The notion of treewidth
is deeply coupled with the notion of chordal super graphs and vertex
separators. See~\cite{bk-tiu-10} for details. The authors show in their Theorem~6
that the maximum upward degree~$d_{u}(v)$ over all vertices~$v$
in~$G_{\pi}^{\wedge}$ is an upper bound to the treewidth of a graph~$G$.
This theorem yields a straightforward algorithm that gives us the
upper bounds presented in Table~\ref{tab:elimination-tree-and-treewidth}. 

Interestingly these numbers correlate with our other findings: The
difference between the bounds on the road graphs reflect that the
KaHIP orders are better than Metis orders. On the game map there is
nearly no difference between Metis and KaHIP, which is in accordance with all
other performance indicators. The fact that the treewidth grows with
the graph size reflects that the running times are not independent
of the graph size. These numbers strongly suggest that road graphs
are not part of a graph class of constant treewidth. However, fortunately,
the treewidth grows sub-linearly. Our findings from Figure~\ref{fig:separator-sizes}
suggest that assuming a $O(\sqrt[3]{n})$ treewidth for road graphs
of $n$ vertices might come close to reality. 

In Table~\ref{tab:ch-sizes-alternative-metric} we evaluate the witness
search performances for different metrics. It turns out that the distance
metric is the most difficult one of the tested metrics. That the distance
metric is more difficult than the travel time metric is well known.
However it surprised us, that uniform and random metrics are easier
than the distance metric. We suppose that the random metric contains
a few very long arcs that are nearly never used. These could just
as well be removed from the graph resulting in a thinner graph with
nearly the same shortest path structure. The CH of a thinner graph
with a similar shortest path structure naturally has a smaller size. To explain
why the uniform metric behaves more similar to the travel time metric
than to the distance metric we have to realize that highways do not have many
degree~2 vertices in the input graph. Note that for different data sources this assumption might not hold. Highways
are therefore also preferred by the uniform metric. We expect
an instance with more degree-2 vertices on highways to behave differently. Interestingly the heuristic
witness search is perfect for a uniform metric. We expect this effect
to disappear on larger graphs. 

\begin{figure}
\begin{centering}
\subfigure[Karlsruhe/KaHIP]{\begin{centering}
\includegraphics[scale=0.6]{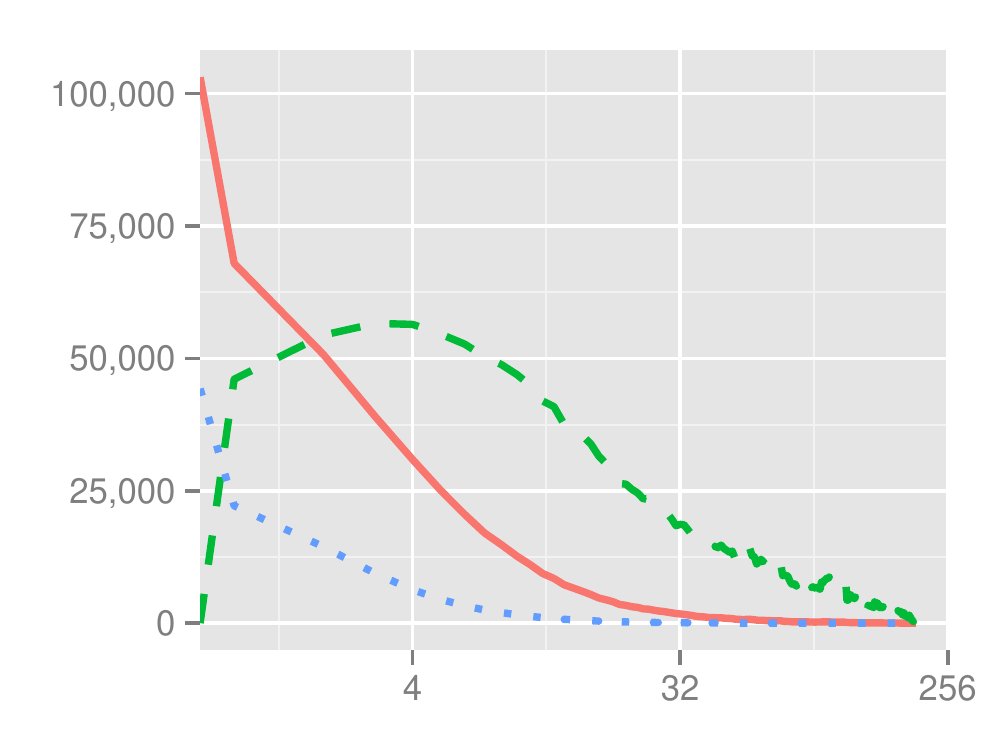}
\par\end{centering}

}~\subfigure[Karlsruhe/Metis]{\begin{centering}
\includegraphics[scale=0.6]{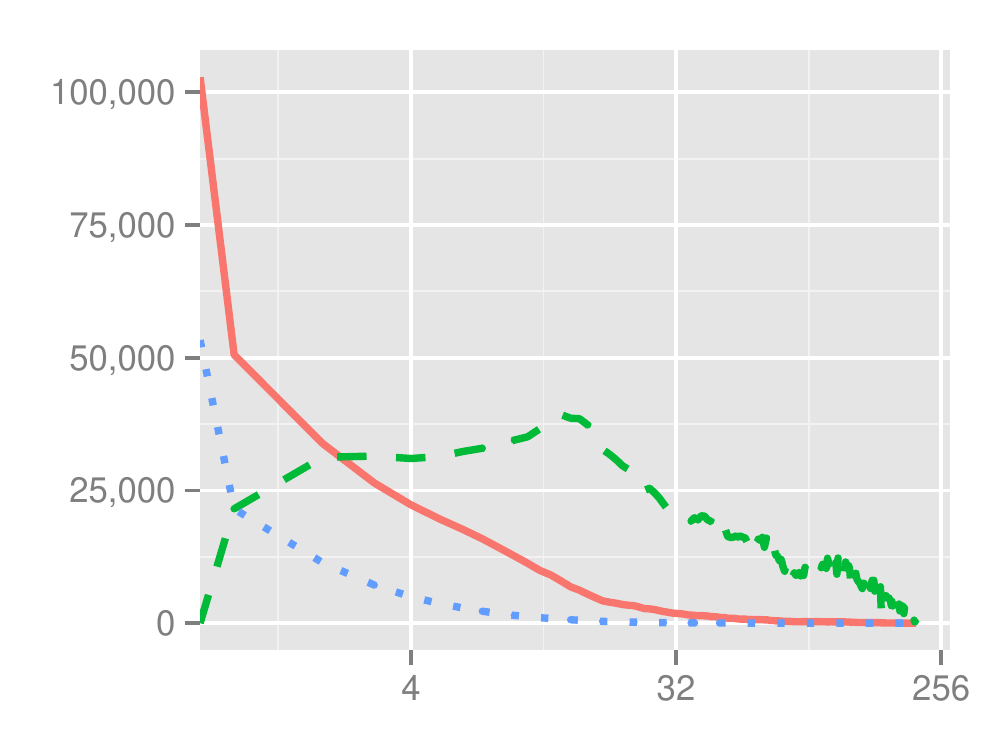}
\par\end{centering}

}
\par\end{centering}

\begin{centering}
\subfigure[TheFrozenSea/KaHIP]{\begin{centering}
\includegraphics[scale=0.6]{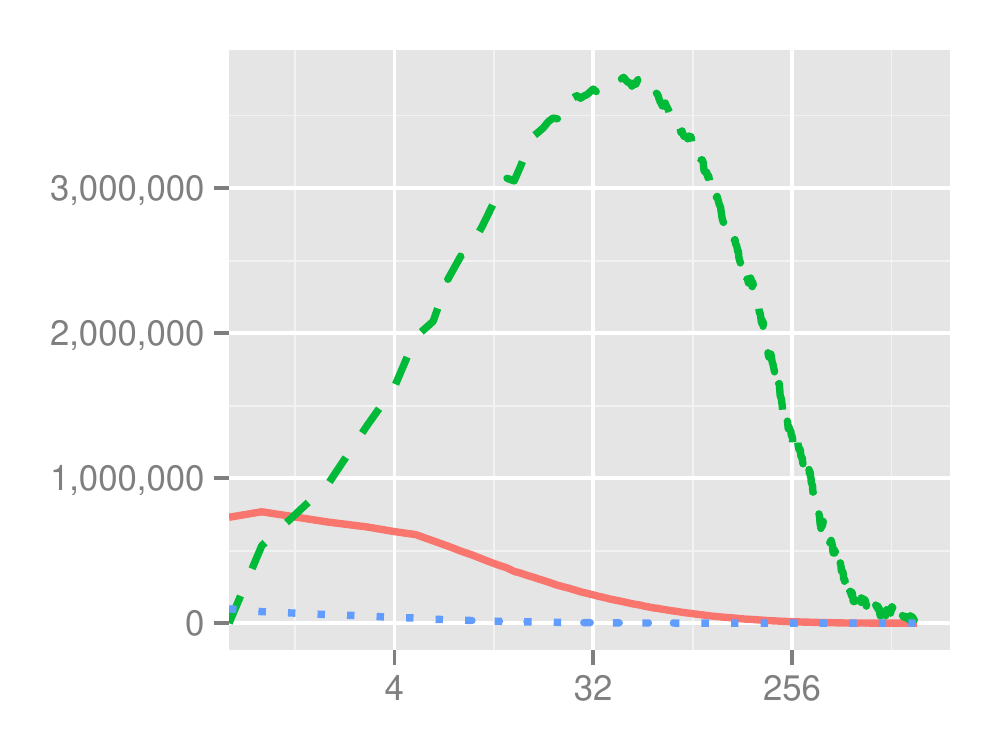}
\par\end{centering}

}~\subfigure[TheFrozenSea/Metis]{\begin{centering}
\includegraphics[scale=0.6]{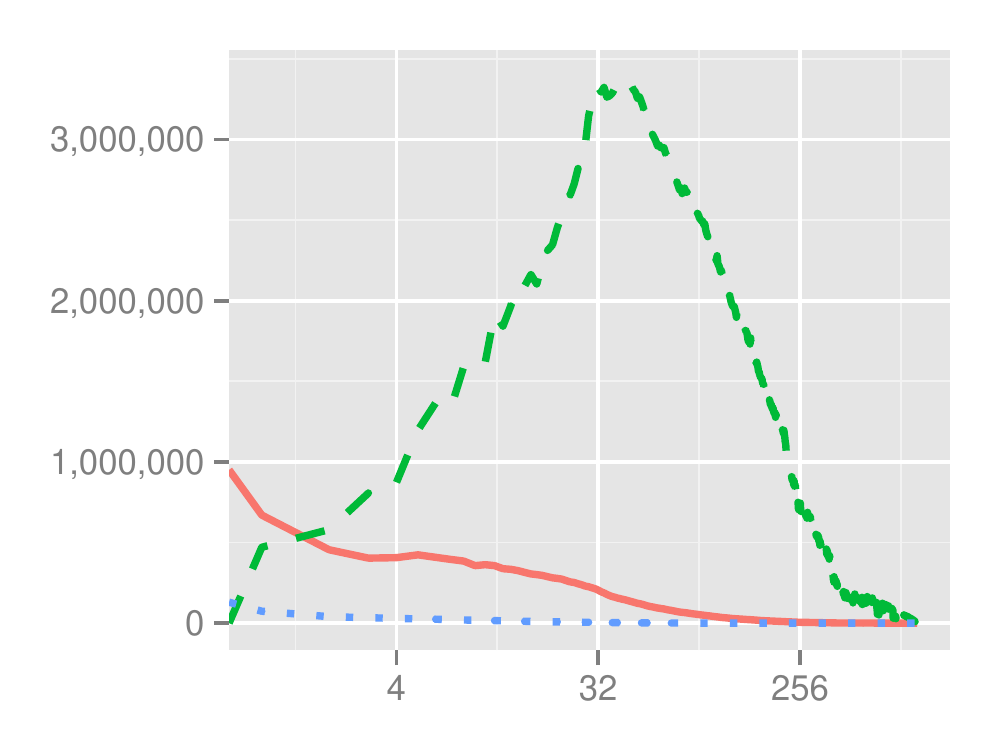}
\par\end{centering}

}
\par\end{centering}

\begin{centering}
\subfigure[Europe/KaHIP]{\begin{centering}
\includegraphics[scale=0.6]{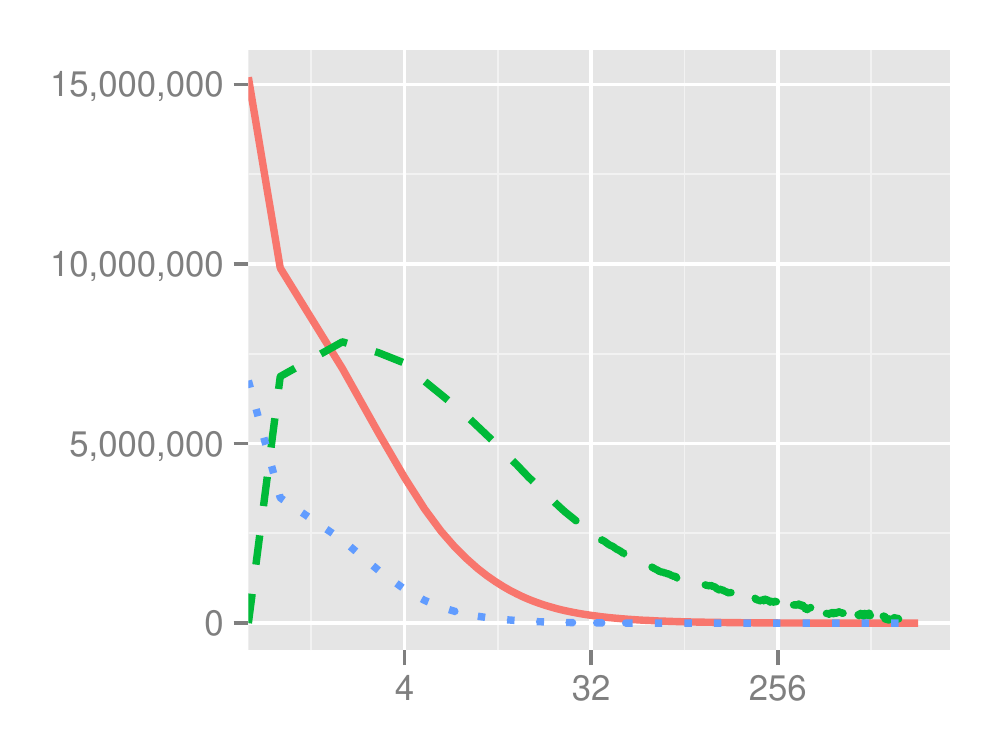}
\par\end{centering}

}~\subfigure[Europe/Metis]{\begin{centering}
\includegraphics[scale=0.6]{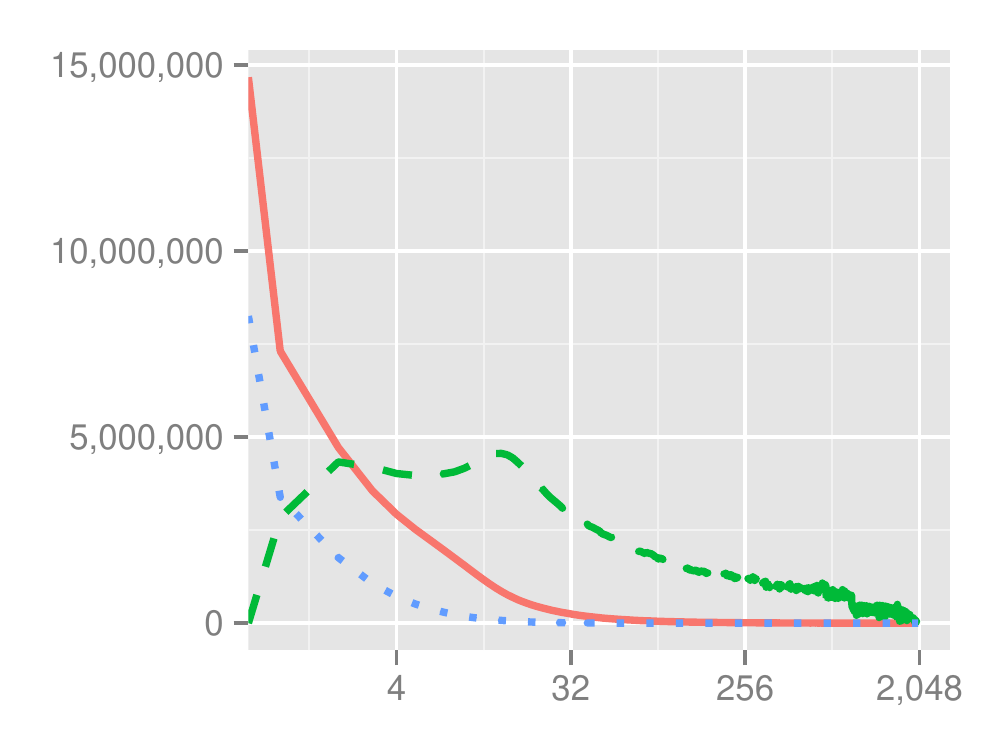}
\par\end{centering}

}
\par\end{centering}

\caption{The number of vertices~(y-axis) per level~(x-axis) is the blue dotted line. The number of arcs departing
in each level is the red solid line, and the number of lower triangles in each level is the green
dashed line. Warning: In contrast to Figure~\ref{fig:level-time}
these figures have a logarithmic $x$-scale.}
\label{fig:level-sizes}

\end{figure}

Recall that a CH is a DAG, and in DAGs each vertex can be assigned
a level. If a vertex can be placed in several levels we put it in
the lowest level. Figure~\ref{fig:level-sizes} illustrates the amount
of vertices and arcs in each level of a CH. The many highly ranked
extremely thin levels are a result of the top level separator clique:
Inside a clique every vertex must be on its own level. A few big separators
therefore significantly increase the level count.

\subsection{Triangle Enumeration}

\begin{figure}[h!]
\begin{centering}
\subfigure[Karlsruhe/KaHIP]{\begin{centering}
\includegraphics[scale=0.6]{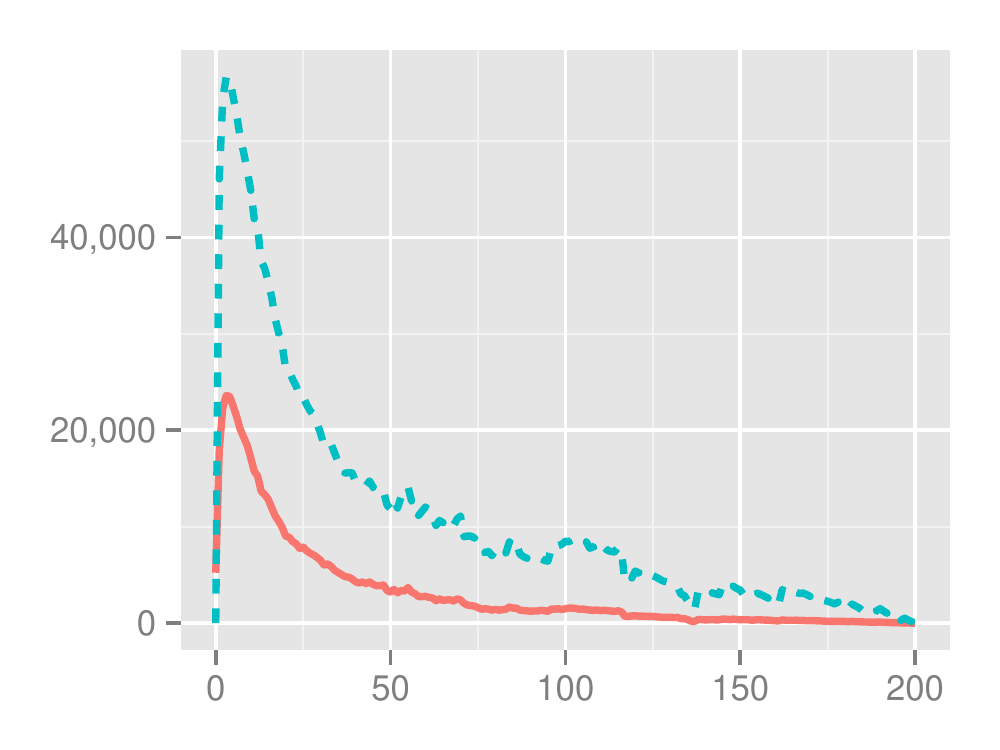}
\par\end{centering}

}~\subfigure[Karlsruhe/Metis]{\begin{centering}
\includegraphics[scale=0.6]{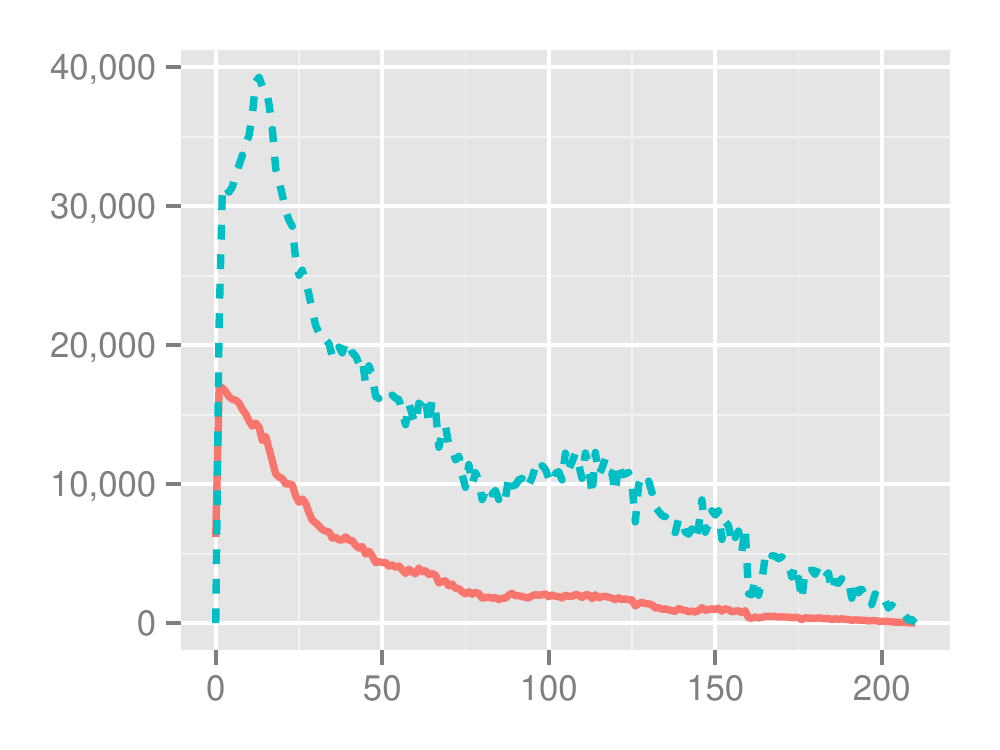}
\par\end{centering}

}
\par\end{centering}

\begin{centering}
\subfigure[TheFrozenSea/KaHIP]{\begin{centering}
\includegraphics[scale=0.6]{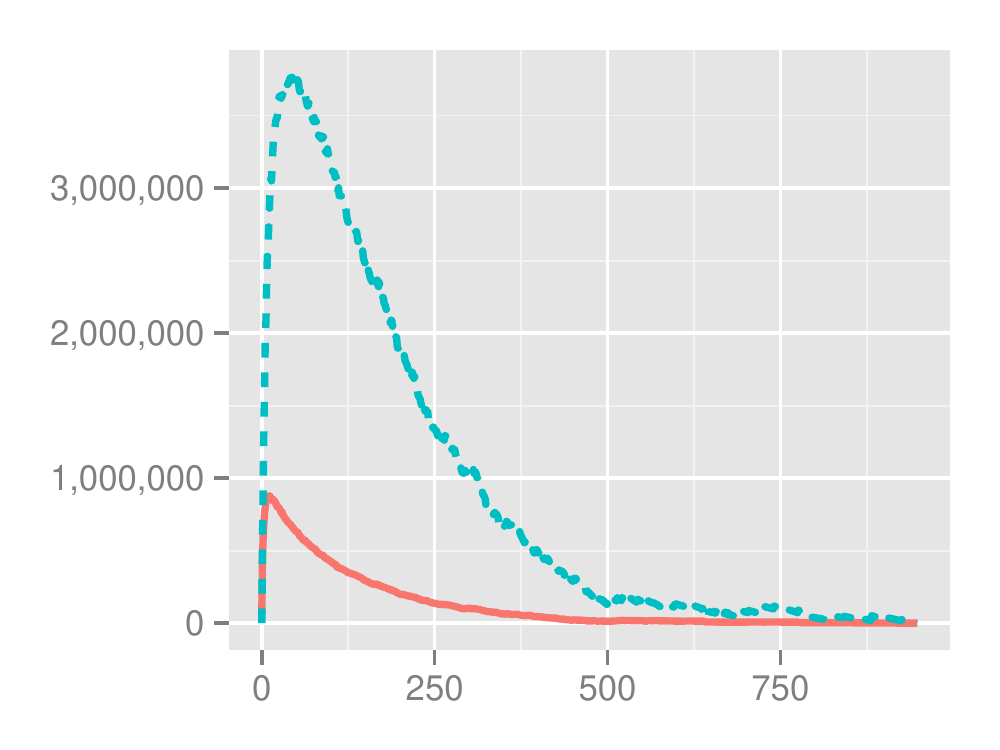}
\par\end{centering}

}~\subfigure[TheFrozenSea/Metis]{\begin{centering}
\includegraphics[scale=0.6]{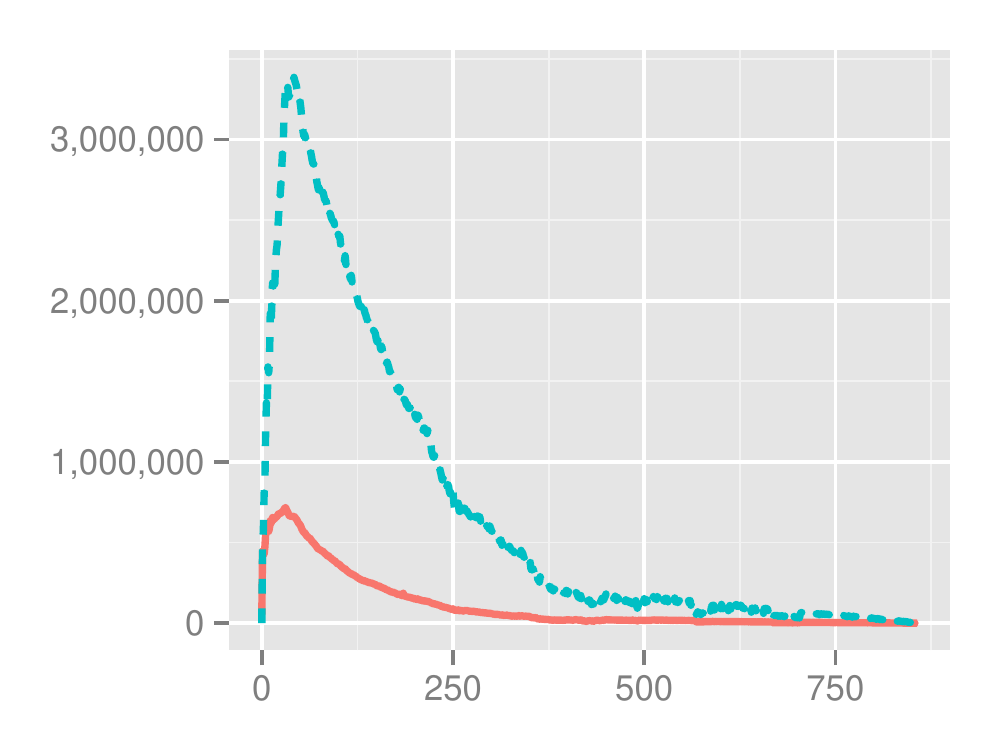}
\par\end{centering}

}
\par\end{centering}

\begin{centering}
\subfigure[Europe/KaHIP]{\begin{centering}
\includegraphics[scale=0.6]{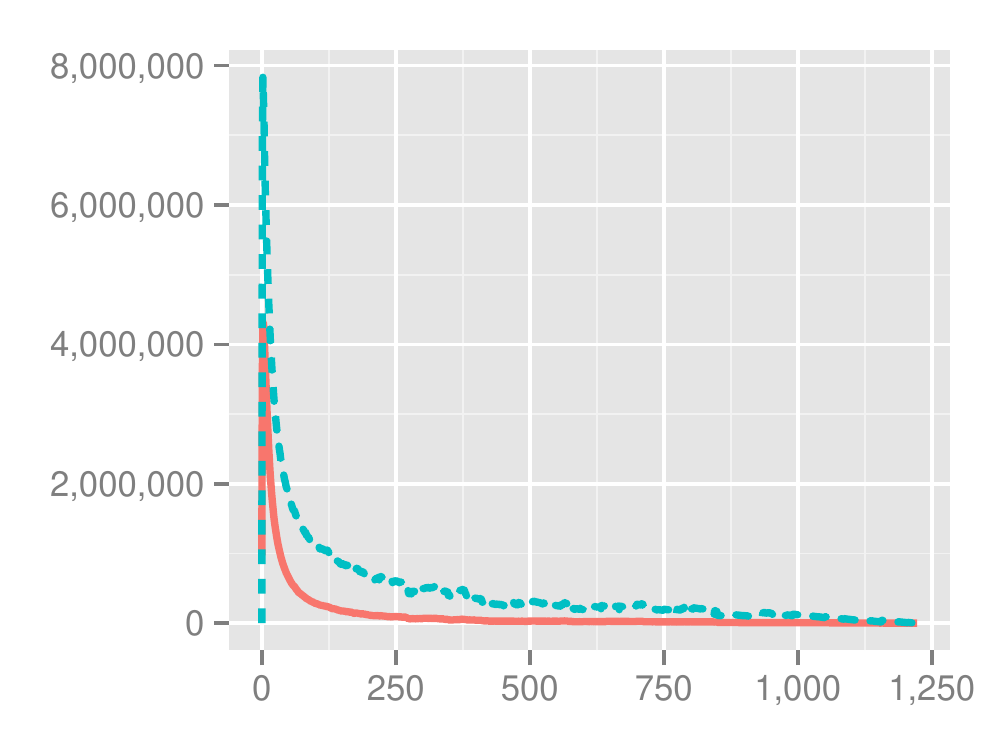}
\par\end{centering}

}~\subfigure[Europe/Metis]{\begin{centering}
\includegraphics[scale=0.6]{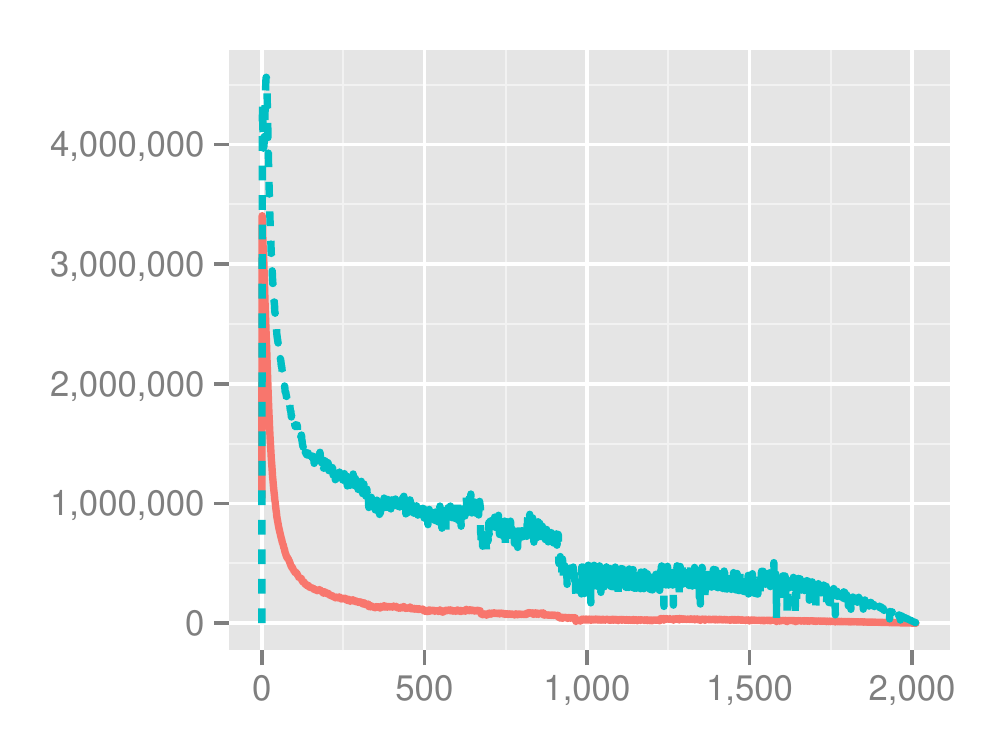}
\par\end{centering}

}
\par\end{centering}

\caption{The number of lower triangles~(y-axis) per level~(x-axis) is the blue dashed line, and the
time needed to enumerate all of them per level is the red solid line. The
time unit is 100 nanoseconds. If the time curve thus rises to 1\,000\,000
on the plot the algorithm needs 0.1 seconds. Warning: In contrast
to Figure~\ref{fig:level-sizes} these figures do not have a logarithmic
$x$-scale.}
\label{fig:level-time}
\end{figure}

\begin{table}
\tbl{Precomputed triangles. As show in Section \ref{sec:enum-triangles},
the memory needed is proportional to $2t+m+1$, where $t$ is the triangle
count and $m$ the number of arcs in the CH. We use 4 byte integers.
We report $t$ and $m$ for precomputing all levels~(``full'')
and all levels below a reasonable threshold level~(``partial'').
We further indicate how much percent of the total unaccelerated enumeration
time is spent below the given threshold level. We chose the threshold
level such that this factor is about $33\,\%$.\label{tab:triangle-size}}{%
\begin{tabular}{clrrrrrr}
\toprule
 &  & \multicolumn{2}{c}{Karlsruhe} & \multicolumn{2}{c}{TheFrozenSea} & \multicolumn{2}{c}{Europe}\\
 \cmidrule(lr){3-4}\cmidrule(lr){5-6}\cmidrule(lr){7-8}
 &  & Metis & KaHIP & Metis & KaHIP & Metis & KaHIP\\
\midrule 
\multirow{4}{*}{\begin{sideways}
full~
\end{sideways}} & \#\,Triangles {[}10\textthreesuperior{}{]} & 2\,590 & 2\,207 & 601\,846 & 864\,041 & 1\,409\,250 & 578\,247\\
 & \#\,CH arcs {[}10\textthreesuperior{}{]} & 478 & 528 & 21\,067 & 25\,100 & 70\,070 & 73\,920\\
 & Memory {[}MB{]} & 22 & 19 & 4\,672 & 6\,688 & 11\,019 & 4\,694\\
\midrule 
\multirow{5}{*}{\begin{sideways}
partial~
\end{sideways}} & Threshold level & 16 & 11 & 51 & 54 & 42 & 17\\
 & \#\,Triangles {[}10\textthreesuperior{}{]} & 507 & 512 & 126\,750 & 172\,240 & 147\,620 & 92\,144\\
 & \#\,CH arcs {[}10\textthreesuperior{}{]} & 367 & 393 & 13\,954 & 15\,996 & 58\,259 & 59\,282\\
 & Memory {[}MB{]} & 5 & 5 & 1\,020 & 1\,375 & 1\,348 & 929\\
 & Enum.~time {[}\%{]} & 33 & 32 & 33 & 33 & 32 & 33\\
\bottomrule
\end{tabular}
}%
\end{table}

We first evaluate the running time of the adjacency-array-based triangle
enumeration algorithm. Figure~\ref{fig:level-time} clearly shows
that most time is spent enumerating the triangles of the lower levels.
This justifies our suggestion to only precompute the triangles for
the lower levels as these are the levels were the optimization is
most effective. However, precomputing more levels does not hurt if
enough memory is available. We propose to determine the threshold
level up to which triangles are precomputed based on the size of the
available unoccupied memory. On modern server machines such as our
benchmarking machine there is enough memory to precompute all levels.
The memory consumption is summarized in Table~\ref{tab:triangle-size}.
However, note that precomputing all triangles is prohibitive in the
game scenario as less available memory should be expected. 

\subsection{Customization}

\begin{table}
\tbl{Customization performance. We report the time needed to compute a maximum customized metric given an initial pair of upward and downward metrics. We show the impact of enabling SSE, precomputing triangles (Pre.\ trian.), multi-threading (\#\,Thr.), and customizing several metric pairs at once.
\label{tab:customization-all}}{%
\begin{tabular}{llrrrrrrrr}
\toprule 
 & & & & \multicolumn{2}{c}{Karlsruhe} & \multicolumn{2}{c}{TheFrozenSea} & \multicolumn{2}{c}{Europe}\\
\cmidrule(lr){5-6}\cmidrule(lr){7-8}\cmidrule(lr){9-10}
   & Pre. & & \#\,Metrics & Metis & KaHIP& Metis & KaHIP & Metis & KaHIP\\
  SSE & trian. & \#\,Thr. & Pairs  & time [s] & time [s] & time [s] & time [s] & time [s] & time [s]\\
\midrule 
 no & none & 1 & 1 & 0.0567 & 0.0468& 7.88 & 10.08& 21.90 & 10.88\\
  yes & none & 1 & 1 & 0.0513 & 0.0427 & 7.33 & 9.34 & 19.91 & 9.55\\
  yes & all & 1 & 1 & 0.0094 & 0.0091 & 3.74 & 3.75 & 7.32 & 3.22\\
  yes & all & 16 & 1 & 0.0034 & 0.0035 & 0.45 & 0.61 & 1.03 & 0.74\\
  yes & all & 16 & 2 & 0.0035 & 0.0033 & 0.66 & 0.76 & 1.34 & 1.05\\
  yes & all & 16 & 4 & 0.0040 & 0.0048 & 1.19 & 1.50 & 2.80 & 1.66\\
\bottomrule
\end{tabular}
}%
\end{table}

In Table~\ref{tab:customization-all} we report the times needed
to compute a customized metric using the basic customization algorithm. A first observation
is that on the road graphs the KaHIP order leads to a faster customization
whereas on the game map Metis dominates. Using all optimizations presented
we customize Europe in below one second. When amortized%
\footnote{We refer to a server scenario of multiple active users that require
simultaneous customization, \eg, due to traffic updates.%
}, we even achieve 415\,ms which is only slightly above the non-amortized
347\,ms reported in~\cite{dw-fcrn-13} for CRP. Note that their experiments
were run on a different machine with a faster clock but $2\times6$
instead of $2\times8$ cores, while using a turn-aware data structure, making an exact comparison difficult.

Previous works have tried to accelerate the preprocessing phase of the original two-phase CH to the point that it can be used in a similar scenario as our technique. 
A fast preprocessing phase can be viewed as form of customization phase.
In~\cite{gssv-erlrn-12} a sequential preprocessing time of 451\,s was reported.
This compares best to our 9.5\,s sequential customization time. 
Note that the machine on which the 451\,s were measured is slower than our machine. 
However, the gap in performance is large enough to conclude that we achieve a significant speedup.
Furthermore, \cite{adgw-hhlsp-12} report a CH preprocessing time of 2\,min when parallelized on 12 cores. This compares best against our 415\,ms parallelized amortized customization time.
While the machine used in~\cite{adgw-hhlsp-12} is slightly older and slower than our machine and the number of cores differs (12 vs.~16), again, the performance gap is large enough to safely conclude that a significant speedup is present.
Besides the differences in running time, both CH preprocessing experiments were only performed for travel time weights.
To the best of our knowledge, nobody has been able to reproduce such running times for less well-behaved weights, such as travel distance.
For example, in~\cite{gssv-erlrn-12} the fastest CH preprocessing time reported for distance metric is 2,853\,s. 
This contrasts with CCH, for which we can prove that the customization and elimination tree query running times are completely independent of the weights used.

\begin{table}
\tbl{Detailed customization performance on TheFrozenSea. We report the time needed to compute a maximum customized metric from an initial metric.
We show the impact of exploiting undirectedness, customizing several metrics at once, reducing the bitwidth of the metric, enabling SSE, multi-threading (\#\,Thr.), and precomputing triangles (Pre.\ trian.). Note that the order in which improvements are investigated is different from Table~\ref{tab:customization-all}. Also note that results are based on the Metis order
as Table~\ref{tab:customization-all} shows that KaHIP is
outperformed.\label{tab:customization-game}}{%
\begin{tabular}{l@{\hspace{1em}}l@{\hspace{1em}}l@{\hspace{0.5em}}lrrrr}
\toprule
       &              & Metric &     &          & Precomputed   & Customization  & Amortized\\
Undirected & \#\,Metrics  & bits   & SSE & \#\,Threads & triangles & time [s] & time [s]\\
\midrule 
no  & 2 (up \& down) & 32 & no  & 1  & none & 7.88 & 7.88\\
yes & 1 & 32 & no  & 1  & none & 6.65 & 6.65\\
yes & 4 & 32 & no  & 1  & none & 9.36 & 2.34\\
yes & 4 & 32 & yes & 1  & none & 8.51 & 2.13\\
yes & 8 & 16 & yes & 1  & none & 8.52 & 1.06\\
yes & 8 & 16 & yes & 2  & none & 5.00 & 0.63\\
yes & 8 & 16 & yes & 2  & all  & 2.16 & 0.27\\
yes & 8 & 16 & yes & 16 & all  & 0.63 & 0.08\\
\bottomrule
\end{tabular}
}%
\end{table}

Unfortunately, the optimizations illustrated in Table~\ref{tab:customization-all}
are pretty far from what is possible with the hardware normally available in a game scenario. Regular
PCs do not have 16 cores and one cannot clutter up the whole RAM with
several GB of precomputed triangles. We therefore ran additionally experiments with different parameters and report the results 
in Table~\ref{tab:customization-game}. The experiments 
show that it is possible to fully customize TheFrozenSea in an amortized%
\footnote{We refer to a multiplayer scenario, where, \eg, fog of war requires
player-specific simultaneous customization.%
} time of 1.06s without precomputing triangles or using multiple cores.
However a whole second is still too slow to be usable, as graphics,
network and game logic also require resources. We therefore evaluated
the time needed by partial updates as described in Section~\ref{sec:partial-updates}.
We report our results in Table~\ref{tab:partial-updates}. The median,
average and maximum running times significantly differ. There are
a few arcs that trigger a lot of subsequent changes whereas for most
arcs a weight change has nearly no effect. The explanation is that
highway arcs and choke point arcs are part of many shortest paths
and thus updating such an arc triggers significantly more changes.

Finally, we report the running times of the perfect customization algorithm
in Table~\ref{tab:perfect-customization}.
The required running time is about 3 times the running time needed by the basic customization.
Recall that the basic customization in essence enumerates all lower triangles, \ie, it visits every triangle once, while the perfect customization also enumerates all intermediate and upper triangles, \ie, it visits every triangle three times.

\begin{table}
\tbl{Partial update performance. We report time required in milliseconds and number of arcs changed for partial metric updates.
We report median, average and maximum over $10000$ runs. In each
run we change the upward and the downward weight of a single random
arc in $G$ to
random values in $[0,10^{5}]$. The metric is reset to initial
state between runs. Timings are sequential without SSE. No triangles
were precomputed.\label{tab:partial-updates}}{%

\begin{tabular}{llrrrrrr}
\toprule 
 &  & \multicolumn{3}{c}{Arcs removed from queue} & \multicolumn{3}{c}{Partial update time [ms]}\\
 \cmidrule(lr){3-5}\cmidrule(lr){6-8}
 &  & med. & avg. & max. & med. & avg. & max.\\
\midrule 
\multirow{2}{*}{Karlsruhe} & Metis &  2 & 3.5 & 857 & 0.001 & 0.003 & 0.9 \\
 & KaHIP & 3 & 3.7 & 466 & 0.001 & 0.002 & 1.0 \\
\addlinespace
\multirow{2}{*}{TheFrozenSea} & Metis & 6 & 311.7 & 14494 & 0.008 & 1.412 & 100.2 \\
 & KaHIP & 6 & 343.1 & 19417 & 0.008 & 1.490 & 164.6\\
\addlinespace
\multirow{2}{*}{Europe} & Metis & 2 & 10.2 & 14188 & 0.005 & 0.052 & 134.6 \\
 & KaHIP & 3 & 9.8 & 8202 & 0.008 & 0.045 & 81.0\\
\bottomrule
\end{tabular}
}%
\end{table}

\begin{table}
\tbl{Perfect Customization. We report the time required to turn an initial metric into a perfect
metric. Runtime is given in seconds, without use of SSE.\label{tab:perfect-customization}}{%
\begin{tabular}{rrrrrrrr}
\toprule
  & & \multicolumn{2}{c}{Karlsruhe} & \multicolumn{2}{c}{TheFrozenSea} & \multicolumn{2}{c}{Europe}\\
\cmidrule(lr){3-4}\cmidrule(lr){5-6}\cmidrule(lr){7-8}
 \#\,Thr. & Pre. trian. & Metis & KaHIP & Metis & KaHIP & Metis & KaHIP\\
\midrule 
1 & none & 0.15 & 0.13 & 30.54  & 33.76 & 67.01 & 32.96 \\
16 & none & 0.03 & 0.02 & 3.26  & 4.37& 14.41 & 5.47 \\
1 & all & 0.05 & 0.05 & 8.95 & 12.51 & 23.93 & 10.75 \\
16 & all & 0.01 & 0.01 & 1.93 & 2.29 & 3.50 & 2.35 \\
\bottomrule 
\end{tabular}
}%
\end{table}

\subsection{Query Performance}
\label{sec:experiments:query}

\begin{table}
\tbl{Contraction Hierarchies query performance. We report the query time \emph{in microseconds} as well as the search space visited, averaged over $10^6$ queries with source and target vertices picked uniformly at random. We use ``visited'' to differentiate from the maximum reachable search space given in Table~\ref{tab:ch-sizes}.
For query algorithms that use stalling, we additionally report the number of vertices stalled, as well as the number of arcs touched during the stalling test. 
Note that stalled vertices are not counted as visited.
All reported vertex and arc counts only refer to the forward search.
We evaluate several algorithmic variants. Each variant is composed of an input graph, a contraction order, and whether a witness search is used.
``+w'' means that a witness search is used, whereas ``-w'' means that no witness search is used.
``MetDep+w'' corresponds to the original CHs. 
The metrics used for ``-w'' are directed and maximum.
Results for Karlsruhe and the TheFrozenSea instances are in this table.
Table~\ref{tab:query-experiments-2} contains the results for the Europe instance.%
\label{tab:query-experiments}}{%
\begin{tabular}{ccllrrrrr}
\toprule 
 &  & &       & \multicolumn{2}{c}{Visited up. search space}          & \multicolumn{2}{c}{Stalling} & Time \\
 \cmidrule(lr){5-6} \cmidrule(lr){7-8}
 Instance & Metric & Variant & Algorithm & \#\,Vertices &  \#\,Arcs  & \#\,Vertices & \#\,Arcs & [\musec]\\

\midrule 
\multirow{28}{*}{\begin{sideways}Karlsruhe\end{sideways}} & \multirow{14}{*}{\begin{sideways}Travel-Time\end{sideways}}
 & \multirow{2}{*}{MetDep+w} & Basic & 81 & 370 & --- & --- & 17\\
 & &  & Stalling & 43 & 182 & 167 & 227 & 16\\
\addlinespace
 &  & \multirow{3}{*}{Metis-w} & Basic & 138 & 5\,594 & --- & --- & 62\\
 &  &  & Stalling & 104 & 4\,027 & 32 & 4\,278 & 67\\
 &  &  & Tree & 164 & 6\,579 & --- & --- & 33\\
\addlinespace
 &  & \multirow{3}{*}{KaHIP-w} & Basic & 120 & 4\,024 & --- & --- & 48\\
 &  &  & Stalling & 93 & 3\,051 & 26 & 3\,244 & 55\\
 &  &  & Tree & 143 & 4\,723 & --- & --- & 25\\

\addlinespace
 &  & \multirow{3}{*}{Metis+w} & Basic & 127 & 2\,432 & --- & --- & 32 \\
 &  &  & Stalling & 104 & 2\,043 & 19 & 2\,146 & 41\\
 &  &  & Tree & 164 & 2\,882 & --- & --- & 17\\
\addlinespace
 &  & \multirow{3}{*}{KaHIP+w} & Basic & 114 & 1\,919 & --- & --- & 27\\
 &  &  & Stalling & 93 & 1\,611 & 18 & 1\,691 & 35\\
 &  &  & Tree & 143 & 2\,198 & --- & --- & 14\\

\cmidrule(l){2-9} 
 & \multirow{14}{*}{\begin{sideways}Distance\end{sideways}} 
 & \multirow{2}{*}{MetDep+w} & Basic & 208 & 1978 & --- & --- & 57\\
 &  &  & Stalling & 70 & 559 & 46 & 759 & 35\\
\addlinespace
 & & \multirow{3}{*}{Metis-w} & \multirow{1}{*}{Basic} & 142 & 5\,725 & --- & --- & 65\\
 &  &  & Stalling & 115 & 4\,594 & 26 & 4\,804 & 75\\
 &  &  & Tree & 164 & 6\,579 & --- & --- & 33\\
\addlinespace
 &  & \multirow{3}{*}{KaHIP-w} & Basic & 123 & 4\,117 & --- & --- & 50\\
 &  &  & Stalling & 106 & 3\,480 & 17 & 3\,564 & 59\\
 &  &  & Tree & 143 & 4\,723 & --- & --- & 26\\

\addlinespace
 &  & \multirow{3}{*}{Metis+w} & Basic & 138 & 3\,221 & --- & --- & 39\\
 &  &  & Stalling & 115 & 2\,757 & 21 & 2\,867 & 50\\
 &  &  & Tree & 164 & 3\,604 & --- & --- & 21\\
\addlinespace
 &  & \multirow{3}{*}{KaHIP+w} & Basic & 122 & 2\,626 & --- & --- & 32\\
 &  &  & Stalling & 106 & 2\,302 & 14 & 2\,350 & 43\\
 &  &  & Tree & 143 & 2\,956 & --- & --- & 17\\

\bottomrule
\end{tabular}
}%
\end{table}

\begin{table}
\tbl{Contraction Hierarchies query performance. Continuation of Table~\ref{tab:query-experiments}.
\label{tab:query-experiments-2}}{%
\begin{tabular}{ccllrrrrr}
\toprule 
 &  & &       & \multicolumn{2}{c}{Visited up. search space}          & \multicolumn{2}{c}{Stalling} & Time \\
 \cmidrule(lr){5-6} \cmidrule(lr){7-8}
 Instance & Metric & Variant & Algorithm & \#\,Vertices &  \#\,Arcs  & \#\,Vertices & \#\,Arcs & [\musec]\\

\midrule 
\multirow{14}{*}{\begin{sideways}TheFrozenSea\end{sideways}} & \multirow{14}{*}{\begin{sideways}Map-Distance\end{sideways}}
 & \multirow{2}{*}{MetDep+w} & Basic & 1\,199 & 12\,692 & --- & --- & 539\\
 &  &  & Stalling & 319 & 3\,460 & 197 & 4\,345 & 286\\
\addlinespace
 &  & \multirow{3}{*}{Metis-w} & Basic & 610 & 81\,909 & --- & --- & 608\\
 &  &  & Stalling & 578 & 78\,655 & 24 & 79\,166 & 837\\
 &  &  & Tree & 676 & 92\,144 & --- & --- & 317\\
\addlinespace
 &  & \multirow{3}{*}{KaHIP-w} & Basic & 603 & 82\,824 & --- & --- & 644\\
 &  &  & Stalling & 560 & 74\,244 & 50 & 74\,895 & 774\\
 &  &  & Tree & 674 & 89\,567 & --- & --- & 316\\

\addlinespace
 &  & \multirow{3}{*}{Metis+w} & Basic & 567 & 28\,746 & --- & --- & 243\\
 &  &  & Stalling & 474 & 25\,041 & 86 & 25\,445 & 333\\
 &  &  & Tree & 676 & 31\,883 & --- & --- & 120\\
\addlinespace
 &  & \multirow{3}{*}{KaHIP+w} & Basic & 578 & 22\,803 & --- & --- & 203\\
 &  &  & Stalling & 475 & 19\,978 & 81 & 20\,138 & 276\\
 &  &  & Tree & 674 & 24\,670 & --- & --- & 106\\
\midrule 
\multirow{28}{*}{\begin{sideways}Europe\end{sideways}} & \multirow{14}{*}{\begin{sideways}Travel-Time\end{sideways}} 
 & \multirow{2}{*}{MetDep+w} & Basic & 546 & 3\,623 & --- & --- & \,283\\
 &  &  & Stalling & 113 & 668 & 75 & 911 & 107\\
\addlinespace
 &  & \multirow{3}{*}{Metis-w} & Basic & 1\,126 & 405\,367 & --- & --- & 2\,838\\
 &  &  & Stalling & 719 & 241\,820 & 398 & 268\,499 & 2\,602\\
 &  &  & Tree & 1\,291 & 464\,956 & --- & --- & 1\,496\\
\addlinespace
 &  & \multirow{3}{*}{KaHIP-w} & Basic & 581 & 107\,297 & --- & --- & 810\\
 &  &  & Stalling & 418 & 75\,694 & 152 & 77\,871 & 857\\
 &  &  & Tree & 652 & 117\,406 & --- & --- & 413\\

\addlinespace
 &  & \multirow{3}{*}{Metis+w} & Basic & 1\,026 & 110\,590 & --- & --- & 731\\
 &  &  & Stalling & 716 & 83\,047 & 271 & 89\,444 & 951\\
 &  &  & Tree & 1\,291 & 126\,403 & --- & --- & 398\\
\addlinespace
 &  & \multirow{3}{*}{KaHIP+w} & Basic & 549 & 41\,410 & --- & --- & 305\\
 &  &  & Stalling & 418 & 33\,078 & 117 & 34\,614 & 425\\
 &  &  & Tree & 652 & 45\,587 & --- & --- & 161\\

\cmidrule(l){2-9} 
 & \multirow{14}{*}{\begin{sideways}Distance\end{sideways}} 
 & \multirow{2}{*}{MetDep+w} & Basic & 3\,653 & 104\,548 & --- & --- & 2\,662\\
 &  &  & Stalling & 286 & 7\,124 & 426 & 11\,500 & 540\\
\addlinespace
 &  & \multirow{3}{*}{Metis-w} & Basic & 1\,128 & 410\,985 & --- & --- & 3\,087\\
 &  &  & Stalling & 831 & 291\,545 & 293 & 308\,632 & 3\,128\\
 &  &  & Tree & 1\,291 & 464\,956 & --- & --- & 1\,520\\
\addlinespace
 &  & \multirow{3}{*}{KaHIP-w} & Basic & 584 & 108\,039 & --- & --- & 867\\
 &  &  & Stalling & 468 & 85\,422 & 113 & 87\,315 & 1\,000\\
 &  &  & Tree & 652 & 117\,406 & --- & --- & 426\\

\addlinespace
 &  & \multirow{3}{*}{Metis+w} & Basic & 1\,085 & 157\,400 & --- & --- & 1\,075\\
 &  &  & Stalling & 823 & 124\,472 & 247 & 127\,523 & 1\,400\\
 &  &  & Tree & 1\,291 & 177\,513 & --- & --- & 557\\
\addlinespace
 &  & \multirow{3}{*}{KaHIP+w} & Basic & 575 & 56\,386 & --- & --- & 425\\
 &  &  & Stalling & 467 & 46\,657 & 101 & 47\,920 & 578\\
 &  &  & Tree & 652 & 61\,714 & --- & --- & 214\\

\bottomrule
\end{tabular}
}%
\end{table}

We experimentally evaluated the running times of the query algorithms.
For this we ran $10^{6}$ shortest path distance queries with the
source and target vertices picked uniformly at random. The presented times are
averaged running times on a single core without SSE.

In Table~\ref{tab:query-experiments} and~\ref{tab:query-experiments-2} we compare the query running times.
The ``MetDep+w'' variant use a metric-dependent order and a non-perfect witness search in the spirit of~\cite{gssv-erlrn-12}.
The details are described in Section~\ref{sec:trad-ch-order}.
The ``Metis-w'' and ``KaHIP-w'' variants use a metric-independent order computed by Metis or KaHIP.
Only a basic customization was performed, \ie, no witness search was performed.
The ``Metis+w'' and ``KaHIP+w'' variants use the same metric-independent order but a perfect customization followed by a perfect witness search was performed.
We evaluate three query variants.
The ``Basic'' variant uses a bidirectional variant of Dijkstra's algorithm with stopping criterion.
The ``Stalling'' variant additionally uses the stall-on-demand optimization as described in Section~\ref{sec:stalling}.
Finally, we also evaluate the elimination tree query and refer to it as ``Tree''.
This query requires the existence of an elimination tree of low depth and is therefore not available for metric-dependent orders.
We ran our experiments on all three of our main benchmark instances.
Experiments on additional instances are available in Section~\ref{sec:further-instances}.
For both road graphs, we evaluate the travel-time and distance variants.
We report the average running time needed to perform a distance query, \ie, we do not unpack the paths.
We further report the average number of ``visited'' vertices in the forward search.
For the ``basic'' and ``stalling'' queries, these are the vertices removed from the queue.
For the ``tree'' query, we regard every ancestor as ``visited''.
The numbers for the backward search are analogous and therefore not reported.
We also report the average number of arcs relaxed in forward search of each query variant.
Finally, we also report the average number of vertices stalled and the average number of arcs that need to be tested in the stalling test.
Note that, a stalled vertex is not counted as ``visited''.

An important detail necessary to reproduce these results consists of reordering the vertex IDs according to the contraction order.
Preliminary experiments showed that this reordering results in better cache behavior and a speed-up of about 2 to 3 because much query time is spent on the topmost clique and this order assures that these vertices appear adjacent in memory.

As already observed by the original CH authors, we confirm that the stall-on-demand heuristic improves running times by a factor of~2--5 compared to the basic algorithm for ``greedy+w''.
Interestingly, this is not the case with any variant using a metric-independent order.
This can be explained by the density of the search spaces.
While, the number of vertices in the search spaces are comparable between metric-independent orders and metric-dependent order, the number of arcs are not comparable and thus metric-independent search spaces are denser.
As consequence, we need to test significantly more arcs in the stalling-test, which makes the test more expensive and therefore the additional time spent in the test does not make up for the time economized in the actual search.
We thus conclude that stall-on-demand is not useful, when using metric-independent orders.

Very interesting is the comparison between the elimination tree query and the basic query.
The elimination tree query always explores the whole search space.
In contrast to the basic query, it does not have a stopping criterion.
However, the elimination tree query does not require a priority queue.
It performs thus less work per vertex and arc than the basic query.
Our experiments show, that the basic query always explores large parts of the search space regardless of the stopping criterion.
The elimination tree query therefore does not visit significantly more.
A consequence of this effect is that the time spent in the priority queue outweighs the additional time necessary to explore the remainder of the search space.
The elimination tree query is therefore always the fastest among the three query types when using metric-independent orders.
Combining a perfect witness search with the elimination tree query results in the fastest queries for metric-independent orders.
However, the perfect witness search results in three times higher customization times. 
Whether, it is superior therefore depends on the specific application and the specific trade-off between customization and query running time needed.

The orders computed by KaHIP are nearly always significantly better than those produced by Metis.
However, significantly more running time must be invested in the preprocessing phase to obtain these better order.
It therefore depends on the situation which order is better.
If the running time of the preprocessing phase is relevant, then Metis seems to strike a very good balance between all criteria. 
However, if the graph topology is fixed, as we expect it to be, then the flexibility gained by using Metis is not worth the price.
Interestingly, on the game map KaHIP and Metis seem to be very close in terms of search space size.
The difference is only apparent when using the perfect customization.
For a setup with basic customization, the two orders should are nearly indistinguishable.

The metric-dependent orders with travel-time outperform the metric-independent orders.
However, it is very interesting how close the query times actually are. 
On the Europe graph, the basic query visits about the same number of vertices, regardless of whether a metric-dependent or the KaHIP order is used.
The real difference lies in the number of arcs that need to be relaxed.
This number is significantly higher with metric-independent orders.
However, the effect this has on the actual running times is comparatively slim. 
Using KaHIP without perfect witness search results in an elimination tree query that is only about 4 times slower than using the stalling query combined with metric-dependent orders.
If a perfect witness search is used then, the gap is below a factor of 2.
Further, the metric-dependent orders only win because of the stall-on-demand optimization.
The KaHIP order combined with perfect customization \emph{outperforms} the basic query combined with metric-dependent orders.

It is well-known that metric-dependent CHs work significantly better with the travel-time metric than with other less well behaved metrics such as the geographic distance.
For such metrics, the KaHIP order outperforms the metric-dependent orders.
For example the basic query with perfect customization visits less vertices and less arcs. 
This is very surprising, especially considering, that the metric-dependent orders that we computed are better than those reported in~\cite{gssv-erlrn-12}, \ie, the gap with respect to the original implementation is even larger.
However, combining the stalling query with metric-dependent orders yields the smallest number of visited vertices and relaxed arcs.
Unfortunately, combining the stalling query with metric-independent orders does not yield the same benefit and even makes the query running times worse.
Fortunately, the metric-independent orders can be combined with the elimination tree query.
As result, the fastest variant is the combination of KaHIP order, perfect witness search, and elimination tree query, which is over a factor of two faster than stalling with the metric-dependent order.
Interestingly, the later is even beaten when no perfect witness search is performed, but with a significantly lower margin.

A huge advantage of metric-independent orders compared to metric-dependent orders is that the resulting CH performs equally well regardless of the weights of the input graph.
The combination of metric-independent order, elimination tree query and basic customization results in setup, where the order in which the vertices are visited and the order in which the arcs are relaxed during the query execution does not even depend on the weights of the input graph.
It is thus impossible to construct a metric, where this setup performs badly.
This contrasts with the CH of~\cite{gssv-erlrn-12}, whose performances varies significantly depending on the input metric.

\begin{table}
\tbl{Detailed elimination tree performance without perfect witness search. We report running time \emph{in microseconds} for the elimination-tree-based query algorithms. We report the time needed to compute the LCA, the time
needed to reset the tentative distances, the time needed to relax
the arcs, the total time of a distance query, and the time needed
for full path unpacking as well as the average number of vertices
on such a path which is metric-dependent.\label{tab:elim-tree-times}}{%
\begin{tabular}{lllrrrrrr}
\toprule 
 &  &  & \multicolumn{4}{c}{Distance query}& \multicolumn{2}{c}{Path}\\
 \cmidrule(lr){4-7}\cmidrule(lr){8-9}
 &  &  & LCA & Reset &  Arc relax & Total & Unpack & Length\\
 &  &  & [\musec] & [\musec] & [\musec] & [\musec] & [\musec] & [vert.]\\
\midrule 
\multirow{4}{*}{Karlsruhe} & \multirow{2}{*}{Travel-Time} & Metis & 0.6 & 0.8 & 31.3 & 33.0 & 20.5 & \multirow{2}{*}{189.6}\\
 &  & KaHIP & 0.6 & 1.4 & 23.1 & 25.2 & 18.6 & \\
\addlinespace
 & \multirow{2}{*}{Distance} & Metis & 0.6 & 0.8 & 31.5 & 33.2 & 27.4 & \multirow{2}{*}{249.4}\\
 &  & KaHIP & 0.6 & 1.4 & 23.5 & 25.7 & 24.7 & \\
\midrule 
\multirow{2}{*}{TheFrozenSea} & \multirow{2}{*}{Map-Distance} & Metis & 2.7 & 3.1 & 310.1 & 316.5 & 220.0 & \multirow{2}{*}{596.3}\\
 &  & KaHIP & 3.0 & 3.2 & 308.7 & 315.5 & 270.8 & \\
\midrule 
\multirow{4}{*}{Europe} & \multirow{2}{*}{Travel-Time} & Metis & 4.6 & 19.0 & 1471.2 & 1496.3 & 323.9 & \multirow{2}{*}{1390.6}\\
 &  & KaHIP & 3.4 & 9.9 & 399.4 & 413.3 & 252.7 & \\
\addlinespace
 & \multirow{2}{*}{Distance} & Metis & 4.7 & 19.0 & 1494.5  & 1519.9 & 608.8 & \multirow{2}{*}{3111.0}\\
 &  & KaHIP & 3.6 & 10.0 & 411.6 & 425.8 & 524.1 & \\
\bottomrule 
\end{tabular}
}%
\end{table}

In Table~\ref{tab:elim-tree-times}, we give a more in-depth experimental analysis of the elimination tree query algorithm without perfect witness search. 
We break the running times down into the time needed to compute the least common ancestor~(LCA), the time needed to reset the tentative distances and the time needed to relax all arcs. 
We further report the total distance query time, which is in essence the sum of the former three.
We additionally report the time needed to unpack the full path. 
Our experiments show that the arc-relaxation phase clearly dominates the running times. 
It is therefore not useful to further optimize the LCA computation or to accelerate tentative distance resetting using, \eg, timestamps. 
We only report path unpacking performance without precomputed lower triangles. 
Using them would result in a further speedup with a similar speed-memory trade-off as already discussed for customization.

\subsection{Comparison with Related Work}

\begin{table}
\tbl{Comparison with related work on the DIMACS Europe instance with travel time and distance metric. We compare our approaches, CCH, CCH with amortized customization~(CCH+a), and CCH with perfect witness search~(CCH+w), with different CRP and CH implementations from the literature. We report performance of the metric-\emph{dependent} fraction of overall preprocessing, \ie, vertex ordering and contraction time for CH, customization time for CRP and CCH. We further report average query search space, including stalled vertices for CH~(which might not be included in the CH figures taken from~\cite{dgpw-crprn-13}).
We finally report running time in microseconds. If parallelized, the number of threads used is noted in parenthesis.
Since the CH performance in~\cite{gssv-erlrn-12} was evaluated on a ten year old machine (AMD Opteron 270), we obtained the source code and re-ran experiments on our hardware (Intel Xeon E5-2670) for better comparability. Also note that the latest CRP implementation by~\cite{dgpw-crprn-13}, evaluated on an Intel Xeon X5680, is turn-aware (\enabled), \ie, it uses turn tables~(set to zero in the reported experiments);  We therefore additionally take results from~\cite{dgpw-crp-11} obtained on an Intel Core-i7 920, which uses a turn-unaware implementation but parallelizes queries.
\label{tab:comparison-related-work}}{%
\setlength{\tabcolsep}{1.0ex}
\begin{tabular}{llllcrrr}
\toprule
&&&&& Metric-Dep.\\
&&&&& Prepro.& \multicolumn{2}{c}{Queries}\\
\cmidrule(lr){6-6}
\cmidrule(lr){7-8}
&&&&Turn-&Time [s] & Search Space & Time [\musec] \\
Algorithm & Implementation & Machine & Metric & aware & (\#\,Threads) & [\#\,Vertices] & (\#\,Threads) \\
\midrule
CH & \cite{gssv-erlrn-12} & Opt 270 &  Time & \disabled & 1\,809\phantom{.00}~~\phantom{0}(1)& 356 & 152~~(1)\\ %
CH & \cite{gssv-erlrn-12} &Opt 270 & Dist & \disabled & 5\,723\phantom{.00}~~\phantom{0}(1) & 1\,582 & 1\,940~~(1)\\ %
CH & \cite{gssv-erlrn-12} & E5-2670 &  Time & \disabled & 1\,075.88~~\phantom{0}(1)& 353 & 91~~(1)\\ 
CH & \cite{gssv-erlrn-12} & E5-2670 &  Dist & \disabled & 3\,547.44~~\phantom{0}(1)&  1\,714 & 1\,135~~(1)\\ %
CH & our & E5-2670 & Time & \disabled & 813.53~~\phantom{0}(1)& 375 & 107~~(1) \\ 
CH & our & E5-2670 & Dist & \disabled & 9\,390.32~~\phantom{0}(1)& 1422 & 540~~(1) \\ 
CH & \cite{dgpw-crprn-13} & X5680 & Time & \disabled & 109\phantom{.00}~~(12)& 280 & 110~~(1)\\ %
CH & \cite{dgpw-crprn-13} & X5680 & Dist & \disabled & 726\phantom{.00}~~(12)& 858 & 870~~(1)\\ %
\addlinespace
CRP & \cite{dgpw-crprn-13}& X5680 & Time & \enabled & 0.37~~(12) & 2\,766 & 1\,650~~(1)\\ %
CRP & \cite{dgpw-crprn-13} & X5680 & Dist & \enabled & 0.37~~(12) & 2\,942 & 1\,910~~(1)\\ 
CRP & \cite{dgpw-crp-11} & i7 920 & Time & \disabled & 4.7\phantom{0}~~\phantom{0}(4) & 3\,828& 720~~(2) \\ %
CRP & \cite{dgpw-crp-11} & i7 920 & Dist & \disabled & 4.7\phantom{0}~~\phantom{0}(4)& 4\,033 & 790~~(2)\\
\addlinespace
CCH & our & E5-2670 & Time & \disabled & 0.74~~(16)& 1\,303 & 413~~(1)\\
CCH & our & E5-2670 & Dist & \disabled & 0.74~~(16)& 1\,303 & 426~~(1)\\
CCH+a & our & E5-2670 & Time & \disabled & 0.42~~(16)& 1\,303 & 416~~(1)\\
CCH+a & our & E5-2670 & Dist & \disabled & 0.42~~(16)& 1\,303 & 421~~(1)\\
CCH+w & our & E5-2670 & Time & \disabled & 2.35~~(16)& 1\,303 & 161~~(1) \\
CCH+w & our & E5-2670 & Dist & \disabled & 2.35~~(16)& 1\,303 & 214~~(1)\\
\bottomrule
\end{tabular}
 }
\end{table}

We conclude our experimental analysis on the DIMACS Europe road network with a final comparison of related techniques, as shown in Table~\ref{tab:comparison-related-work}. For Contraction Hierarchies~(CH), we report results based on implementations by~\cite{gssv-erlrn-12,dgpw-crprn-13} and ourselves, covering different trade-offs in terms of preprocessing versus query speed. More precisely, we observe that our own CH implementation~(used for detailed analysis and comparison in Section~\ref{sec:experiments:orders}--\ref{sec:experiments:query}) 
has slightly slower queries on travel time metric but factor of~2.1 faster queries on distance metric, at the cost of higher preprocessing time.
Recall from Section~\ref{sec:trad-ch-order} that we employ a different vertex priority function and no lazy updates.
For Customizable Route Planning~(CRP), we report results from~\cite{dgpw-crprn-13,dgpw-crp-11}.

Traditional, metric-dependent CH offers the fastest query time~(91\,\musec, on our machine), but it incurs substantial metric-dependent preprocessing costs, even when parallelized~(109\,s, 12 cores). Furthermore, CH performance is very sensible regarding metrics used: For distance metric, preprocessing time increases by factor of~3.2--11.5 and query time by factor of~4.9--12.8.

In contrast to traditional CH, Customizable Contraction Hierarchies~(CCH) by design achieve a performance trade-off with much lower metric-dependent preprocessing costs, similar to CRP. Accounting for differences in hardware, CCH basic customization time is about a factor of~2--3 slower than CRP customization, but still well below a second. On the other hand, CCH query performance is factor of~2--4 faster than CRP, both in terms of search space as well as query time~(even when accounting for differences due to turn-aware implementation and hardware used). Most interestingly, on travel distance, CCH outperforms even the best CH result in terms of query performance.
Overall, CCH is more robust \wrt~to the metric than CRP: By design, CCH customization processes the same sequence of lower triangles for any metric, while the CCH elimination-tree query~(given a fixed source and target) processes the same sequence of vertices and arcs for any metric.

The CRP implementation of~\cite{dgpw-crprn-13} uses SSE to achieve its customization time of 0.37\,s. In a server scenario where customization is run for many users concurrently~(\eg, to customize traffic updates for all active users), we propose to amortize triangle enumeration time by using SSE to customize metrics for four users at once~(\cf~``Metric Pairs'' in Table~\ref{tab:customization-all}). With this amortized customization~(CCH+a, 0.42\,s), we can almost close the gap to CRP customization performance.

For even better CCH query performance, we may employ perfect customization and witness search (CCH+w). It increases customization time by factor of~3.2~(enumerating all lower, intermediate and upper triangles), but enables a CCH query variant that, while still visiting all vertices in the elimination tree, needs to consider far fewer arcs~(\cf~Table~\ref{tab:query-experiments-2}). Thereby, CCH+w further improves CCH query performance by factor of~1.9 for distance metric and factor of~2.6 for time metric. With 161\,\musec{} for travel time, CCH+w query times are almost as fast as the best CH result of~91\,\musec.

\section{Further Instances}
\label{sec:further-instances}

\subsection{OpenStreetMap-based Road Graphs}
\label{sec:osm-europe}

\begin{wrapfigure}{o}{4.5cm}%
\begin{center}
\includegraphics[scale=0.1]{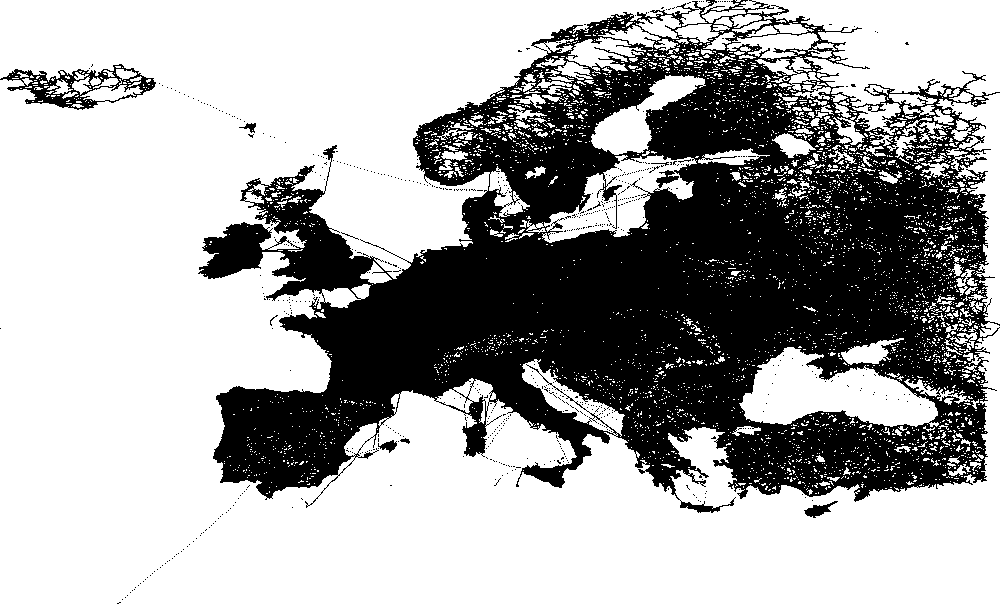}
\end{center}

\caption{All vertices in the OSM-Europe graph.}
\label{fig:osm-europe}
\end{wrapfigure}

OpenStreetMap (OSM) is a very popular collaborative effort to create a map of the world. 
From this huge data source very large road graphs can be extracted, that are very detailed depending on the exact region considered.
Using the data provided by GeoFabrik\footnote{\url{http://download.geofabrik.de/}} and the tools provided by OSRM\footnote{\url{http://project-osrm.org/}}, we extracted a road graph of Europe and report its size in Table~\ref{tab:osm-size}.
The exact graph is available in DIMACS format on our website\footnote{\url{http://i11www.iti.kit.edu/resources/roadgraphs.php}}.
The geographic region corresponding to the graph is depicted in Figure~\ref{fig:osm-europe}.
Note that compared to the DIMACS Europe, our OSM Europe graph also contains Eastern Europe and Turkey.
The graph's east border ends at the east border of Turkey and then goes upward cutting through Russia.
On the other hand, the DIMACS Europe graph stops at the German-Polish border.

\begin{wraptable}{o}{7.5cm}
\vspace{-2em}

\begin{center}
\begin{tabular}{rrr}
\toprule
 & OSM-Eur & DIMACS-Eur\tabularnewline
\midrule 
Vertices & 174M & 18M\tabularnewline
Arcs & 348M & 42M\tabularnewline
\addlinespace
Deg. 1 Vertices & 8M & 4M\tabularnewline
Deg. 2 Vertices & 143M & 2M\tabularnewline
Deg. $>$2 Vertices & 23M & 11M\tabularnewline
\bottomrule 
\end{tabular}
\end{center}
\vspace{-1em}
\caption{Size of the DIMACS Europe instance compared to the OSM Europe instance.\label{tab:osm-size}}
\vspace{1em}
\end{wraptable}

At first glance the DIMACS Europe graph looks drastically smaller, at least in terms of vertex count.
However, this is very misleading.
A peculiarity of OSM is that the road graphs have a huge number of degree-2 vertices.
These vertices are used to encode the curvature of a road.
This information is needed to correctly represent a road graph on a map but not necessarily for routing.
However, most other data sources, including the one on which the DIMACS graph is based, encode this information as arc attributes and thus have fewer degree-2 vertices.
Accelerating shortest path computations on graphs with a huge number of vertices of degree 1 or 2 is significantly easier relative to the graph size.
One reason is that Dijkstra's algorithm cannot exploit the abundance of low-degree vertices.
Dijkstra's algorithm with stopping criterion needs on average 27\,s for a $st$-query with $s$ and $t$ picked uniformly at random on the OSM-Europe graph.
This contrasts with the DIMACS Europe graph, where only 1.6\,s are needed.
A slower baseline obviously leads to larger speedups.
Table~\ref{tab:osm-size} shows that the difference between the two Europe graphs in terms of vertex count is significantly smaller, when discarding degree 1 and degree-2 vertices.
In fact, relative to their geographical region's area, the two graphs seem to be approximately comparable in size.

\begin{wraptable}{o}{7.5cm}
\begin{center}
\begin{tabular}{llrr}
\toprule
 & Partitioner & Vertices & Arcs\tabularnewline
\midrule 
Input &  & 174M & 348M\tabularnewline
\addlinespace
Search Graph Size & Metis & 174M & 400M\tabularnewline
Avg. Search Space & Metis & 1~312 & 495~930\tabularnewline
\addlinespace
Search Graph Size & KaHIP & 174M & 434M\tabularnewline
Avg. Search Space & KaHIP & 678 & 119~295\tabularnewline
\bottomrule 
\end{tabular}
\end{center}
\vspace{-1em}
\caption{CH sizes for OSM Europe. The search space sizes were obtained by randomly sampling 10\,000 vertices uniformly at random.\label{tab:osm-ch-sizes}}
\vspace{1em}
\begin{center}
\begin{tabular}{ccrr}
\toprule
 &  & Triangle  & Customization \tabularnewline
 & \#Thr. & space {[}GB{]} & time {[}s{]}\tabularnewline
\midrule
\multirow{4}{*}{Metis} & 1 &  --- & 43.1\tabularnewline
 & 16 & --- & 5.3\tabularnewline
 & 1 & 16.0 & 17.3\tabularnewline
 & 16 & 16.0 & 2.1\tabularnewline
\midrule
\multirow{4}{*}{KaHIP} & 1 &  --- & 30.6\tabularnewline
 & 16 & --- & 3.4\tabularnewline
 & 1 & 7.2 & 11.4\tabularnewline
 & 16 & 7.2 & 1.7\tabularnewline
\bottomrule
\end{tabular}
\end{center}
\vspace{-1em}
\caption{Customization performance on OSM~Europe. We vary the number of threads and whether precomputed triangles are used. SSE is enabled, running times are non-amortized and no perfect customization was performed.
\label{tab:osm-custom}}
\vspace{1em}
\begin{center}
\begin{tabular}{ccr}
\toprule 
 &  & Query time {[}ms{]}\tabularnewline
\midrule 
Dijkstra- & Metis & 3.7\tabularnewline
based & Metis & 1.0\tabularnewline
\addlinespace
Elimination- & KaHIP & 1.7\tabularnewline
Tree & KaHIP & 0.5\tabularnewline
\bottomrule
\end{tabular}
\end{center}
\vspace{-1em}
\caption{Query performance on OSM~Europe, averaged over 10~000 random $st$-pairs chosen uniformly
at random.\label{tab:osm-query}}
\vspace{1em}

\end{wraptable}%

We computed contraction orders for OSM-Europe.
The sizes of the resulting CHs are reported in Table~\ref{tab:osm-ch-sizes}.
These sizes can be compared with the ``undirected'' numbers of Table~\ref{tab:ch-sizes}.
We did not perform experiments with a perfect witness search.
Metis ordered the vertices within 29 minutes, whereas the KaHIP-based ordering algorithm needed slightly less than 3 weeks.
However, as already discussed in detail, we did not optimize the later for speed and therefore one must \emph{not} conclude from this experiment that KaHIP is slow.
The CHs for OSM-based graphs are significantly larger.
The DIMACS-Europe CH only contains 70M arcs for Metis whereas the OSM-Europe CH contains 400M arcs for Metis.
However, this is due to the huge amount of low-degree vertices in the input. 
On the DIMACS graph the size increase compared to the number of input arcs is $70\,\text{M}/42\,\text{M}=1.67$ whereas for the OSM-based graph the size increase is only $400\,\text{M}/348\,\text{M}=1.15$.
This effect can be explained by considering what happens %
when contracting a graph consisting of a single path.
In the input graph every vertex, except the endpoints, has 2 outgoing arcs, one in each direction.
As long as the endpoints are contracted last, every vertex, except the endpoints, in the resulting CH search graph also have degree 2.
There is thus no size increase.
As the OSM-based graph has many degree-2 vertices, this effect dominates and explains the comparatively small size increase.
The search space sizes are nearly identically. 
For example the KaHIP search space contains 117K arcs for the DIMACS Europe graph, whereas it contains 119K arcs for the OSM Europe graph.
This effect is explained by the fact, that both data sources correspond to almost the same geographical region.
The mountains and rivers are thus in the same locations and the number of roads through these geographic obstacles are the same in both graphs, \ie, both graphs have very similar recursive separators.
The small size increase is explained by the fact that the OSM-based graph also includes Eastern Europe.
The customization times are reported in Table~\ref{tab:osm-custom}.
As the OSM-based graph has more arcs, the customization times are higher on that graph.
On the DIMACS Europe graph 0.61s are needed whereas 1.7s are needed on OSM Europe for the KaHIP order and 16 threads, which is a surprisingly small gap considering the differences in input sizes.
Eliminating the degree-2 vertices from the input should further narrow this gap.
As the search space sizes are very similar, it is not surprising that the query running times reported in Table~\ref{tab:osm-query} are nearly identical.

\subsection{Further DIMACS-Instances}

\begin{table}
\tbl{Instance sizes and experimental results for the additional DIMACS road graphs graphs. The instances are weighted by travel time. They are undirected, \ie, no one-way streets exist and the weight of an arc corresponds to its backward arc's weight. The "Uni" numbers exploit this and have one weight per CH arc, whereas the "Bi" numbers do not and have two weights per CH arc. We report the number of vertices, directed arcs after removing multi-arcs, arcs in the CH search graph, the time needed to do a full non-amortized customization with 1 and 16 threads, the average running time of Dijkstra's algorithm with stopping criterion, and the average running time of an elimination-tree distance query. The order were computed with KaHIP. We averaged over $10\,000$ queries where $s$ and $t$ were picked randomly at uniform. We only do a basic customization and no perfect customization.\label{tab:dimacs-extra}}{

\begin{tabular}{lrrrrrrrrrr}
\toprule
 &  &  &  & \multicolumn{4}{c}{Customization  {[}ms{]}} &  & \multicolumn{2}{c}{CH Query}\tabularnewline
 & Vertices & Arcs & CH Arcs & \multicolumn{2}{c}{1 thread} & \multicolumn{2}{c}{16 threads} & Dijkstra  &  \multicolumn{2}{c}{{[}$\mu$s{]}} \tabularnewline
 \cmidrule(lr){5-8} \cmidrule(lr){10-11}
Graph & {[}$\cdot10^{3}${]} & {[}$\cdot10^{3}${]} & {[}$\cdot10^{3}${]} & Uni & Bi & Uni & Bi & {[}$\mu$s{]} & Uni & Bi\tabularnewline
\midrule
NY & 264 & 730 & 1~547 & 46 & 52 & 11 & 12 & 16~303 & 34 & 34\tabularnewline
BAY & 321 & 795 & 1~334 & 29 & 39 & 7 & 8 & 17~964 & 20 & 20\tabularnewline
COL & 436 & 1~042 & 1~692 & 40 & 51 & 12 & 12 & 25~505 & 35 & 41\tabularnewline
FLA & 1~070 & 2~688 & 4~239 & 93 & 117 & 25 & 32 & 63~497 & 30 & 26\tabularnewline
NW & 1~208 & 2~821 & 4~266 & 88 & 110 & 24 & 31 & 73~045 & 27 & 27\tabularnewline
NE & 1~524 & 3~868 & 6~871 & 195 & 255 & 54 & 57 & 96~628 & 68 & 64\tabularnewline
CAL & 1~891 & 4~630 & 7~587 & 195 & 250 & 61 & 64 & 114~047 & 43 & 43\tabularnewline
LKS & 2~758 & 6~795 & 12~829 & 478 & 646 & 75 & 87 & 175~084 & 138 & 149\tabularnewline
E & 3~599 & 8~708 & 14~169 & 395 & 514 & 85 & 96 & 233~511 & 86 & 88\tabularnewline
W & 6~262 & 15~120 & 24~115 & 682 & 894 & 121 & 132 & 425~244 & 82 & 84\tabularnewline
CTR & 14~082 & 33~867 & 57~222 & 2~656 & 3~592 & 392 & 416 & 1~050~314 & 276 & 285\tabularnewline
USA & 23~947 & 57~709 & 97~902 & 3~617 & 7~184 & 698 & 979 & 1~883~053 & 264 & 286\tabularnewline
\bottomrule
\end{tabular}

}
\end{table}

During the DIMACS challenge on shortest path~\cite{dgj-spndi-09} several benchmark instances were made available. 
Among them is the Europe instance used throughout our in-depth experiments in previous chapters.
Besides this instance, also a set of graphs representing the road network of the USA was published.
In Table~\ref{tab:dimacs-extra} we report experiments for these additional DIMACS road graphs.
Other than the DIMACS-Europe instance, these USA instances originate from the U.S.~Census Bureau. 
Note that the USA instances have some known data quality issues: The graphs are generally undirected~(no one-way streets) and highways are sometimes not connected at state borders.
The DIMACS-Europe comes from another data source and does not have these limitations. 
This is the reason why we focus on the Europe instance in the main part of our evaluation.

However, as the graphs are undirected we can evaluate the impact of using a single undirected metric has on customization running times compared to using two directed metrics~(as used on DIMACS Europe).
Experiments using a single metric are marked with ``Uni'' in the table, whereas the experiments with two metrics are marked with ``Bi''.
The query running times are very similar.
This is not surprising as the number of relaxed arcs does not depend on whether one or two weights are used.
For larger graphs there is a slightly larger difference in running times.
We believe that this is a cache effect.
As the ``Bi'' variant has twice as many weights, less arcs fit into the L3~cache. 
For the smaller graphs this effect does not occur because the higher CH levels occupy less memory than the cache's size and thus doubling the memory consumption is non-problematic.

The difference in customization times between the two variants is larger.
The number of enumerated triangles is the same, but twice as many instructions are executed per triangle.
We would thus expect a factor of~2 difference in the customization running times.
However, this factor is only observed on the largest instance.
On all smaller instances, the gap is significantly smaller.
Again, this is most likely the result of cache effects. 

\subsection{Further Game Instances}

\begin{table}
\tbl{Instance sizes and experimental results for the additional game-based graphs. We report the number of vertices, undirected edges, arcs in the CH search graph, the running time Metis needed to compute the order, the time needed to do a full non-amortized customization with 1 and 16 threads using an undirected weight, the average running time of Dijkstra's algorithm with stopping criterion and the average running time of an elimination-tree distance query. We averaged over $10\,000$ queries where $s$ and $t$ were picked randomly at uniform.}{

\begin{tabular}{lrrrrrrrr}
\toprule
 & Vertices & Edges & CH Arcs & Metis & \multicolumn{2}{r}{Customize {[}ms{]}} & Dijkstra  & CH Query\tabularnewline
Graph & {[}$\cdot10^{3}${]} & {[}$\cdot10^{3}${]} & {[}$\cdot10^{3}${]} & {[}s{]} & 1thr. & 16thr. & {[}$\mu$s{]} &  {[}$\mu$s{]}\tabularnewline
\midrule
16room\_005 & 231 & 838 & 2~959 & 1.196 & 359 & 41 & 10~913 & 15\tabularnewline
AcrosstheCape & 392 & 1~534 & 12~632 & 2.452 & 4~789 & 563 & 23~609 & 239\tabularnewline
blastedlands & 131 & 507 & 3~740 & 0.896 & 1~250 & 144 & 6~075 & 304\tabularnewline
maze512-4-3 & 209 & 686 & 1~641 & 0.996 & 138 & 35 & 7~810 & 6\tabularnewline
ost100d & 137 & 531 & 3~722 & 0.880 & 1~076 & 124 & 5~607 & 116\tabularnewline
random512-35-8 & 161 & 421 & 1~805 & 0.988 & 469 & 39 & 7~422 & 223\tabularnewline
random512-40-8 & 114 & 280 & 797 & 0.684 & 115 & 16 & 4~856 & 41\tabularnewline
\bottomrule
\end{tabular}
}
\end{table}

Besides our main game benchmark instance TheFrozenSea, the benchmark data set of~\cite{s-bgbp-12} contains a large variety of different game maps. 
To demonstrate that our technique also works on other game maps we ran our experiments on a selection of different graphs from the set.
``16room\_005'' is a synthetic map with many rooms in grid shape that are connected through small doors.
``AcrosstheCape'' is another Star Craft map that is sometimes used as benchmark instance.
``blastedlands'' originates from WarCraft 3 and is the largest map in that set in terms of vertices.
``maze512-4-3'' is a synthetic map that consists of a random maze with corridors that are 4 fields wide.
``ost100d'' is the largest map from the Dragon Age Origins map set.
``random512-35-8'' and ``random512-40-8'' are synthetic maps that contain random obstacles. 
The difference between them is the amount of space covered by obstacles.
The website\footnote{\url{http://www.movingai.com/benchmarks/}} from which the data originates includes pictures depicting each instance.

All experiments were run using a single undirected metric with 32bits per weight. The customization running times are non-amortized. We did not perform experiments with a perfect witness search.

All additional game-based instances have fewer vertices than TheFrozenSea. 
Further, the CH query is the slowest on TheFrozenSea with 316 $\mu$s.
Interestingly, a full customization is slower on AcrosstheCape than on TheFrozenSea by about a factor of~2.
This is most likely due to slight differences in the structures of the maps.
However, we believe that it is safe to conclude from the experiments that our technique works across a wide range of maps.

\section{Conclusions}
\label{sec:conclusion}

We have extended Contraction Hierarchies~(CH) to a three-phase customization approach and demonstrated in an extensive experimental evaluation that our Customizable Contraction Hierarchies approach is practicable and efficient not only on real world road graphs but also on game maps. 
We have proposed new algorithms that improve on the state-of-the-art for nearly all stages of the toolchain: Using our contraction graph data structure, a metric-independent CH can be constructed faster than with the established approach based on dynamic arrays.
We have shown that the customization phase is essentially a triangle enumeration algorithm.
We have provided two variants of the customization: The basic variant yields faster customization running times, while perfect customization and witness search computes CHs with a provable minimum number of shortcuts within seconds given a metric-independent vertex order.
We proposed an elimination-tree based query that unlike previous approaches is not based on Dijkstra's algorithm and thus does not use a priority queue.
This results in significantly lower overhead per visited arc, enabling faster queries.
Finally, our extensive experimental analysis shed some light onto the inner workings of Contraction Hierarchies. 

\subsection{Future Work}
\label{sec:future-work}

Good separators are the foundation of Customizable Contraction Hierarchies: Finding better separators directly improves both customization as well as query performance. For our purposes, the time required to compute good separators was of no primary concern~(we do it once per graph). Hence, \emph{our} nested dissection implementation \emph{based} on KaHIP~\cite{ss-tlagh-13} was not optimized for speed but rather to demonstrate that good separators exist.

For some applications this may be too slow. However, since we performed the experiments reported in this paper, significant improvements have been made in this domain: The works of~\cite{w-fsns-14} lay the foundations of a well-implemented KaHIP-based algorithm. \cite{ss-obsrn-15}~introduce a new and surprisingly simple road graph bisection algorithm called Inertial Flow. FlowCutter~\cite{hs-gbpo-15}, available as preprint, computes the best metric-independent contraction orders we have observed so far, much faster than the nested dissection implementation presented in this work. %
This results in decreased query and customization times and reduced memory consumption for Customizable Contraction Hierarchies.

Our experiments suggest, that even metric-dependent Contraction Hierarchies implicitly exploit the existence of small graph cuts.
However, this does not seem to be the only exploited feature as metric-independent orders behave differently when it comes to details:
For example, the stall-on-demand optimization only works with metric-dependent orders, for both travel time as well as distance metric, but not for the metric-independent orders.
Further investigations into this effect might yield additional valuable insights.
A good starting point could be the theoretical works of~\cite{afgw-hdspp-10}.

Further investigation into algorithms explicitly exploiting treewidth~\cite{cz-spdst-00,pwk-caspl-12} seems promising. Also, determining the precise treewidth of road networks could prove useful. 

In practice, route planning services consider several additional real-world constraints beyond the scope of this paper. These include, \eg, turn costs and restrictions, historic traffic data for rush hours, and range constraints due to limited electric vehicle batteries.
For turns, we assume that Customizable Contraction Hierarchies is well applicable as the size of graph cuts cannot grow arbitrarily when turn-expanding the road graph. We are interested in further in-depth experimental analysis of aforementioned scenarios.

\ack{We would like to thank Ignaz Rutter and Tim Zeitz for very inspiring conversations.}

\bibliographystyle{plain}

\begin{thebibliography}{10}

\bibitem{adfgw-h-13}
Ittai Abraham, Daniel Delling, Amos Fiat, Andrew~V. Goldberg, and Renato~F.
  Werneck.
\newblock Highway dimension and provably efficient shortest path algorithms.
\newblock 2013.

\bibitem{adgw-hhlsp-12}
Ittai Abraham, Daniel Delling, Andrew~V. Goldberg, and Renato~F. Werneck.
\newblock Hierarchical hub labelings for shortest paths.
\newblock In {\em Proceedings of the 20th Annual European Symposium on
  Algorithms (ESA'12)}, volume 7501 of {\em Lecture Notes in Computer Science},
  pages 24--35. Springer, 2012.

\bibitem{afgw-hdspp-10}
Ittai Abraham, Amos Fiat, Andrew~V. Goldberg, and Renato~F. Werneck.
\newblock Highway dimension, shortest paths, and provably efficient algorithms.
\newblock In {\em Proceedings of the 21st Annual {ACM--SIAM} Symposium on
  Discrete Algorithms (SODA'10)}, pages 782--793. SIAM, 2010.

\bibitem{bdgmpsww-rptn-14}
Hannah Bast, Daniel Delling, Andrew~V. Goldberg, Matthias
  {M{\"u}ller--Hannemann}, Thomas Pajor, Peter Sanders, Dorothea Wagner, and
  Renato~F. Werneck.
\newblock Route planning in transportation networks.
\newblock Technical Report abs/1504.05140, ArXiv e-prints, 2015.

\bibitem{bckkw-psuth-10}
Reinhard Bauer, Tobias Columbus, Bastian Katz, Marcus Krug, and Dorothea
  Wagner.
\newblock Preprocessing speed-up techniques is hard.
\newblock In {\em Proceedings of the 7th Conference on Algorithms and
  Complexity (CIAC'10)}, volume 6078 of {\em Lecture Notes in Computer
  Science}, pages 359--370. Springer, 2010.

\bibitem{bcrw-sssch-13}
Reinhard Bauer, Tobias Columbus, Ignaz Rutter, and Dorothea Wagner.
\newblock Search-space size in contraction hierarchies.
\newblock In {\em Proceedings of the 40th International Colloquium on Automata,
  Languages, and Programming (ICALP'13)}, volume 7965 of {\em Lecture Notes in
  Computer Science}, pages 93--104. Springer, 2013.

\bibitem{bddsw-tspca-12}
Reinhard Bauer, Gianlorenzo D'Angelo, Daniel Delling, Andrea Schumm, and
  Dorothea Wagner.
\newblock The shortcut problem -- complexity and algorithms.
\newblock {\em Journal of Graph Algorithms and Applications}, 16(2):447--481,
  2012.

\bibitem{b-atgt-93}
Hans~L. Bodlaender.
\newblock A tourist guide through treewidth.
\newblock {\em j-acta-cybernet}, 11:1--21, 1993.

\bibitem{b-tsa-07}
Hans~L. Bodlaender.
\newblock Treewidth: Structure and algorithms.
\newblock In {\em Proceedings of the 14th International Colloquium on
  Structural Information and Communication Complexity}, volume 4474 of {\em
  Lecture Notes in Computer Science}, pages 11--25. Springer, 2007.

\bibitem{bk-tiu-10}
Hans~L. Bodlaender and Arie M. C.~A. Koster.
\newblock Treewidth computations i. upper bounds.
\newblock {\em Information and Computation}, 208(3):259--275, 2010.

\bibitem{cz-spdst-00}
Soma Chaudhuri and Christos Zaroliagis.
\newblock Shortest paths in digraphs of small treewidth. part i: Sequential
  algorithms.
\newblock {\em Algorithmica}, 2000.

\bibitem{dgpw-crp-11}
Daniel Delling, Andrew~V. Goldberg, Thomas Pajor, and Renato~F. Werneck.
\newblock Customizable route planning.
\newblock In {\em Proceedings of the 10th International Symposium on
  Experimental Algorithms (SEA'11)}, volume 6630 of {\em Lecture Notes in
  Computer Science}, pages 376--387. Springer, 2011.

\bibitem{dgpw-crprn-13}
Daniel Delling, Andrew~V. Goldberg, Thomas Pajor, and Renato~F. Werneck.
\newblock Customizable route planning in road networks.
\newblock {\em Transportation Science}, 2015.

\bibitem{dgrw-gpnc-11}
Daniel Delling, Andrew~V. Goldberg, Ilya Razenshteyn, and Renato~F. Werneck.
\newblock Graph partitioning with natural cuts.
\newblock In {\em 25th International Parallel and Distributed Processing
  Symposium (IPDPS'11)}, pages 1135--1146. IEEE Computer Society, 2011.

\bibitem{dgrw-ecbbg-12}
Daniel Delling, Andrew~V. Goldberg, Ilya Razenshteyn, and Renato~F. Werneck.
\newblock Exact combinatorial branch-and-bound for graph bisection.
\newblock In {\em Proceedings of the 14th Meeting on Algorithm Engineering and
  Experiments (ALENEX'12)}, pages 30--44. SIAM, 2012.

\bibitem{dw-fcrn-13}
Daniel Delling and Renato~F. Werneck.
\newblock Faster customization of road networks.
\newblock In {\em Proceedings of the 12th International Symposium on
  Experimental Algorithms (SEA'13)}, volume 7933 of {\em Lecture Notes in
  Computer Science}, pages 30--42. Springer, 2013.

\bibitem{dgj-spndi-09}
Camil Demetrescu, Andrew~V. Goldberg, and David~S. Johnson, editors.
\newblock {\em The Shortest Path Problem: Ninth DIMACS Implementation
  Challenge}, volume~74 of {\em DIMACS Book}.
\newblock American Mathematical Society, 2009.

\bibitem{d-ntpcg-59}
Edsger~W. Dijkstra.
\newblock A note on two problems in connexion with graphs.
\newblock {\em Numerische Mathematik}, 1:269--271, 1959.

\bibitem{fg-imig-65}
Delbert~R. Fulkerson and O.~A. Gross.
\newblock Incidence matrices and interval graphs.
\newblock {\em Pacific Journal of Mathematics}, 15(3):835--855, 1965.

\bibitem{gssd-chfsh-08}
Robert Geisberger, Peter Sanders, Dominik Schultes, and Daniel Delling.
\newblock Contraction hierarchies: Faster and simpler hierarchical routing in
  road networks.
\newblock In {\em Proceedings of the 7th Workshop on Experimental Algorithms
  (WEA'08)}, volume 5038 of {\em Lecture Notes in Computer Science}, pages
  319--333. Springer, June 2008.

\bibitem{gssv-erlrn-12}
Robert Geisberger, Peter Sanders, Dominik Schultes, and Christian Vetter.
\newblock Exact routing in large road networks using contraction hierarchies.
\newblock {\em Transportation Science}, 46(3):388--404, August 2012.

\bibitem{g-ndrfe-73}
Alan George.
\newblock Nested dissection of a regular finite element mesh.
\newblock {\em SIAM Journal on Numerical Analysis}, 10(2):345--363, 1973.

\bibitem{gl-aqgms-78}
Alan George and Joseph~W. Liu.
\newblock A quotient graph model for symmetric factorization.
\newblock In {\em Sparse Matrix Proceedings}. SIAM, 1978.

\bibitem{gt-t-86}
John~R. Gilbert and Robert Tarjan.
\newblock The analysis of a nested dissection algorithm.
\newblock {\em Numerische Mathematik}, 1986.

\bibitem{hs-gbpo-15}
Michael Hamann and Ben Strasser.
\newblock Graph bisection with pareto-optimization.
\newblock Technical report, arXiv, 2015.

\bibitem{hsw-emlog-08}
Martin Holzer, Frank Schulz, and Dorothea Wagner.
\newblock Engineering multilevel overlay graphs for shortest-path queries.
\newblock {\em ACM Journal of Experimental Algorithmics}, 13(2.5):1--26,
  December 2008.

\bibitem{kst-tpcpc-99}
Haim Kaplan, Ron Shamir, and Robert Tarjan.
\newblock Tractability of parameterized completion problems on chordal,
  strongly chordal, and proper interval graphs.
\newblock {\em SIAM Journal on Computing}, 1999.

\bibitem{kk-mspig-99}
George Karypis and Vipin Kumar.
\newblock A fast and high quality multilevel scheme for partitioning irregular
  graphs.
\newblock {\em SIAM Journal on Scientific Computing}, 20(1):359--392, 1999.

\bibitem{lrt-gnd-79}
Richard~J. Lipton, Donald~J. Rose, and Robert Tarjan.
\newblock Generalized nested dissection.
\newblock {\em SIAM Journal on Numerical Analysis}, 16(2):346--358, April 1979.

\bibitem{m-o-12}
Nikola Milosavljevi{\'c}.
\newblock On optimal preprocessing for contraction hierarchies.
\newblock In {\em Proceedings of the 5th ACM SIGSPATIAL International Workshop
  on Computational Transportation Science}, pages 33--38. ACM Press, 2012.

\bibitem{pwk-caspl-12}
L\'{e}on Planken, Mathijs de~Weerdt, and Roman van Krogt.
\newblock Computing all-pairs shortest paths by leveraging low treewidth.
\newblock {\em Journal of Artificial Intelligence Research}, 2012.

\bibitem{p-t-88}
Alex Pothen.
\newblock The complexity of optimal elimination trees.
\newblock Technical report, Pennsylvania State University, 1988.

\bibitem{ss-ehh-12}
Peter Sanders and Dominik Schultes.
\newblock Engineering highway hierarchies.
\newblock {\em ACM Journal of Experimental Algorithmics}, 17(1):1--40, 2012.

\bibitem{ss-tlagh-13}
Peter Sanders and Christian Schulz.
\newblock Think locally, act globally: Highly balanced graph partitioning.
\newblock In {\em Proceedings of the 12th International Symposium on
  Experimental Algorithms (SEA'13)}, volume 7933 of {\em Lecture Notes in
  Computer Science}, pages 164--175. Springer, 2013.

\bibitem{ss-obsrn-15}
Aaron Schild and Christian Sommer.
\newblock On balanced separators in road networks.
\newblock In {\em Proceedings of the 14th International Symposium on
  Experimental Algorithms (SEA'15)}, Lecture Notes in Computer Science.
  Springer, 2015.

\bibitem{sww-daola-00}
Frank Schulz, Dorothea Wagner, and Karsten Weihe.
\newblock {D}ijkstra's algorithm on-line: An empirical case study from public
  railroad transport.
\newblock {\em ACM Journal of Experimental Algorithmics}, 5(12):1--23, 2000.

\bibitem{s-chgg-13}
Sabine Storandt.
\newblock Contraction hierarchies on grid graphs.
\newblock In {\em Proceedings of the 36rd Annual German Conference on Advances
  in Artificial Intelligence}, Lecture Notes in Computer Science. Springer,
  2013.

\bibitem{s-bgbp-12}
Nathan Sturtevant.
\newblock Benchmarks for grid-based pathfinding.
\newblock {\em Transactions on Computational Intelligence and AI in Games},
  2012.

\bibitem{w-fsns-14}
Michael Wegner.
\newblock Finding small node separators.
\newblock Bachelor thesis, Karlsruhe Institute of Technology, October 2014.

\bibitem{w-tedi-10}
Fang Wei.
\newblock Tedi: efficient shortest path query answering on graphs.
\newblock In {\em Proceedings of the 2010 ACM SIGMOD International Conference
  on Management of Data (SIGMOD'10)}. ACM Press, 2010.

\bibitem{y-cnp-81}
Mihalis Yannakakis.
\newblock Computing the minimum fill-in is np-complete.
\newblock {\em SIAM Journal on Algebraic and Discrete Methods}, 1981.

\bibitem{z-wchw-13}
Tim Zeitz.
\newblock Weak contraction hierarchies work!
\newblock Bachelor thesis, Karlsruhe Institute of Technology, 2013.

\end{thebibliography}

\end{document}